\documentclass[a4paper,UKenglish,cleveref, autoref, thm-restate]{lipics-v2021}

\hideLIPIcs  

\usepackage{color,xcolor}
\usepackage{tikz}
\usetikzlibrary{arrows.meta}
\usetikzlibrary{patterns,patterns.meta}
\usetikzlibrary{fit}
\usepackage{tikzmacros}
\usepackage[keep-defaults,Tol]{colorblind}
\usepackage{newunicodechar}
\usepackage{hyphenat}
\usepackage{cite}
\usepackage[prependcaption,colorinlistoftodos]{todonotes}

\usepackage{appendix-tools}
\usepackage[fullversion]{flags}
\usepackage{colorpalette,macros,computational-problems}

\definecolor{marceloColor}{RGB}{26, 192, 19}

\title{The directed temporal exploration problem}

\author{Marcelo {Garlet Milani}}{National Institute of Informatics, Tokyo, Japan \and \url{https://mgarletmilani.com} }{research@mgarletmilani.com}{https://orcid.org/0000-0001-8398-4751}{}
\ifanonymous\else
\author{Lucas {Picasarri-Arrieta}}{National Institute of Informatics, Tokyo, Japan  \and \url{https://lucaspicasarri.github.io} }{lpicasarr@nii.ac.jp}{https://orcid.org/0000-0003-0414-8136}{}

\author{Chaoliang {Tang}}{Shanghai Center for Mathematical Sciences, Fudan University, Shanghai, China  \and \url{https://chaoliangtang.github.io/minimalistic/} }{cltang22@m.fudan.edu.cn}{https://orcid.org/0009-0005-4708-762X}{}

\author{Hehui {Wu}}{Shanghai Center for Mathematical Sciences, Fudan University, Shanghai, China}{hhwu@fudan.edu.cn}{https://orcid.org/0000-0001-9849-9097}{}
\fi
\authorrunning{M. G. Milani, L. Picasarri-Arrieta, C. Tang, and H. Wu} 
\Copyright{Marcelo Garlet Milani, Lucas Picasarri-Arrieta, Chaoliang Tang, Hehui Wu} 



\funding{Research of Marcelo Garlet Milani and Lucas Picasarri-Arrieta was supported by JSPS KAKENHI JP20A402 and 22H05001, and by JST ASPIRE JPMJAP2302. Research of Chaoliang Tang and Hehui Wu was supported by  National Natural Science Foundation of China grant 12371343 and grant 11931006.}

\acknowledgements{We thank Pierre Charbit and Samuel Coulomb for discussions on the proof of \cref{lemma:directed_king}. We thank Leszek Gasieniec for suggesting the special case of temporal semicomplete digraphs.}
\nolinenumbers 
\begin{document}

\maketitle

\begin{abstract}
	We study the temporal exploration problem on temporal digraphs.
  We prove that a lifetime of \(\BigO{n^2}\) suffices to guarantee the existence of a temporal exploration on
  always-unilateral temporal digraphs.
  We complement this with a \(\Omega(n^2)\) lower bound,
  even in the case where each snapshot has maximum undirected degree 2;
  for always-strong temporal digraphs,
  the lower bound still holds even if the maximum undirected degree is 3.
  This stands in stark contrast with the undirected setting.

  For the large minimum degree setting,
  we show that a lifetime of \(4n/3 - 1\) is sufficient and necessary
  for guaranteeing the existence of a temporal exploration on temporal digraphs
  where each snapshot is semicomplete.
  For always-strong temporal digraphs where each snapshot has
  minimum undirected degree at least \(n - c - 1\),
  we prove that a lifetime of \(\BigO{cn}\) guarantees the existence of a temporal exploration,
  and we also prove that this is asymptotically tight.

  From a computational perspective, our results for temporal semicomplete digraphs
  also yield a polynomial-time, factor-\(4/3\) algorithm for deciding if a temporal semicomplete digraph
  admits a temporal exploration within the first \(\ell\) snapshots.
  We complement this showing that no polynomial-time, factor-\((4/3 - \epsilon)\) approximation algorithm exists,
  even if every snapshot is a tournament,
  unless \(\P/=\NP/\).
\end{abstract}

\newpage
\tableofcontents

\newcommand{\pathmaxdegreethreeref}[1]{  \hyperlink{thm:always_path_maxdegree_3:body}{\cref*{thm:always_path_maxdegree_3}}}
\newcommand{\semicompleteexplorationref}[1]{  \hyperlink{theorem:tournament-exploration-upper-bound:body}{\cref*{theorem:tournament-exploration-upper-bound}}}
\newcommand{\semicompleteexplorationlowerboundref}[1]{  \hyperlink{statement:tournament-exploration-lower-bound:body}{\cref*{statement:tournament-exploration-lower-bound}}}
\newcommand{\largeminimumdegreeupperboundref}[1]{  \hyperlink{statement:degree-n-c-upper-bound-strong:body}{\cref*{statement:degree-n-c-upper-bound-strong}}}
\newcommand{\largeminimumdegreelowerboundref}[1]{  \hyperlink{statement:degree-n-c-lower-bound:body}{\cref*{statement:degree-n-c-lower-bound}}}
\newcommand{\explorationsemicompletenphardref}[1]{  \hyperlink{thm:np_hardness:body}{\cref*{thm:np_hardness}}}
\newcommand{\explorationsemicompletenapproximationhardref}[1]{  \hyperlink{statement:semicomplete approximation hardness:body}{\cref*{statement:semicomplete approximation hardness}}}

\newcommand{\MaI}[1]{}
\theoremstyle{plain}
\newtheorem{problem}[theorem]{Problem}

\section{Introduction}

Temporal graphs are an extension of the concept of graphs allowing for modifications on the graph structure over time.
This additional expressiveness compared to static graphs
has found applications in biology and medicine \cite{HCGD23} and,
in the case of temporal digraphs, in structural digraph theory \cite{hkmm24cows}.
There has been a lot of research on different topics related to temporal graphs
\cite{bumpus2023edge,fluschnik2020temporal,hamm2022complexity,marino2021konigsberg,marino2023eulerian,aaron2014dmvp,akrida2021temporal,erlebach2020non,ilcinkas2018exploration,balev2025brief,BGMR2025ietg,EHK2021}.
See \cite{casteigts2012time,michail2016introduction,ES2022} for surveys on temporal graphs and \cite{FMNRZ20} for a survey on temporal applications of treewidth.

A temporal digraph \(D\) is a sequence of digraphs \(D_{1}, D_{2}, \ldots, D_{t}\)
on the same vertex set but on potentially different arc sets, called \emph{snapshots}.
The number of snapshots in \(D\) is the \emph{lifetime} of \(D\) and is given by \(\Lifetime{D}\).
This work focus on \emph{temporal walks} which take
at most one arc per snapshot and allow staying on a vertex for multiple time steps.
In this setting, we study the \emph{temporal exploration problem},
which asks whether a temporal digraph admits a temporal walk spanning all vertices.
In particular, we want to determine sufficient conditions guaranteeing the existence of a temporal exploration.
Since it is clearly not possible to find a temporal exploration if certain vertices pairs of vertices can never reach each other,
we impose some connectivity requirements on the snapshots.
More precisely, we require of every snapshot that, for every pair of vertices \(u,v\),
\(u\) can reach \(v\) \emph{or} \(v\) can reach \(u\).
Such digraphs are called \emph{unilateral}, and
a temporal digraph where every snapshot is unilateral is said to be \emph{always-unilateral}.
Similarly, in an always-strong temporal digraph, every snapshot is strongly connected,
so both \(u\) can reach \(v\) \emph{and} \(v\) can reach \(u\).

We start with a question posed by
\cite{hkmm24cows}, who proved, as an intermediate step for their main results,
that every always-unilateral temporal digraph with at least \(\BigO{k^2n^{kn + 2}}\)
snapshots admits a temporal walk spanning a given set \(S\) of \(k\) vertices.
They asked whether their bound can be improved to a polynomial.
We answer this affirmatively with the following Theorem (\cref{sec:always-unilateral exploration}).

\begin{restatable}{theorem}{unilateralExploration}
	\label{statement:always-unilateral walk covering S}
	\label{statement:unilateral layers -> temporal exploration}
	Let \(D\) be an always-unilateral temporal digraph on \(n\) vertices and
	let \(S ⊆ \V{D}\) be a set of size \(k\).
	If \(\Lifetime{D} \geq \sum_{i=2}^k\bound{statement:unilateral layers -> temporal root}{t}{n,i} ∈ \BigO{kn + k^2}\), then there is a temporal walk \(W\) in \(D\)
	such that 
	\(S ⊆ \V{W}\).
\end{restatable}

The case where \(S\) is the whole vertex set corresponds to the temporal exploration problem, 
where \(\BigO{n^2}\) snapshots suffice.
We complement this result with a lower-bound,
showing that the bound above is asymptotically tight.
Below, \(\UnderlyingGraph{D}\) is the underlying undirected graph of the union of all snapshots of \(D\) and \(\Delta(\UnderlyingGraph{D})\) is its maximum degree.
\begin{restatable}{theorem}{pathBoundedDegreeLowerBound}
    \label{thm:always_path_maxdegree_3}
    For every positive integer $n'$, there exists an always-unilateral temporal digraph $D$ of order $n\geq n'$ such that each snapshot of $D$ is a directed path, $\Delta(\UnderlyingGraph{D}) \leq 3$, $\Lifetime{D}\geq \frac{1}{128} n^2$, and $D$ does not admit any temporal exploration. 
\end{restatable}

The conditions imposed on the maximum degree above
reveal a strong discrepancy between the exploration of temporal digraphs and of temporal graphs.
In~\cite{ErlebachS18} it is proved that \(\BigO{(n^2 \log d) / \log n}\)
snapshots suffice to explore always-connected temporal graphs where
each snapshot has degree at most \(d\).
This was improved by \cite{BGMR2025ietg}
to \(\BigO{n^{3/2} \sqrt{d \log n}}\).
\pathmaxdegreethreeref{} implies that similar results are not possible for always-unilateral temporal digraphs.
The current best lower bound for temporal graphs is \(\Omega(dn)\) and holds for every \(d\) \cite{EHK2021}.
For the case where each snapshot is a path and the underlying graph (that is, the union of all edge sets) is planar with maximum degree 3,
the best lower bound known is \(\Omega(n \log n)\) \cite{AdamsonGMZ22,EHK2021}.

In order to understand which temporal digraphs can be explored within a linear number of snapshots,
we look into the other side of the spectrum where the snapshots are very dense.
We start from the densest case where each snapshot is \emph{semicomplete}.
A digraph is said to be semicomplete if
for every pair of distinct vertices \(u,v\)
the arc \((u,v)\) or the arc \((v,u)\) exist (or both).
Note that, unlike undirected graphs, where there is only one complete graph,
there can be exponentially many non-isomorphic semicomplete digraphs.
Already the class of {\it tournaments}, where exactly one of the arcs \((u,v),(v,u)\) exists, is rather rich, as shown for instance by the results in~\cite{moon1968topics,APSComb21,NSSJCTB173}.
Furthermore, while temporal graphs whose snapshots are complete graphs
can be trivially explored in \(n - 1\) snapshots by taking any ordering of the vertices,
it is not true that the vertices of any semicomplete temporal digraph
can be explored in any order.
We can, however, prove that they always admit a temporal exploration
within a linear number of snapshots.
We also show that the bounds we obtain are tight (\cref{sec:always-semicomplete exploration}).
\begin{restatable}{theorem}{semicompleteExploration}	\label{theorem:tournament-exploration-upper-bound}	Let \(D\) be a semicomplete temporal digraph of order $n\geq 2$. If $\Lifetime{D}\geq \lfloor \frac{4}{3}n - 1\rfloor$, then \(D\) admits a temporal exploration.
\end{restatable}
\begin{restatable}{theorem}{semicompleteExplorationLowerBound}
	\label{statement:tournament-exploration-lower-bound}
	For every integer \(n\geq 2\) there exists a temporal tournament \(D\) of order \(n\) with $\Lifetime{D}=\lfloor \frac{4}{3}n-2\rfloor$ that does not admit any temporal exploration.
\end{restatable}

From an algorithmic perspective,
we consider the \TemporalExploration/ problem,
where we are given a temporal digraph \(D\) and
need to decide whether it admits a temporal exploration or not, and
also its optimization version, \TemporalExplorationFastest/, where we need to find the smallest \(t'\)
such that \(D\) admits a temporal exploration within the first \(t'\) snapshots
(to simplify the technical details, \(D\) is guaranteed to have a temporal exploration in the optimization setting).
These problems remain computationally hard even in very restricted settings,
with no factor \(2 - \epsilon\) approximation algorithm for \TemporalExplorationFastest/, even on always-strong temporal digraphs
\cite{michailTCS634}, and we provide further hardness results.

\begin{restatable}{theorem}{explorationSemicompleteNPHard}
	\label{thm:np_hardness}
	\TemporalExploration/ remains \NP/-hard even if the input $D$ is a temporal tournament and $\Lifetime{D} =|V(D)|-1$.
\end{restatable}
\begin{restatable}{theorem}{explorationSemicompleteApproximationHard}
    \label{statement:semicomplete approximation hardness}
    For every $\epsilon>0$, \TemporalExplorationFastest/ cannot be approximated within a factor $(\frac{4}{3}-\epsilon)$, even when restrained to temporal tournaments, unless \P/ = \NP/.
\end{restatable}

Note that \cref{theorem:tournament-exploration-upper-bound} provides a simple
approximation algorithm for \TemporalExplorationFastest/:
return the minimum between $\Lifetime{D}$ and $\Floor{\frac{4}{3}|V(D)|-1}$.
As every temporal exploration requires at least \(\Abs{\V{D}} - 1\) many snapshots, and \(D\) is guaranteed to admit some temporal exploration,
by \cref{theorem:tournament-exploration-upper-bound}
this simple algorithm outputs a solution $t\in [\ell(D)]$ such that
\(
    t\leq \left\lfloor \tfrac{4}{3}t_{OPT} +\tfrac{1}{3}\right\rfloor.
\)

Back to the structural standpoint,
we generalize \cref{theorem:tournament-exploration-upper-bound,statement:tournament-exploration-lower-bound} by considering temporal digraphs
where the snapshots have large minimum degree.
To simplify notation, given an undirected graph \(G\),
we define \(\MindegComplement{G} \coloneqq \Abs{\V{G}} - \Mindeg{G} - 1\).
For a digraph \(D\), define \(\MindegComplement{D} \coloneqq \MindegComplement{\UnderlyingGraph{D}}\),
and for a temporal digraph \(D' = (D_i)_{i ∈ [t]}\), we define
\(\MindegComplement{D'} \coloneqq \max_{i ∈ [t]} \MindegComplement{D_i}\).
Note that for every \(v ∈ \V{G}\) there are at most \(\MindegComplement{G}\)
distinct vertices of \(\V{G} ∖ \{v\}\)
which do not have an edge to \(v\).
In particular, \(\MindegComplement{D'} = 0\)
if and only if every \(D_i\) is semicomplete.
The following asymptotically generalizes \cref{theorem:tournament-exploration-upper-bound,statement:tournament-exploration-lower-bound} in the case where all snapshots are strongly connected.
We leave as an open question whether the upper bound can be further generalized to always-unilateral temporal digraphs (\cref{sec:large minimum degree}).
\begin{restatable}{theorem}{largeMinimumDegree}
	\label{statement:degree-n-c-upper-bound-strong}
    For every integer $c$, every always-strong temporal digraph \(D\) on \(n\) vertices
	with \(\MindegComplement{D} \leq c\)
	and \(\Lifetime{D} \geq \bound{statement:degree-n-c-upper-bound-strong}{t}{n,c} ∈ \BigO{cn}\) admits a temporal exploration.
\end{restatable}
\begin{restatable}{theorem}{largeMinimumDegreeLowerBound}
	\label{statement:degree-n-c-lower-bound}
	For every \(c \geq 1\) and every \(n \geq c+1\) there is 
	an always-strong temporal digraph \(D\) on \(n\) vertices
	where	\(\MindegComplement{D} \leq c\), \(\Lifetime{D} = c(n - c - 2)\), and
	\(D\) does not admit a temporal exploration.
\end{restatable}

To see why temporal digraphs of large minimum degree may arise naturally,
consider the following construction.
Given an always-unilateral temporal digraph \(D\) and a number \(k\),
group the snapshots of \(D\) \(k\)-by-\(k\) as
\(D = (D_1^1, D_2^1, \ldots, D^1_k, D^2_1, \ldots, D^t_k)\) and
let \(H_i\) be the digraph with same vertex set as \(D\) and
with an arc \((u,v)\) if \(u\) can temporally reach \(v\)
within the snapshots \(D^i_{1}, D^i_{2}, \ldots, D^i_{k}\) and
let \(H = (H_{1}, H_{2}, \ldots, H_{t})\).
Clearly, a temporal exploration using \(t'\) snapshots of \(H\)
induces a temporal exploration of \(D\) with at most \(kt'\) snapshots.
Moreover, \(H\) has minimum degree at least \(k\),
but potentially much more.
This implies that, even if \(D\) has bounded maximum degree
we might still be able to use \cref{statement:degree-n-c-upper-bound-strong}
in order to obtain an exploration for \(D\)
if the temporal digraph \(H\) constructed above has large minimum degree.
Even in the undirected setting, 
it is not clear that 
the parameter maximum degree offers the same generalizability.

Our proof of \cref{statement:degree-n-c-upper-bound-strong}
relies on the \emph{probabilistic method}, in particular,
on the famous \emph{Lovász Local Lemma} \cite{erdosIFS10}.
Toward this end, we define an auxiliary structure,
called a \emph{travel plan},
which describes a whole family of temporal walks that can be constructed in a temporal digraph.
Our techniques separate the task of finding a good ordering of the vertices for a temporal exploration from the task of finding vertices which can temporally reach many unexplored vertices.
As the definition of travel plan is independent from the degree conditions we consider,
we believe our methods can be applied to other settings by
reducing the core of the problem to finding many travel plans.

Additionally, to the best of our knowledge,
there are no known counterparts of \cref{statement:degree-n-c-lower-bound,statement:degree-n-c-upper-bound-strong} for the undirected setting. While \cref{statement:degree-n-c-upper-bound-strong} immediately holds for temporal graphs,
\cref{statement:degree-n-c-lower-bound} does not,
as that would contradict previous subquadractic results on temporal graphs of bounded maximum degree due to \cite{EHK2021,BGMR2025ietg,ErlebachKLSS19}.
Moreover, if \(\MindegComplement{G} < n/2\),
then every vertex can reach every other vertex within two time steps,
and so \(2(n - 1)\) snapshots already guarantee a temporal exploration in this case.
The cases where \(\MindegComplement{G} \geq n/2\) are, however, open.
In \cref{sec:future work}, we discuss directions for further research.

\section{Preliminaries}
\label{sec:preliminaries}

Let $a \leq b$ be a integers. The set of integers $\{a,\dots,b\}$ is denoted by $[a,b]$,
and \([1,b]\) is also denoted by \([b]\).
For real numbers \(x,y\), we write \(\log_x y\) for the logarithm of \(y\) to base \(x\) and
\(\log y\) for \(\log_2 y\).

Our notation on digraphs follows~\cite{BG2009}. In this paper, digraphs may contain two arcs in opposite directions between the same pair of vertices, but no parallel arcs. A digraph without opposite arcs is called an \emph{oriented graphs}.

Let $D$ be a digraph. The vertex set of $D$ is denoted by $\V{D}$, and its arc set by $A(D)$.
The \emph{underlying graph} of $D$, denoted by $\UnderlyingGraph{D}$, is the undirected graph with vertex set $V(D)$ containing all edges $\{u,v\}$ such that $(u,v)\in A(D)$ or $(v,u)\in A(D)$.
A digraph is \emph{semicomplete} if its underlying graph is complete.
A \emph{tournament} is a digraph with exactly one arc between every pair of vertices.

Let $u\in V(D)$. An \emph{out-neighbor} (respectively \emph{in-neighbor}) of $u$ is a vertex $v$ such that $uv\in A(D)$ (respectively $vu\in A(D)$). The set of out-neighbors (respectively in-neighbors) of $u$ is called the \emph{out-neighborhood} (respectively \emph{in-neighborhood}) of $u$ and is denoted by $\OutN[D]{u}$ (respectively $\InN[D]{u}$).
Given a set \(U ⊆ \V{D}\) of vertices, we define \(\OutN[D]{U} = \bigcup_{u ∈ U}\OutN[D]{u} ∖ U\) and
\(\InN[D]{U} = \bigcup_{u ∈ U}\InN[D]{u} ∖ U\).
A \emph{source} of $D$ is a vertex $v$ with $\OutN[D]{v} = \emptyset$, and a \emph{sink} is a vertex $v$ with $\InN[D]{u} = \emptyset$.
The \emph{minimum degree} of an undirected graph \(G\) is given by \(\Mindeg{G} = \min_{v ∈ \V{G}} \Abs{\{e ∈ \E{G} \mid v ∈ e\}}\).
We say that $u$ \emph{dominates} $U$ in $D$ if $U\subseteq \OutN[D]{v}$. A set $X\subseteq V(D)$ disjoint from $U$ \emph{dominates} $U$ if every vertex $x\in X$ dominates $U$.

A \emph{directed walk} in $D$ is a finite sequence $(v_1,\dots,v_n)$ of vertices such that, for every $i\in [n-1]$, $v_iv_{i+1}\in A(D)$.
A \emph{directed path} is a directed walk with pairwise distinct vertices, and a \emph{directed cycle} is a directed walk $(v_1,\dots,v_{n},v_1)$, $n\geq2$, where $v_1,\dots,v_{n}$ are pairwise distinct.
We denote by $\vec{C_n}$ the digraph consisting of a directed cycle on $n$ vertices.
A digraph is \emph{acyclic} if it does not contain any directed cycle. The unique (up to isomorphism) acyclic tournament is called the \emph{transitive tournament}.

The digraph $D$ is \emph{strongly connected} if, for any pair of vertices $u,v\in V(D)$, there exists a directed path from $u$ to $v$ and a directed path from $v$ to $u$.
When at least one of these two paths exist for any pair, we say that $D$ is \emph{unilateral}.

A \emph{temporal digraph} $H$ is a finite sequence of digraphs $(H_i)_{i \in [t]}$ on the same vertex set $V(H) = V(H_1) = \ldots = V(H_t)$.
The \emph{order} of a (temporal) digraph is its number of vertices.
Each digraph $H_i$ is called a \emph{snapshot} of $H$, and the number of digraphs in the sequence is called the \emph{lifetime} of $H$, denoted by $\Lifetime{H}$.
The \emph{underlying graph} of $H$, denoted by $\UnderlyingGraph{H}$, is the union of the underlying graphs of the snapshots of $H$, that is $\UnderlyingGraph{H} = \UnderlyingGraph{H_1} \cup \ldots \cup \UnderlyingGraph{H_t}$.
Given two integers \(i \leq j ∈ [t]\),
we write \(\TemporalSlice{H}{i}{j}\) for the temporal digraph
\((H_{i}, H_{i+1}, \ldots, H_{j})\).

We say that $H$ is \emph{always-unilateral} (respectively \emph{always-strong}) if every $H_i$ is unilateral (respectively strongly connected).
Note that always-strong temporal digraphs are always-unilateral, but the converse does not hold in general.

A \emph{temporal walk} $W$ of $H$ is a sequence of vertices $(v_1,\dots,v_n)$ for which there exists an increasing sequence $1\leq k_1 <k_2 <\ldots <k_{n-1}\leq t$ such that, for every $i\in [n-1]$, $v_iv_{i+1} \in A(H_{k_i})$.
The \emph{length} of a (temporal) walk $(v_1,\dots,v_n)$ is $n-1$.
A \emph{temporal path} is a temporal walk with pairwise distinct vertices. 
A vertex $u$ is \emph{spanned} by $W$ if $u\in \{v_1,\dots,v_n\}$, and the set of vertices spanned by $W$ is denoted by $\V{W}$.
A \emph{temporal exploration} of $H$ is a temporal walk spanning all vertices of $H$.

Given two vertices $u,v\in V(H)$, a \emph{temporal $u$-$v$-walk} is a temporal walk starting on $u$ and ending on $v$.
We say that $v$ is \emph{reachable} from $u$ in $H$ if there exists a temporal $u$-$v$-walk in $H$. The set of vertices reachable from $u$ in $H$ is denoted by 
$\TemporalReach[H]{}{u}$. For an integer $s\in [t]$, the set of vertices reachable from $u$ in $\TemporalSlice{H}{1}{s}$ is denoted by $\TemporalReach[H]{s}{u}$. 
By convention, we set $\TemporalReach[H]{0}{u}= \{u\}$.
Note that $\TemporalReach[H]{s}{u} \subseteq \TemporalReach[H]{}{u}$. When $H$ is clear from the context, we simply write $\TemporalReach[]{}{u}$ and $\TemporalReach[]{s}{u}$ for $\TemporalReach[H]{}{u}$ and $\TemporalReach[H]{s}{u}$, respectively.
Given a set \(U ⊆ \V{D}\), we define \(\TemporalReach[H]{s}{U} = \bigcup_{u ∈ U}\TemporalReach[H]{s}{u}\).

Given two temporal walks \(W_1 = (v_{1}, v_{2}, \ldots, v_{a}), W_2 = (u_{1}, u_{2}, \ldots, u_{b})\),
we define the \emph{concatenation} of \(W_1\) and \(W_2\) as
\(W_1 · W_2 = (v_{1}, v_{2}, \ldots, v_{a}, u_{1}, u_{2}, \ldots, u_{b})\) if \(v_a \neq u_1\) and
\(W_1 · W_2 = (v_{1}, v_{2}, \ldots, v_{a}, u_{2}, \ldots, u_{b})\) if \(v_a = u_1\).

\section{Temporal exploration problem}
\label{sec:always-unilateral exploration}
The main goal of this section is to study temporal exploration in always-unilateral temporal digraphs. In Section~\ref{subsec:general_quadratic_upperbound}, we prove that every always-unilateral temporal digraph of order \(n\) and sufficiently large lifetime admits a temporal exploration; more precisely, a lifetime of at least $(\frac{3}{2} - o(1)) \cdot n^2$
is sufficient. We also establish auxiliary reachability results for always-unilateral temporal digraphs which are tight.

This is essentially tight, even in very restricted cases, as we prove in Section~\ref{subsec:general_quadratic_lowerbound} that there exist always-strong temporal digraphs $D$ with arbitrarily large order $n$, lifetime $\Lifetime{D}\geq \frac{1}{9}n^2$, whose underlying graph has maximum degree $3$, and that do not admit any temporal exploration.
For always-unilateral temporal digraphs, we can further have that every snapshot has maximum degree $2$.
As we explain below, this is in huge contrast with the undirected setting. 

\subsection{Always-unilateral temporal digraphs}
\label{subsec:general_quadratic_upperbound}

In the case of always-strong temporal digraphs,
an argument very similar to the undirected case works
for showing that \(\Lifetime{D} \geq (n-1)^2\) suffices
in order to guarantee the existence of a temporal exploration.

\begin{observation}
	\label{statement:strong-temporal-u-v-path}
	Let \(D\) be an always-strong temporal digraph.
	Let \(u,v ∈ \V{D}\).
	If \(\Lifetime{D} ≥ n - 1\),
	then there is a temporal \(u\)-\(v\)-path in \(D\).
\end{observation}
\begin{proof}
    Let $D=(D_i)_{i\in [t]}$.
	For every \(i ∈ [t - 1]\),
	if \(v ∉ \TemporalReach[D]{i}{u}\),
	then \(\Abs{\TemporalReach[D]{i+1}{u}} > \Abs{\TemporalReach[D]{i}{u}}\)
	as \(D_i\) is strongly connected.
	Since \(\Abs{\TemporalReach[D]{1}{u}} \geq 2\),
	after \(n - 1\) steps, we have
	\(\Abs{\TemporalReach[D]{t}{u}} = \Abs{\V{D}}\),
	in which case \(v ∈ \TemporalReach[D]{t}{u}\).
\end{proof}

Choosing an arbitrary order of the vertices \(v_{1}, v_{2}, \ldots, v_{n}\),
we apply \cref{statement:strong-temporal-u-v-path}
and obtain, for each \(i ∈ [n-1]\), a temporal \(v_{i}\)-\(v_{i+1}\)-path using at most \(n - 1\) snapshots.
Connecting all vertices in this fashion requires at most \((n-1)^2\) snapshots.
It is simple to see that \cref{statement:strong-temporal-u-v-path} is tight
by taking \(n - 2\) copies of the directed cycle on \(n\) vertices as snapshots of a temporal digraph.
Clearly, no vertex can reach its predecessor along the cycle in \(n - 2\) steps.

An analogue of \cref{statement:strong-temporal-u-v-path} for always-unilateral temporal digraphs holds,
as shown below.
\Cref{statement:lower bound:unilateral layers -> temporally unilateral} complements this by showing that our bounds are tight.

\begin{lemma}
	\label{statement:unilateral layers -> temporally unilateral}
	Let $D$ be an always-unilateral temporal digraph of order $n\geq 2$.
	If $\ell(D)\geq 2n-3$, then for every \(u,v \in \V{D}\), there is
	a temporal \((u,v)\)-walk or
	a temporal \((v,u)\)-walk in \(D\).
\end{lemma}
\begin{proof}
	Let $D=(D_i)_{i\in [t]}$ and \(u,v \in \V{D}\) be distinct.
    Assume for a contradiction that \(u\notin \TemporalReach[D]{t}{v}\) and \(v\notin \TemporalReach[D]{t}{u}\).

	Let \(1 \leq i \leq 2n - 4\).
	If there is directed path from $u$ to $v$ in \(D_{i + 1}\),
	then there is some arc \(ww' \in A(D_{i+1})\) 
	such that \(w \in \TemporalReach[D]{i}{u}\) and \(w' \notin \TemporalReach[D]{i}{u}\). Hence, \(w' \in \TemporalReach[D]{i+1}{u}\) and \(\Abs{\TemporalReach[D]{i+1}{u}} > \Abs{\TemporalReach[D]{i}{u}}\). 
    By symmetry, if there is a directed path from $v$ to $u$ in \(D_{i + 1}\), then \(\Abs{\TemporalReach[D]{i+1}{v}} > \Abs{\TemporalReach[D]{i}{v}}\). 
	Thus, at least one of the sets \(\TemporalReach[D]{i+1}{u}, \TemporalReach[D]{i+1}{v}\) is larger than the corresponding set at step \(i\).
    Therefore,  we have
    \[
        |\TemporalReach[D]{t}{u}| + |\TemporalReach[D]{t}{v}| \geq |\TemporalReach[D]{1}{u}| + |\TemporalReach[D]{1}{v}| + t-1.
    \]
	Observe that, if \(\Abs{\TemporalReach[D]{1}{v}} = 1\), then \(\Abs{\TemporalReach[D]{1}{u}} \geq 2\), as
	there can be at most one sink in an unilateral digraph. Hence,
    \[
        |\TemporalReach[D]{t}{u}| + |\TemporalReach[D]{t}{v}| \geq |\TemporalReach[D]{1}{u}| + |\TemporalReach[D]{1}{v}| + t-1 \geq t+2 \geq 2n-1.
    \]
    Therefore, one of the sets $\TemporalReach[D]{t}{u}, \TemporalReach[D]{t}{v}$ has size $n$, a contradiction. 
\end{proof}

\begin{proposition}
	\label{statement:lower bound:unilateral layers -> temporally unilateral}
	For every integer \(n \geq 3\) there is an always-unilateral temporal digraph \(D\) of order $n$ with $\Lifetime{D} = 2n-4$ containing two vertices \(u,v \in \V{D}\)
	such that
	no temporal $u$-$v$- or $v$-$u$-walk exists in \(D\).
\end{proposition}
\begin{proof}
	Let \(V = \{v_{1}, v_{2}, \ldots, v_{n}\}\).
	For every \(1 \leq i \leq n - 2\), let $D_{2i-1}$ consist of the directed path $(v_1, v_2, \ldots, v_n)$ and let $D_{2i}$ consist of the directed path $(v_n,v_{n-1},\ldots,v_1)$.
	Clearly, each $D_i$ is unilateral, and there is no temporal \((v_1,v_n)\)- or \((v_n,v_1)\)-path in \(D=(D_i)_{i\in [2n-4]}\).
\end{proof}

While in the undirected and in the always-strong settings we can construct
temporal explorations visiting the vertices in any desired order,
this is not possible in always-unilateral temporal digraphs:
If every snapshot of a temporal digraph is the same directed path, with vertices in the same order,
then there is only one order of the vertices which any temporal exploration can use,
no matter how many snapshots we have.

We thus need a way of choosing the correct order of vertices to visit.
The next two lemmas allow us to build a good ordering
by repeatedly picking an unexplored vertex which can temporally reach all unexplored vertices.

\begin{lemma}
	\label{statement:no growth implies containment}
	Let \(D\) be an unilateral digraph and \(A, B ⊆ \V{D}\).
	If \(\OutN{A}= \OutN{B} = \emptyset \), then \(A ⊆ B\) or \(B ⊆ A\).
\end{lemma}
\begin{proof}
		Assume towards a contradiction that \(A\) and \(B\) are incomparable with respect to \(⊆\).
		Choose \(x \in A \setminus B\) and \(y \in B \setminus A\).
		Since \(D\) is unilateral,
		there is a directed path \(P\) from \(x\) to \(y\), or from \(y\) to \(x\) in \(D\).
		In the first case, all vertices of \(P\) are inside \(A\), so
		\(y\in A\), a contradiction. The second case similarly gives
		\(x\in B\). Hence \(A\) and \(B\) are comparable with respect to \(⊆\).
\end{proof}

\begin{lemma}
	\label{statement:reachability growth}
	Let \(D\) be an always-unilateral temporal digraph on $n$ vertices,
	\(S ⊆ \V{D}\), \(S' ⊆ S\), and let \(d_0 \leq d_1\) be integers.
	If, for some integer $t'$, \(S ⊆ \TemporalReach{t'}{S'}\),
	\(\Abs{\TemporalReach{t'}{v}} \geq d_0\) for all \(v ∈ S'\),
	\(\Abs{\TemporalReach{t'}{v}} \geq d_1\) for some \(v ∈ S'\), and
	\(\Lifetime{D} \geq t' + 2n - 1 - (d_1 + d_0)\),
	then there is some \(v^* ∈ S'\)
	such that
	\(S ⊆ \TemporalReach{}{v^*}\).
\end{lemma}
\begin{proof}
    Let $D=(D_i)_{i\in [t]}$.
    Let \(M_{t'-1} = S'\).
	For each \(i ∈ [t',t]\), 
	let \(M_i \subseteq M_{i-1}\) be a minimal set
    such that
    \(S \subseteq \TemporalReach{i}{M_i}\).
    Furthermore, let
	\(d_{0,i} = \min_{v ∈ M_i}\Abs{\TemporalReach{i}{v}}\),
	\(d_{1,i} = \max_{v ∈ M_i}\Abs{\TemporalReach{i}{v}}\), and
	\(d_{+,i} = d_{0,i} + d_{1,i}\).

	If \(M_i = \{v^*\}\) for some \(i ∈ [t',t]\),
	then \(S ⊆ \TemporalReach{i}{v^*}\) and we are done.
	If \(d_{+,i} \geq 2n - 1\) for some \(i ∈ [t',t]\),
	then \(d_{1,i} \geq n\), and
	so there is some \(v ∈ M_i\)
	such that
	\(\TemporalReach{i}{v} = \V{D}\) and we are done.

	Otherwise, let \(v_1 ∈ M_i\) be a vertex
	such that
	\(\Abs{\TemporalReach{i}{v_1}} = d_{1,i}\).

	If \(\OutN[D_{i}]{\TemporalReach{i - 1}{v_1}} = \emptyset\),
	then by \cref{statement:no growth implies containment} and by definition of \(M_i\), we have
	\(\OutN[D_{i}]{\TemporalReach{i - 1}{v'}} \neq \emptyset\)
	for every \(v' ∈ M_i ∖ \{v_1\}\).
	In particular, \(d_{0,i} > d_{0,i - 1}\) and,
	thus, \(d_{+,i} > d_{+,i-1}\). 

	Otherwise, we have \(d_{1,i} > d_{1,i-1}\) and
	again \(d_{+,i} > d_{+,i-1}\).

	Let \(t'' = t' + 2n - 1 - (d_1 + d_0)\).
	The argumentation above implies that \(d_{+,t''} \geq 2n - 1\),
	and so \(d_{1,t''} \geq n\), or
	\(M_{t''} = \{v^*\}\).
	In both cases, as argued above, we obtain the claimed statement.
\end{proof}

With \cref{statement:reachability growth} we can construct a temporal exploration for always-unilateral temporal digraphs in the following way.
Clearly, for any set \(S ⊆ \V{D}\) and every \(v ∈ S\),
the set \(\TemporalReach{1}{v}\) contains at least one vertex.
Hence, \(d_1 \geq d_0 \geq 1\), and so in the first \(2n - 3\) snapshots
we can find a vertex \(v_1 ∈ \V{D}\) which can temporally reach all other vertices of \(D\).
We iterate this procedure, setting \(S_1 = \V{D}\) and \(S_{i + 1} = S_i ∖ \{v_i\}\).
We thus obtain an ordering \(v_{1}, v_{2}, \ldots, v_{n}\) of the vertices of \(D\),
such that \(v_i\) can temporally reach \(v_{i + 1}\) within \(2n - 3\) snapshots.
This implies that \(D\) admits a temporal exploration if \(\Lifetime{D} \geq (2n - 3)n ∈ \BigO{n^2}\).
Next, we improve the multiplicative constant of this bound.
Towards this end, we need the following lemma.

\begin{lemma}
\label{statement:unilateral layers -> temporal root} 
	Let \(D\) be an always-unilateral temporal digraph of order \(n\) and $S\subseteq V(D)$ with \(\Abs{S}=k\).
	If \(\Lifetime{D} \ge \bound{statement:unilateral layers -> temporal root}{t}{n,k} \coloneqq 2n - k - 1\),
	then there exists \(v^* \in S\) such that $S\subseteq\TemporalReach[D]{}{v^*}$.
\end{lemma}
\begin{proof}
    Let $D=(D_i)_{i\in [t]}$ and $|S|=k$.
	Write \(S=\{s_1,\ldots,s_k\}\).
	Given a set \(A ⊆ S\) and a time step \(i ∈ [t]\),
	define
	\(w_i(A) = \Abs{\TemporalReach{i}{A}}\).

	We iterate for \(i ∈ [t]\), constructing sets \(U'_i ⊆ U_i ⊆ S\)
	such that \(U_i\) is a minimal set satisfying \(U_i' \subseteq U_i\) and
	\begin{equation}
		S ⊆ \TemporalReach{i}{U_i},
		\tag{C1}
		\label{eq:C1}
	\end{equation}
	for every non-empty set \(A \subseteq U_i'\), we have
	\begin{equation}
		w_i(A) \geq i + \Abs{A},
		\tag{C2}
		\label{eq:C2}
	\end{equation}
	and
		\begin{equation}
			\label{eq:I1}
			\Abs{U_i'}+\mu_i\ge k,
			\tag{C3}
		\end{equation}
	where \(\mu_i:=\min_{v \in U_i}w_i(\{v\})\).

	Start by setting \(U_0 = U_0' = S\).
    Clearly, \eqref{eq:C1}, \eqref{eq:C2} and \eqref{eq:I1} hold for these sets.
    Moreover, by definition of $\TemporalReach[D]{0}{U_0}$, the set \(U_0\) is a minimal set satisfying \eqref{eq:C1}. 

	On step \(i\),
	stop the construction if 
    \(\Abs{U_{i-1}'} = 1\). or
        We say that a non-empty set \(A ⊆ U'_{i-1}\) is \emph{bad}
	if \eqref{eq:C2} does not hold for \(A\) on step \(i\).
	As \eqref{eq:C2} holds in time step \(i - 1\), every bad set \(A\) satisfies
	\[
		w_{i-1}(A) = i - 1 + \Abs{A}
		\quad\text{and}\quad
		\OutN[D_{i}]{\TemporalReach{i-1}{A}}
		= \OutN[D_{i}]{\TemporalReach{i}{A}} =
		\emptyset.
	\]

	By \cref{statement:no growth implies containment}, for two distinct bad sets \(A, B\)
	we have that the sets 
	\(\TemporalReach{i}{A}\) and
	\(\TemporalReach{i}{B}\)
	are comparable with respect to \(⊆\).
    We show that $A$ and $B$ are also comparable with respect to~\(⊆\).
	Without loss of generality, we have
	\(\TemporalReach{i}{A} ⊆ \TemporalReach{i}{B}\).
	The other case follows analogously.

	Assume towards a contradiction that there is some \(v ∈ A ∖ B\).
	Let \(B' = B ∪ \{v\}\).
	Then, since $v\in A$ and \(\TemporalReach{i}{A} ⊆ \TemporalReach{i}{B}\), we have
	\(\TemporalReach{i}{B'} = \TemporalReach{i}{B}\),
	which implies that
	\(\Abs{\TemporalReach{i}{B}} = \Abs{\TemporalReach{i}{B'}} \geq i - 1 + \Abs{B'} = i + \Abs{B}\),
	a contradiction to the assumption that \(B\) is bad.

	Thus, the family of bad sets forms a chain with respect to \(⊆\).
    If there are no bad sets, we choose \(U_i' = U_{i-1}'\).
	Otherwise, let \(A_0\) be the unique minimal bad set and choose \(v_0^i ∈ A_0\) arbitrarily.
	We set \(U'_i = U'_{i - 1} ∖ \{v_0^i\}\).
	As every bad set in \(U'_{i-1}\) contains \(v_0^i\),
	there are no bad set in \(U'_i\),
	satisfying \eqref{eq:C2} on step \(i\).

	Finally, choose \(U_i ⊆ U_{i - 1}\) as a minimal subset
	containing \(U'_i\)
	such that \eqref{eq:C1} holds.

	We now prove that \eqref{eq:I1} holds on step \(i\).
	Observe that \(\mu_{i-1} \leq \mu_{i}\), since
	\(\TemporalReach{i-1}{v} ⊆ \TemporalReach{i}{v}\) for any vertex $v$.

	Assume towards a contradiction that \eqref{eq:I1} does not hold.
	Since \(\Abs{U'_i} \geq \Abs{U_{i-1}'} - 1\) and
	\(\Abs{U'_0} = k\),
	we have \(\Abs{U_i'} \ge k - i\), and
	hence \(\mu_i<i\).
	Since \eqref{eq:I1} holds at time \(i-1\) and
	\(\mu_i\) is non-decreasing, the failure at
	time \(i\) must occur because \(\Abs{U'_i} = \Abs{U'_{i-1}} - 1\) and
	\(\mu_i=\mu_{i-1}\).

    Let \(v ∈ U_i\) such that \(w_i(\{v\}) = \mu_i\).
	It cannot be the case that \(v\in U_{i-1}'\), as
    \[
        w_{i-1}(\{v\}) \leq w_{i}(\{v\}) = \mu_i < i = (i-1)+|\{v\}|
    \]
    would contradict \eqref{eq:C2} for $i-1$.
    		    Hence $v\in U_i \setminus U_{i-1}'$.
			    As \(\mu_i = \mu_{i-1} \leq w_{i - 1}(\{v\}) \leq w_i(\{v\}) = \mu_i\),
    we have that \(w_{i-1}(\{v\}) = w_i(\{v\}) = \mu_i\).
	Therefore, \(\OutN[D_i]{\TemporalReach{i-1}{v}} = \emptyset\).
                                	
	Since \(\Abs{U'_i} <  \Abs{U'_{i-1}}\), 
	there is some bad set \(A ⊆ U'_{i-1}\) at step \(i\).
	As argued before, \(\OutN[D_i]{\TemporalReach{i-1}{A}} = \emptyset\).
	By \cref{statement:no growth implies containment},
	\(\TemporalReach{i}{A}\) and \(\TemporalReach{i}{v}\) are comparable with respect to \(⊆\).

    As \(v \notin U_{i-1}'\), by minimality of \(U_{i-1}\)
    there is some \(y_v \in \TemporalReach{i-1}{v}\)
    such that
    \(v\) is the only vertex in \(U_{i-1}\)
    with
    \(y_v \in \TemporalReach{i-1}{v}\),
    as otherwise we could remove \(v\) from \(U_{i-1}\) and
    still satisfy \eqref{eq:C1} at step \(i - 1\).
    As \(A \subseteq U_{i-1}' \subseteq U_{i-1}\),
    we have that \(y_v \notin \TemporalReach{i-1}{A} = \TemporalReach{i}{A}\).
    Hence, \(\TemporalReach{i}{A} \subseteq \TemporalReach{i}{v}\).
    This implies that \(w_i(\{v\}) \geq w_{i-1}(A) \geq i\),
    a contradiction to \(w_i(\{v\}) = \mu_i < i\).
 	Hence, \eqref{eq:I1} holds on step \(i\).

	Let \(t'\) be the number of iterations above.
    Assume towards a contradiction that \(t' \geq n\).
    Then, because \eqref{eq:C2} holds at step \(n\),
    we have \(w_{n}(U_{n}') \geq n + \Abs{U_{i}'} > n\), a contradiction.
	Hence, \(t' \leq n - 1 \leq \Lifetime{D}\) and
    the iteration stops when \(\Abs{U_{t'}'} = 1\).

	Let \(\{u\} = U_{t'}'\).
	Then \eqref{eq:C2} gives \(w_{t'}(\{u\}) \ge {t'}+1\).
	By \eqref{eq:I1}, we also have \(\mu_{t'}\ge k-1\).
	Hence, there is some \(v_1 ∈ U_{t'}\)
	such that
	\(w_{t'}(\{v_1\}) \geq t' + 1\) and
	for every \(v ∈ U_{t'}\)
	we have \(w_{t'}(\{v\}) \geq k - 1\).
	As \(t \geq 2n - k - 1 = t' + 2n - 1 - (t' + 1 + k - 1)\),
	we can apply \cref{statement:reachability growth} with
    \(S' = U_{t'}\),
    \(d_1 = t' + 1\) and
    \(d_0 = k - 1\),
	obtaining some \(v^* ∈ U_{t'} \subseteq S\)
	such that
	\(S ⊆ \TemporalReach{t}{v^*}\), as desired.
\end{proof}
We remark that Lemma~\ref{statement:unilateral layers -> temporal root} is tight by considering a temporal digraph with \(2n-k-2\) snapshots (see \cref{fig:lower_bound_temporal_root}).
Let \(V=\{v_1,\ldots,v_n\}\) and \(S=\{v_1\}\cup\{v_{n-k+2},v_{n-k+3},\ldots,v_n\}\).
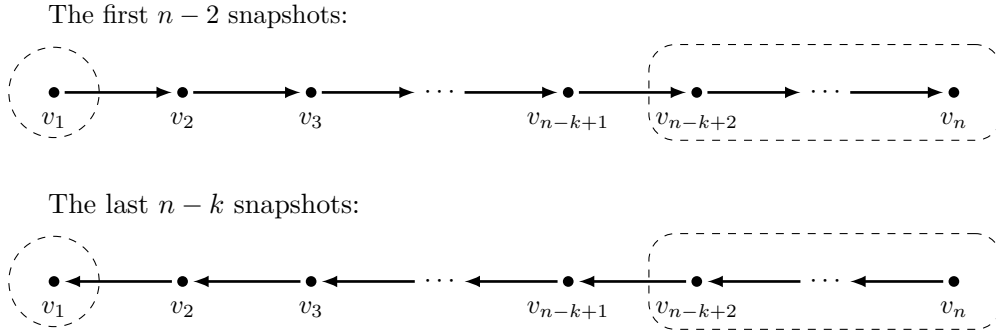
\begin{figure}[htbp]
    \centering
    \begin{tikzpicture}[scale=1]
        \tikzset{
                        marked/.style={circle, dashed, draw, inner sep=8pt},
            tailbox/.style={draw, dashed, rounded corners=10pt, inner sep=8pt}                    }

        \newcommand{\spc}{1.7}
                \begin{scope}
            \node at (-0.2,1.0) [anchor=west] {The first \(n-2\) snapshots:};

                        \node[vertex, label=below:$v_1$] (a1) at (0,0) {};
            \node[marked, fit=(a1)] {};

            \node[vertex, label=below:$v_2$] (a2) at (\spc,0) {};
            \node[vertex, label=below:$v_3$] (a3) at (\spc*2,0) {};
            \node (a3b) at (\spc*3,0) {\(\cdots\)};
            \node[vertex, label=below:$v_{n-k+1}$] (a4) at (\spc*4,0) {};
            \node[vertex, label=below:$v_{n-k+2}$] (a5) at (\spc*5,0) {};
            \node (a5b) at (\spc*6,0) {\(\cdots\)};
            \node[vertex, label=below:$v_n$] (a6) at (\spc*7,0) {};

                        \draw[directededge] (a1) -- (a2);
            \draw[directededge] (a2) -- (a3);
            \draw[directededge] (a3) -- (a3b);
            \draw[directededge] (a3b) -- (a4);
            \draw[directededge] (a4) -- (a5);
            \draw[directededge] (a5) -- (a5b);
            \draw[directededge] (a5b) -- (a6);

                        \node[tailbox, fit=(a5)(a6), inner sep=14pt] {};
        \end{scope}

                \begin{scope}[yshift=-2.3cm]
            \node at (-0.2,1.0) [anchor=west] {\large The last \(n-k\) snapshots:};

                        \node[vertex, label=below:$v_1$] (b1) at (0,0) {};
            \node[marked, fit=(b1)] {};

            \node[vertex, label=below:$v_2$] (b2) at (\spc,0) {};
            \node[vertex, label=below:$v_3$] (b3) at (\spc*2,0) {};
            \node (b3b) at (\spc*3,0) {\(\cdots\)};
            \node[vertex, label=below:$v_{n-k+1}$] (b4) at (\spc*4,0) {};
            \node[vertex, label=below:$v_{n-k+2}$] (b5) at (\spc*5,0) {};
            \node (b5b) at (\spc*6,0) {\(\cdots\)};
            \node[vertex, label=below:$v_n$] (b6) at (\spc*7,0) {};

                        \draw[directededge] (b2) -- (b1);
            \draw[directededge] (b3) -- (b2);
            \draw[directededge] (b3b) -- (b3);
            \draw[directededge] (b4) -- (b3b);
            \draw[directededge] (b5) -- (b4);
            \draw[directededge] (b5b) -- (b5);
            \draw[directededge] (b6) -- (b5b);

                        \node[tailbox, fit=(b5)(b6), inner sep=14pt] {};
        \end{scope}
    \end{tikzpicture}
		\caption{Tightness of Lemma~\ref{statement:unilateral layers -> temporal root}: the dashed box represents the initial set \(S\).}
        \label{fig:lower_bound_temporal_root}
\end{figure}

Note that \(v_1\) cannot reach \(v_n\), and the vertex \(v_q\) cannot reach \(v_1\) for every
\(q\in\{n-k+2,\ldots,n\}\). Therefore, no vertex of \(S\) reaches all vertices of
\(S\) within \(2n-k-2\) snapshots. Hence the bound \(2n-k-1\) is best
possible.

We complete this section with our main theorem.
In the case where \(S = \V{D}\),
we get that a lifetime of \((3n^2 - 7n)/2\) is sufficient to guarantee a temporal exploration.

\hypertarget{statement:unilateral layers -> temporal exploration:body}{}
\unilateralExploration*
\newcommand{\unilateralexplorationref}{	\hyperlink{statement:unilateral layers -> temporal exploration:body}{		\cref*{statement:unilateral layers -> temporal exploration}}}

\begin{proof}
	Let \(\ell_i = \bound{statement:unilateral layers -> temporal root}{t}{n,i} = 2n - i - 1\), 
	\(t_{k+1} = 0\) and, for $2\leq i\leq k$,
	\(t_i = \sum_{j=i}^k\ell_j\).
	Thus
	\begin{align*}
		t_{2} & = \sum_{i=2}^{k}\ell_i = (2n - 1)(k - 1) - \sum_{i=2}^ki 
		\\[0em] & = (2n - 1)(k - 1)- k(k + 1)/2 - 1
		\\[0em] & = 2kn -(k^2 + 3k)/2 - 2n \leq \Lifetime{D}.
	\end{align*}
	For each \(2 \leq i \leq k\),
	let \(H_i = \TemporalSlice{D}{1 + t_{i + 1}}{t_i}\).
	Note that \(\Lifetime{H_i} = \ell_i\) and
	that \(H_{i+1}\) occurs before \(H_i\) along the snapshots of \(D\).

	Set \(S_k = S\).
	Iterate for \(i = k\) to \(2\) as follows.
	On step \(i\), by \cref{statement:unilateral layers -> temporal root},
	there is some \(v_i ∈ S_i\)
	such that 
	\(S_i ⊆ \TemporalReach[H_i]{\ell_i}{v_i}\).
	Set \(S_{i - 1} = S_i ∖ \{v_i\}\). 
	If \(i < k\),
	let \(W_i\) be a temporal \(v_{i+1}\)-\(v_{i}\)-walk in \(H_{i+1}\).
	Because \(v_i ∈ S_i ⊆ S_{i+1}\), such a walk exists.

	After completing the iteration above, we obtain temporal paths \(W_{k-1}, W_{k-2}, \ldots, W_{1}\).
	Let \(W = W_{k-1} · W_{k-2} · \ldots · W_{1}\).
	It is immediate by construction that \(S ⊆ \V{W}\).
\end{proof}

\Cref{statement:always-unilateral walk covering S} improves upon \cite[Lemma 6.14]{hkmm26cowsi},
which previously required \(\Lifetime{D}\) to be greater than \(\BigO{k^2n^{kn + 2}}\).
This answers their question of whether it is possible to obtain a polynomial bound
for this statement.

\subsection{Constructions with large lifetime and no temporal exploration}
\label{subsec:general_quadratic_lowerbound}

As explained at the beginning of the present section, the bound of Theorem~\ref{statement:unilateral layers -> temporal exploration} is essentially tight (up to a multiplicative constant).
Indeed, it is known (see~\cite{changJCO25}) that there exist strongly connected digraphs of arbitrarily large order $n$ without any Hamiltonian walk of length less than $\frac{1}{4}n^2$.
Therefore, the bound of Theorem~\ref{statement:unilateral layers -> temporal exploration} is tight -- up to a multiplicative constant -- already in the static case (that is, when all snapshots of the temporal digraphs are identical).
We note that this remains true even for oriented graphs with maximum degree $3$.

\begin{proposition}
	\label{statement:degree-three-lower-bound}
	For every integer \(n' \geq 1\) there is an always-strong temporal digraph \(D\) of order \(n\geq n'\) such that $\UnderlyingGraph{D}$ has maximum degree $3$, \(\Lifetime{D} \geq \frac{1}{9}n^2\), all snapshots of $D$ are equal, and \(D\) does not admit any temporal exploration.
\end{proposition}

The result above is in huge contrast with the undirected setting. Indeed, Bastide, Groenland, Michel, and Rambaud~\cite{BGMR2025ietg} recently obtained that every always-connected temporal graph $G$ with maximum degree $\Delta$ and lifetime at least $O(n^{3/2}\sqrt{\Delta\log n})$ admits a temporal exploration. In fact, this holds even in the more general case where $\Delta$ bounds the average degree (taken over all snapshots of $G$) of each vertex of $G$.

\begin{proof}[Proof of \cref{statement:degree-three-lower-bound}]
	For each \(i ∈ [n']\) let \(P_i\) be the directed path \((u_i, v_i, w_i)\).
	Let \(C\) be the cycle \((v_{n'}, w_{n'}, w_{n' - 1}, \ldots,\allowbreak w_{1}, u_{1}, u_{2}, \ldots, u_{n'}, v_{n'})\).
	Let \(D'\) be the digraph consisting of the union of all \(P_i\)'s and \(C\), see \cref{fig:strong-degree-three-lower-bound} for an illustration of the construction.
    
    \begin{figure}
    	\centering
    	\begin{tikzpicture}
    		\node[vertex] (v1) at (0,0) {};
    		\node[left =0.0mm of v1] (v1l) {$v_1$};
    		\node[vertex] (v2) at (1,0) {};
    		\node[left =0.0mm of v2] (v2l) {$v_2$};
    		\node[vertex] (v3) at (2,0) {};
    		\node[left =0.0mm of v3] (v3l) {$v_3$};
    		\node[vertex] (w1) at (0,1) {};
    		\node[above =0.0mm of w1] (w1l) {$w_1$};
    		\node[vertex] (w2) at (1,1) {};
    		\node[above =0.0mm of w2] (w2l) {$w_2$};
    		\node[vertex] (w3) at (2,1) {};
    		\node[above =0.0mm of w3] (w3l) {$w_3$};
    		\node[vertex] (u1) at (0,-1) {};
    		\node[below =0.0mm of u1] (u1l) {$u_1$};
    		\node[vertex] (u2) at (1,-1) {};
    		\node[below =0.0mm of u2] (u2l) {$u_2$};
    		\node[vertex] (u3) at (2,-1) {};
    		\node[below =0.0mm of u3] (u3l) {$u_3$};
    		\path[directededge] (u1) to (v1);
    		\path[directededge] (v1) to (w1);
    		\path[directededge] (u2) to (v2);
    		\path[directededge] (v2) to (w2);
    		\path[directededge] (u3) to (v3);
    		\path[directededge] (v3) to (w3);
    		\path[directededge] (u1) to (u2);
    		\path[directededge] (u2) to (u3);
    		\path[directededge] (w3) to (w2);
    		\path[directededge] (w2) to (w1);
    		\path[directededge, bend right = 70] (w1) to (u1);
    		\path[directededge, dashed, color=blue] 
    		(2.2,-0.8) to      		(2.2,0) to         		(2.2,0.8) to       		(-1.0,0.8) to      		(-1.0,-0.8) to     		(0.2, -0.8) to     		(0.2, 0.6) to      		(-0.8, 0.6) to      		(-0.8, -0.6) to     		(1.2, -0.6) to
    		(1.2, 0.0);
    	\end{tikzpicture}
    	\caption{\label{fig:strong-degree-three-lower-bound}
    	The digraph \(D'\) constructed in the proof of \cref{statement:degree-three-lower-bound} for \(n' = 3\).
    	A Hamiltonian walk of length 10 is given by the dashed path.
    	}
    \end{figure}
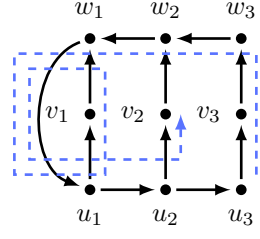

	We first prove that a shortest Hamiltonian walk \(W\) in \(D'\) has length at least \(\ell \coloneqq (n')^2\).
	Let \(W_{i,j}\) be a subpath of \(W\) from \(v_i\) to \(v_j\) which is internally disjoint from \(V\coloneqq \{v_1,\dots,v_{n'}\}\).
	Clearly, \(D'\) contains only one possibility for \(W_{i,j}\), visiting \(w_i\) to \(w_{1}\) and then \(u_{1}\) to \(u_j\).
	Hence, \(W_{i,j}\) has length \(i + j + 1\).

	Let \(v_a\) be the first vertex of \(V\) in \(W\) and let \(v_b\) be the last.
	Let \(V' = V \setminus \{v_a, v_b\}\).
	For each vertex \(v_i ∈ V'\), there are two subpaths \(W_{j,i}\), \(W_{i,k}\) in \(W\),
	one from some \(v_j\) to \(v_i\) and one from \(v_i\) to some \(v_k\).
	The length of \(W\) is thus at least 
	\begin{align*}
		\sum_{v_i ∈ V'}(2i + 1) + a + b + 1 = \sum_{v_i ∈ V}(2i + 1) - a - b - 1 & = 2\sum_{i=1}^{n'}i + n' - a - b - 1\\[0em]
		& = n'(n' + 1) + n' - a - b - 1.
	\end{align*}
    
	This value is minimized when we choose \(a = n'\) and \(b = n' - 1\), yielding the length of
	\(n'(n' + 1) + n' - 2n' = (n')^2\) for \(W\).

	Let \(D\) be the temporal digraph consisting of \(\ell - 1\) snapshots, all equal to \(D'\).
	It is clear from the definition of temporal walks that
	a temporal exploration of \(D\) corresponds to a Hamiltonian walk for \(D'\)
	of length \(\ell - 1\).
	However, as proven above, there is no Hamiltonian walk of length \(\ell\) in \(D'\).
	Hence, \(D'\) does not admit a temporal exploration.

	Since \(D\) contains \(n \coloneqq 3n'\) vertices, we have that \(\ell = n^2/9\).
	Further, it is immediate from the construction that maximum degree of \(\UnderlyingGraph{D}\) is 3,
	as every snapshot is the same and has maximum degree 3.
\end{proof}

For always-unilateral temporal digraph, we can further ask that each snapshot has maximum degree $2$.

\hypertarget{thm:always_path_maxdegree_3:body}{}
\pathBoundedDegreeLowerBound*

To prove \pathmaxdegreethreeref{}, we first prove that a quadratic bound can be obtained when the degree condition holds for all vertices except one. We then show that we can blow-up this specific vertex with a gadget of linear size, hence obtaining the degree condition on all vertices.
We prove these two steps in the two following lemmas, respectively, and then put the pieces together to derive Theorem~\ref{thm:always_path_maxdegree_3}.

\begin{lemma}
    \label{lemma:lower_bound_path_apex}
    For every integer $n\geq 0$, there exists a temporal digraph $H$ on $n+1$ vertices $v_0,\dots,v_n$ such that:
    \begin{itemize}
        \item $\Lifetime{H}= \frac{1}{2}n(n+1)-1$;
        \item each snapshot of $H$ is a directed path;
        \item for every $j\geq 1$, $v_j$ has degree at most $3$ in $\UnderlyingGraph{H}$, and $v_j$ is adjacent to $v_0$ in $\UnderlyingGraph{H}$; 
        \item $v_0$ is adjacent to $v_1$ in every snapshot of $H$; and
        \item $H$ does not admit any temporal exploration. 
    \end{itemize}
\end{lemma}
\begin{proof}
    Let us define a temporal digraph $H$ as follows. 
    Let $V(H) = \{v_0,v_1,\ldots,v_n\}$. For every $i\in [n]$, let $D_i$ be the digraph consisting of the directed path
    \[
        (v_{i-1},v_{i-2},\ldots,v_0,v_{i},v_{i+1},\ldots,v_n).
    \]
    We let $H$ be the temporal digraph consisting of $i$ copies of $D_i$ for every $i\in [n-1]$ and $n-1$ copies of $D_n$ (where the copies of each $D_i$ appear consecutively and in the natural ordering $D_1,D_2,\ldots,D_n$). See Figure~\ref{fig:quadratic_bound_paths_unb_degree} for an illustration.

    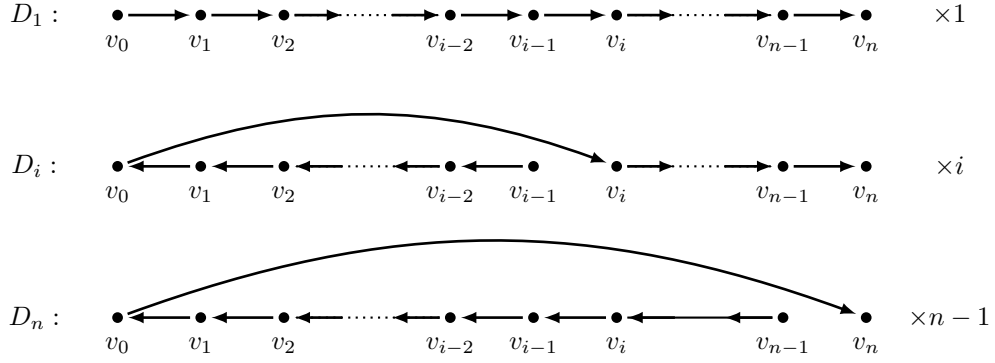
\begin{figure}
    \centering
    \begin{tikzpicture}
        \newcommand{\spc}{1.1}
                \begin{scope}
            \node[] at (-\spc,0) {$D_1:$};
            \node[vertex, label=below:$v_0$] (a1) at (0,0) {};
            \node[vertex, label=below:$v_1$] (a2) at (\spc,0) {};
            \node[vertex, label=below:$v_2$] (a3) at (\spc*2,0) {};
            \node[vertex, label=below:$v_{i-2}$] (a3b) at (\spc*4,0) {};
            \node[vertex, label=below:$v_{i-1}$] (a4) at (\spc*5,0) {};
            \node[vertex, label=below:$v_{i}$] (a5) at (\spc*6,0) {};
            \node[vertex, label=below:$v_{n-1}$] (a5b) at (\spc*8,0) {};
            \node[vertex, label=below:$v_n$] (a6) at (\spc*9,0) {};
            
            \node[] at (\spc*10,0) {$\times 1$};

                        \draw[directededge] (a1) -- (a2);
            \draw[directededge] (a2) -- (a3);
            \draw[directededge] (a3) -- (\spc*2.7,0);
            \draw[thick,dotted] (a3) -- (a3b);
            \draw[directededge] (\spc*3.3,0) -- (a3b);
            \draw[directededge] (a3b) -- (a4);
            \draw[directededge] (a4) -- (a5);
            \draw[directededge] (a5) -- (\spc*6.7,0);
            \draw[thick,dotted] (a5) -- (a5b);
            \draw[directededge] (\spc*7.3,0) -- (a5b);
            \draw[directededge] (a5b) -- (a6);
        \end{scope}

        \begin{scope}[yshift=-2cm]
            \node[] at (-\spc,0) {$D_i:$};
            \node[vertex, label=below:$v_0$] (a1) at (0,0) {};
            \node[vertex, label=below:$v_1$] (a2) at (\spc,0) {};
            \node[vertex, label=below:$v_2$] (a3) at (\spc*2,0) {};
            \node[vertex, label=below:$v_{i-2}$] (a3b) at (\spc*4,0) {};
            \node[vertex, label=below:$v_{i-1}$] (a4) at (\spc*5,0) {};
            \node[vertex, label=below:$v_{i}$] (a5) at (\spc*6,0) {};
            \node[vertex, label=below:$v_{n-1}$] (a5b) at (\spc*8,0) {};
            \node[vertex, label=below:$v_n$] (a6) at (\spc*9,0) {};
            
            \node[] at (\spc*10,0) {$\times i$};

                        \draw[directededge] (a2) -- (a1);
            \draw[directededge] (a3) -- (a2);
            \draw[directededge] (\spc*2.7,0) -- (a3);
            \draw[thick,dotted] (a3b) -- (a3);
            \draw[directededge] (a3b)-- (\spc*3.3,0);
            \draw[directededge] (a4) -- (a3b);
            \draw[directededge] (a1) to[out=20, in=160] (a5);
            \draw[directededge] (a5) -- (\spc*6.7,0);
            \draw[thick,dotted] (a5) -- (a5b);
            \draw[directededge] (\spc*7.3,0) -- (a5b);
            \draw[directededge] (a5b) -- (a6);
        \end{scope}

        \begin{scope}[yshift=-4cm]
            \node[] at (-\spc,0) {$D_n:$};
            \node[vertex, label=below:$v_0$] (a1) at (0,0) {};
            \node[vertex, label=below:$v_1$] (a2) at (\spc,0) {};
            \node[vertex, label=below:$v_2$] (a3) at (\spc*2,0) {};
            \node[vertex, label=below:$v_{i-2}$] (a3b) at (\spc*4,0) {};
            \node[vertex, label=below:$v_{i-1}$] (a4) at (\spc*5,0) {};
            \node[vertex, label=below:$v_{i}$] (a5) at (\spc*6,0) {};
            \node[vertex, label=below:$v_{n-1}$] (a5b) at (\spc*8,0) {};
            \node[vertex, label=below:$v_n$] (a6) at (\spc*9,0) {};
            
            \node[] at (\spc*10,0) {$\times n-1$};

                        \draw[directededge] (a2) -- (a1);
            \draw[directededge] (a3) -- (a2);
            \draw[directededge] (\spc*2.7,0) -- (a3);
            \draw[thick, dotted] (a3b) -- (a3);
            \draw[directededge] (a3b)-- (\spc*3.3,0);
            \draw[directededge] (a4) -- (a3b);
            \draw[directededge] (a5) -- (a4);
            \draw[directededge] (\spc*6.7,0) -- (a5);
            \draw[thick] (a5b) -- (a5);
            \draw[directededge] (a5b) -- (\spc*7.3,0);
            \draw[directededge] (a1) to[out=20,in=160] (a6);
        \end{scope}
    \end{tikzpicture}
    \caption{An illustration of the construction of $H$.}
    \label{fig:quadratic_bound_paths_unb_degree}
\end{figure}

    Note that, by construction, $H$ has $\frac{1}{2}n(n+1)-1$ snapshots, each snapshot of $H$ is a directed path, all vertices except $v_0$ has degree $3$ in $\UnderlyingGraph{H}$ (the neighborhood of $v_j$ being precisely $\{v_{j-1},v_{j+1},v_0\}$), and $v_0$ is adjacent to $v_1$ in every snapshot of $H$.
    
    We now prove that $H$ does not admit any  temporal exploration. Assume for a contradiction that such a temporal exploration exists, so there exists a sequence of vertices
    \[
        (u_0,u_1,\ldots,u_{\ell(H)})
    \]
    where, for every $j\in [\ell(H)]$, either $u_{j-1}=u_j$ or $u_{j-1}u_j \in A(H_j)$ -- where $H_j$ is the $j^{\rm th}$ snapshot of $H$ -- and all vertices of $H$ appear along the sequence. Note that $u_{\ell(H)} = v_n$ since $v_n$ is a sink in every snapshot of $H$.

    For every $i\in [n-1]$, let $f_i = u_{\frac{1}{2}i(i+1)}$, and let $S_i = \{u_j: j\leq \frac{1}{2}i(i+1)\}$.
    Intuitively speaking, the sequence $(u_j)_{j\in [\ell(H)]}$ can be seen as the sequence of moves of an agent exploring $H$. In this case, $f_i$ corresponds to the position of the agent after the sequence of all $D_j$'s, for $j\leq i$, and $S_i$ corresponds to the set of vertices spanned by the agent.
    We obtain the following by induction.

    \begin{claim}
        For every $i\in [n-1]$, each of the following holds:
        \begin{itemize}
            \item if $f_i = v_i$ then $S_i \cap \{v_{i+1},v_{i+2},\ldots,v_n\}=  \emptyset$;
            \item if $f_i\in \{v_0,v_1,\ldots,v_{i-1}\}$ then $S_i \cap \{v_{i},v_{i+1},\ldots,v_n\}=  \emptyset$; and
            \item if $f_i\in \{v_{i+1},v_{i+2},\ldots,v_{n}\}$ then $\{v_0,v_{1},v_{2},\ldots,v_{i}\} \setminus S_i \neq  \emptyset$.
        \end{itemize}
    \end{claim}
    \begin{proofclaim}
        Let us fix $i\in [n]$, and assume that the statement holds for all $i'<i$.
        
        Observe first that, among the $\frac{1}{2}i(i+1)$ first snapshots of $H$ (which are copies of $D_1,\ldots,D_i$), there is no arc from $\{v_{i+1},\ldots,v_n\}$ to $\{v_{1},\ldots,v_i\}$. Therefore, if $f_i\in \{v_{1},\ldots,v_i\}$, then $S_i \cap \{v_{i+1},\ldots,v_n\} = \emptyset$, for otherwise there would be some index $j\leq \frac{1}{2}i(i+1)-1$ such that $u_j \in \{v_{i+1},\ldots,v_n\}$ and $u_{j+1} \in \{v_{1},\ldots,v_i\}$, a contradiction. This shows the first item of the claim.
        Similarly, among the $\frac{1}{2}i(i+1)$ first snapshots of $H$, there is no arc from $\{v_{i},\ldots,v_n\}$ to $\{v_{1},\ldots,v_{i-1}\}$. The second item follows.

        Henceforth, we assume that $f_i \in \{v_{i+1},v_{i+2},\ldots,v_{n}\}$. If $i=1$ then $f_1 \notin \{v_0,v_1\}$, so $f_1$ is at distance at least $2$ from $v_0$ in $D_1$. Since there is only one snapshot corresponding to $D_1$,  $v_0 \notin S_i$. 
        Hence, assume that $i>1$. We distinguish the three possible cases for $f_{i-1}$.

        If $f_{i-1} = v_{i-1}$, then observe that $f_i$ is at distance $i+1$ from $v_{i-1}$. Since there are only $i$ copies of $D_i$, $f_i$ is not reachable from $f_{i-1}$ in the snapshots corresponding to $D_i$, a contradiction.

        Else if $f_{i-1} \in \{ v_{0},\ldots,v_{i-2}\}$ then by induction $v_{i-1} \notin S_{i-1}$. Since $v_{i-1}$ is a source in $D_i$, it cannot be spanned from $f_{i-1}$ in the snapshots corresponding to $D_i$. Therefore, $v_{i-1} \notin S_{i}$, and the third item holds.

        Finally, if $f_{i-1} \in \{v_{i},\ldots,v_{n}\}$ then by induction there exists $u\in \{v_0,\ldots,v_{i-1}\} \setminus S_{i-1}$. Since $u$ is not reachable from $v_i$ in $D_i$, it follows that $u\notin S_i$. The claim follows.
    \end{proofclaim}

    We apply the claim above with $i=n-1$. 
    If $f_{n-1}=v_{n-1}$, then $v_n\notin S_n$. Since $v_n$ is at distance $n$ from $v_{n-1}$ in $D_n$, and there are only $n-1$ copies of $D_n$, $v_n$ does not appear along the sequence $(u_0,\dots, u_{\ell(H)})$, a contradiction.
    If $f_{n-1} \in \{v_0,\ldots,v_{n-2}\}$, then $v_{n-1}\notin S_{n-1}$. Since $v_{n-1}$ is a source in $D_n$, $v_{n-1}$ is not spanned by $(u_0,\dots, u_{\ell(H)})$, a contradiction.
    Hence, $f_{n-1} = v_{n}$. Therefore, there exists $u\neq v_n$ such that $u\notin S_{n-1}$, and $u$ is not reachable from $v_n$ is $D_n$. Therefore, $u$ is not spanned by $(u_0,\dots, u_{\ell(H)})$, a contradiction. The lemma follows.
\end{proof}

\begin{lemma}
    \label{lemma:hamiltonian_subcubic_graph}
    For every integer $n\geq 2$, there exists a graph $W_n$ and a set $\{s_1,\dots,s_n\} \subseteq V(W_n)$ such that each of the following holds:
    \begin{itemize}
        \item $|V(W_n)| = 7n-6$;
        \item $\Delta(W_n) \leq 3$;
        \item for every $i\in [n]$, $d_{W_n}(s_i) \leq 2$; and
        \item for every $i\in [2,n]$, there exists a Hamiltonian path of $W_n$ with extremities $\{s_1,s_i\}$.
    \end{itemize}
\end{lemma}
\begin{proof}
    We define $W_n$ explicitly. Let 
    \[
        V(W_n) = \{s_1\} \cup \bigcup_{i=2}^{n} \{u_i,v_i,w_i,x_i,y_i,z_i,s_i\}
    \]
    and
    \begin{align*}
        E(W_n) = &~ \{s_{i-1}u_i,u_iv_i,v_iw_i, w_ix_i,x_iy_i,y_iz_i, z_is_i, u_iy_i, w_iz_i: 2\leq i \leq n\}\\[0em]
        & \cup \{x_iv_{i+1} : 2\leq i \leq n-1\}\\[0em]
        & \cup \{ x_ns_n\}.
    \end{align*}
    See Figure~\ref{fig:blow_up_gadget} for an illustration. It is straightforward to check that $W_n$ has $7n-6$ vertices, $\Delta(W_n)=3$, and all vertices $s_1,\dots,s_n$ have degree at most $2$ in $W_n$.
    It remains to justify that, for every $i \in [2,n]$, $W_n$ admits a Hamiltonian path going from $s_1$ to $s_i$. Let us fix $i\in [2,n]$.

    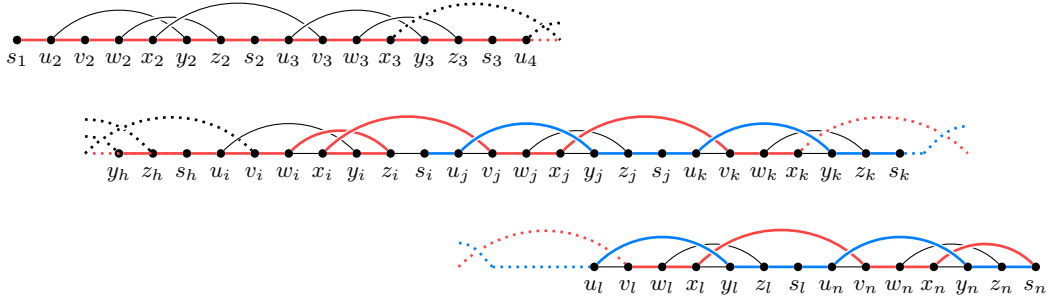
\begin{figure}
        \centering
        \begin{tikzpicture}
                        \newcommand{\spc}{0.449}
            \node[smallvertex, label=below:{\footnotesize $s_1$}] (s1) at (0,0) {};
            \node[smallvertex, label=below:{\footnotesize $u_2$}] (u2) at (\spc,0) {};
            \node[smallvertex, label=below:{\footnotesize $v_2$}] (v2) at (\spc*2,0) {};
            \node[smallvertex, label=below:{\footnotesize $w_2$}] (w2) at (\spc*3,0) {};
            \node[smallvertex, label=below:{\footnotesize $x_2$}] (x2) at (\spc*4,0) {};
            \node[smallvertex, label=below:{\footnotesize $y_2$}] (y2) at (\spc*5,0) {};
            \node[smallvertex, label=below:{\footnotesize $z_2$}] (z2) at (\spc*6,0) {};
            \node[smallvertex, label=below:{\footnotesize $s_2$}] (s2) at (\spc*7,0) {};
            
            \node[smallvertex, label=below:{\footnotesize $u_3$}] (u3) at (\spc*8,0) {};
            \node[smallvertex, label=below:{\footnotesize $v_3$}] (v3) at (\spc*9,0) {};
            \node[smallvertex, label=below:{\footnotesize $w_3$}] (w3) at (\spc*10,0) {};
            \node[smallvertex, label=below:{\footnotesize $x_3$}] (x3) at (\spc*11,0) {};
            \node[smallvertex, label=below:{\footnotesize $y_3$}] (y3) at (\spc*12,0) {};
            \node[smallvertex, label=below:{\footnotesize $z_3$}] (z3) at (\spc*13,0) {};
            \node[smallvertex, label=below:{\footnotesize $s_3$}] (s3) at (\spc*14,0) {};
            \node[smallvertex, label=below:{\footnotesize $u_4$}] (u4) at (\spc*15,0) {};

            \draw[edge, red] (s1) -- (u2);
            \draw[edge, red] (u2) -- (v2);
            \draw[edge, thin] (u2) to[out=45, in=135] (y2);
            \draw[edge, red] (v2) -- (w2);
            \draw[edge, red] (w2) -- (x2);
            \draw[edge, red] (x2) -- (y2);
            \draw[edge, red] (y2) -- (z2);
            \draw[edge, red] (z2) -- (s2);
            \draw[thickedge, white] (w2) to[out=45, in=135] (z2);
            \draw[edge, thin] (w2) to[out=45, in=135] (z2);
            \draw[thickedge, white] (x2) to[out=45, in=135] (v3) ;
            \draw[edge, thin] (x2) to[out=45, in=135] (v3) ;

            \draw[edge, red] (s2) -- (u3);
            \draw[edge, red] (u3) -- (v3);
            \draw[thickedge, white] (u3) to[out=45, in=135] (y3);
            \draw[edge, thin] (u3) to[out=45, in=135] (y3);
            \draw[edge, red] (v3) -- (w3);
            \draw[edge, red] (w3) -- (x3);
            \draw[edge, red] (x3) -- (y3);
            \draw[edge, red] (y3) -- (z3);
            \draw[edge, red] (z3) -- (s3);
            \draw[thickedge, white] (w3) to[out=45, in=135] (z3);
            \draw[edge, thin] (w3) to[out=45, in=135] (z3);
            \draw[thickedge, white] (x3) to[out=45, in=135] (\spc*16,0) ;
            \draw[edge, dotted] (x3) to[out=45, in=135] (\spc*16,0) ;
            \draw[edge,red] (s3) to (u4);
            \draw[edge,red, dotted] (u4) to (\spc*16,0);
            \draw[thickedge, white] (u4) to[out=45,in=180] (\spc*16,0.5*\spc);
            \draw[edge, dotted] (u4) to[out=45,in=180] (\spc*16,0.5*\spc);

            \begin{scope}[xshift=\spc*5cm, yshift=-1.5cm]
            
            \node[smallvertex, label=below:{\footnotesize $y_{h}$}] (yh) at (-2*\spc,0) {};
            \node[smallvertex, label=below:{\footnotesize $z_{h}$}] (zh) at (-\spc,0) {};
            \node[smallvertex, label=below:{\footnotesize $s_{h}$}] (sh) at (0,0) {};
            
            \draw[edge, red] (yh) -- (zh);
            \draw[edge, red] (zh) -- (sh);
            \draw[edge, red, dotted] (-3*\spc,0) -- (sh);
            \draw[edge, dotted] (-3*\spc,0.5*\spc) to[out=0, in=135] (yh);
            \draw[edge, dotted] (-3*\spc,\spc) to[out=0, in=135] (zh);
            
            \node[smallvertex, label=below:{\footnotesize $u_i$}] (ui) at (\spc,0) {};
            \node[smallvertex, label=below:{\footnotesize $v_i$}] (vi) at (\spc*2,0) {};
            \node[smallvertex, label=below:{\footnotesize $w_i$}] (wi) at (\spc*3,0) {};
            \node[smallvertex, label=below:{\footnotesize $x_i$}] (xi) at (\spc*4,0) {};
            \node[smallvertex, label=below:{\footnotesize $y_i$}] (yi) at (\spc*5,0) {};
            \node[smallvertex, label=below:{\footnotesize $z_i$}] (zi) at (\spc*6,0) {};
            \node[smallvertex, label=below:{\footnotesize $s_i$}] (si) at (\spc*7,0) {};
            
            \draw[thickedge, white] (-3*\spc,0) to[out=45, in=135] (vi);
            \draw[edge, dotted] (-3*\spc,0) to[out=45, in=135] (vi);
            \draw[edge, red] (sh) -- (ui);
            \draw[edge, red] (ui) -- (vi);
            \draw[thick, white] (ui) to[out=45, in=135] (yi);
            \draw[edge, thin] (ui) to[out=45, in=135] (yi);
            \draw[edge, red] (vi) -- (wi);
            \draw[edge, thin] (wi) -- (xi);
            \draw[edge, red] (xi) -- (yi);
            \draw[edge, red] (yi) -- (zi);
            \draw[edge, thin] (zi) -- (si);
            \draw[thickedge,white] (wi) to[out=45, in=135] (zi);
            \draw[edge,red] (wi) to[out=45, in=135] (zi);
            \end{scope}

            \begin{scope}[xshift=\spc*12cm, yshift=-1.5cm]
            \node[smallvertex, label=below:{\footnotesize $u_{j}$}] (uk) at (\spc,0) {};
            \node[smallvertex, label=below:{\footnotesize $v_{j}$}] (vk) at (\spc*2,0) {};
            \node[smallvertex, label=below:{\footnotesize $w_{j}$}] (wk) at (\spc*3,0) {};
            \node[smallvertex, label=below:{\footnotesize $x_{j}$}] (xk) at (\spc*4,0) {};
            \node[smallvertex, label=below:{\footnotesize $y_{j}$}] (yk) at (\spc*5,0) {};
            \node[smallvertex, label=below:{\footnotesize $z_{j}$}] (zk) at (\spc*6,0) {};
            \node[smallvertex, label=below:{\footnotesize $s_{j}$}] (sk) at (\spc*7,0) {};

            \draw[thickedge, white] (xi) to[out=45, in=135] (vk);
            \draw[edge, red] (xi) to[out=45, in=135] (vk);
            \draw[edge, lightblue] (si) to (uk);
            \draw[edge, thin] (uk) to (vk);
            \draw[edge, red] (vk) -- (wk);
            \draw[edge, red] (wk) -- (xk);
            \draw[edge, thin] (xk) -- (yk);
            \draw[edge, lightblue] (yk) -- (zk);
            \draw[edge, lightblue] (zk) -- (sk);

            \node[smallvertex, label=below:{\footnotesize $u_{k}$}] (ul) at (\spc*8,0) {};
            \node[smallvertex, label=below:{\footnotesize $v_{k}$}] (vl) at (\spc*9,0) {};
            \node[smallvertex, label=below:{\footnotesize $w_{k}$}] (wl) at (\spc*10,0) {};
            \node[smallvertex, label=below:{\footnotesize $x_{k}$}] (xl) at (\spc*11,0) {};
            \node[smallvertex, label=below:{\footnotesize $y_{k}$}] (yl) at (\spc*12,0) {};
            \node[smallvertex, label=below:{\footnotesize $z_{k}$}] (zl) at (\spc*13,0) {};
            \node[smallvertex, label=below:{\footnotesize $s_{k}$}] (sl) at (\spc*14,0) {};
            
            \draw[edge, lightblue] (sk) to (ul);
            \draw[edge, thin] (ul) to (vl);
            \draw[edge, red] (vl) -- (wl);
            \draw[edge, red] (wl) -- (xl);
            \draw[edge, thin] (xl) -- (yl);
            \draw[edge, lightblue] (yl) -- (zl);
            \draw[edge, lightblue] (zl) -- (sl);
            \draw[edge, thin] (wl) to[out=45, in=135] (zl);
            \draw[edge, lightblue, dotted] (sl) -- (\spc*14.7,0);
            
            \draw[edge, thin] (wk) to[out=45, in=135] (zk);
            \draw[thickedge, white] (xk) to[out=45, in=135] (vl);
            \draw[edge, red] (xk) to[out=45, in=135] (vl);
            \draw[thickedge, white] (xl) to[out=45, in=135] (\spc*16,0);
            \draw[edge, red, dotted] (xl) to[out=45, in=135] (\spc*16,0);
            \draw[thickedge, white] (ul) to[out=45, in=135] (yl);
            \draw[edge, lightblue] (ul) to[out=45, in=135] (yl);
            \draw[thickedge, white] (uk) to[out=45, in=135] (yk);
            \draw[edge, lightblue] (uk) to[out=45, in=135] (yk);
            \draw[thickedge, white] (\spc*14.7,0) to[out=45, in=180] (\spc*16,0.8*\spc) ;
            \draw[edge, lightblue, dotted] (\spc*14.7,0) to[out=45, in=180] (\spc*16,0.8*\spc) ;
            
            \end{scope}

            \begin{scope}[xshift=\spc*16cm, yshift=-3cm]
            \node[smallvertex, label=below:{\footnotesize $u_{l}$}] (um) at (\spc,0) {};
            \node[smallvertex, label=below:{\footnotesize $v_{l}$}] (vm) at (\spc*2,0) {};
            \node[smallvertex, label=below:{\footnotesize $w_{l}$}] (wm) at (\spc*3,0) {};
            \node[smallvertex, label=below:{\footnotesize $x_{l}$}] (xm) at (\spc*4,0) {};
            \node[smallvertex, label=below:{\footnotesize $y_{l}$}] (ym) at (\spc*5,0) {};
            \node[smallvertex, label=below:{\footnotesize $z_{l}$}] (zm) at (\spc*6,0) {};
            \node[smallvertex, label=below:{\footnotesize $s_{l}$}] (sm) at (\spc*7,0) {};

            \draw[edge, red, dotted] (-3*\spc,0) to[out=45, in=135] (vm);
            \draw[edge, lightblue, dotted] (-2*\spc,0) to (um);
            \draw[thickedge, white] (-2*\spc,0) to[out=135, in=0] (-3*\spc,0.7*\spc);
            \draw[edge, lightblue, dotted] (-2*\spc,0) to[out=135, in=0] (-3*\spc,0.7*\spc);

            \draw[edge, thin] (um) to (vm);
            \draw[edge, red] (vm) -- (wm);
            \draw[edge, red] (wm) -- (xm);
            \draw[edge, thin] (xm) -- (ym);
            \draw[edge, lightblue] (ym) -- (zm);
            \draw[edge, lightblue] (zm) -- (sm);
            \draw[edge, thin] (wm) to[out=45, in=135] (zm);

            \node[smallvertex, label=below:{\footnotesize $u_{n}$}] (un) at (\spc*8,0) {};
            \node[smallvertex, label=below:{\footnotesize $v_{n}$}] (vn) at (\spc*9,0) {};
            \node[smallvertex, label=below:{\footnotesize $w_{n}$}] (wn) at (\spc*10,0) {};
            \node[smallvertex, label=below:{\footnotesize $x_{n}$}] (xn) at (\spc*11,0) {};
            \node[smallvertex, label=below:{\footnotesize $y_{n}$}] (yn) at (\spc*12,0) {};
            \node[smallvertex, label=below:{\footnotesize $z_{n}$}] (zn) at (\spc*13,0) {};
            \node[smallvertex, label=below:{\footnotesize $s_{n}$}] (sn) at (\spc*14,0) {};
            
            \draw[thickedge, white] (xm) to[out=45, in=135] (vn);
            \draw[edge, red] (xm) to[out=45, in=135] (vn);
            \draw[edge, lightblue] (sm) to (un);
            \draw[edge, thin] (un) to (vn);
            \draw[edge, red] (vn) -- (wn);
            \draw[edge, red] (wn) -- (xn);
            \draw[edge, thin] (xn) -- (yn);
            \draw[edge, lightblue] (yn) -- (zn);
            \draw[edge, lightblue] (zn) -- (sn);

            \draw[edge, thin] (wn) to[out=45, in=135] (zn);
            \draw[thickedge, white] (xn) to[out=45, in=135] (sn);
            \draw[edge, red] (xn) to[out=45, in=135] (sn);
            
            \draw[thickedge, white] (um) to[out=45, in=135] (ym);
            \draw[edge, lightblue] (um) to[out=45, in=135] (ym);
            \draw[thickedge, white] (un) to[out=45, in=135] (yn);
            \draw[edge, lightblue] (un) to[out=45, in=135] (yn);
            \end{scope}
        \end{tikzpicture}
        \caption{An illustration of $W_n$. 
        A Hamiltonian path of $W_n$ going from $s_1$ to $s_i$ is highlighted in red (from $s_1$ to $s_n$) and blue (from $s_n$ to $s_i$). For the sake of better readability of indices, in the figure we let $h=i-1$, $j = i+1$, $k=i+2$, and $l=n-1$.}
        \label{fig:blow_up_gadget}
    \end{figure}
    
    For every $j\in [2,n]$ let $P_j$ be the path $(s_{j-1},u_j,v_j,w_j,x_j,y_j,z_j,s_j)$, and let $Q_j$ be the concatenation $Q_j = P_2 \cdot P_3 \cdot \ldots \cdot P_j$,
    so $Q_j$ is a path from $s_1$ to $s_i$ that spans 
    \[
    \{s_1\} \cup \{u_k,v_k,w_k,x_k,y_k,z_k,s_k : 2\leq k \leq j \}.
    \]
    Note that, if $i=n$ then $Q_n$ is the desired Hamiltonian path. Henceforth, assume $i\in [2,n-1]$. Let $R$ be the path going from $x_i$ to $x_n$ consisting of the concatenation
    \[
        (x_i,v_{i+1},w_{i+1},x_{i+1}) \cdot (x_{i+1},v_{i+2},w_{i+2},x_{i+2}) \cdot \ldots \cdot (x_{n-1}, v_n, w_n, x_n)
    \]
    and $R'$ be the path going from $s_n$ to $s_i$ consisting of the concatenation
    \[
        (s_n,z_n,y_n,u_n,s_{n-1}) \cdot (s_{n-1},z_{n-1},y_{n-1},u_{n-1},s_{n-2}) \cdot \ldots \cdot (s_{i+1},z_{i+1},y_{i+1},u_{i+1},s_{i}).
    \]
    Now, 
    \[
        Q_{i-1} \cdot (s_{i-1}, u_i,v_i,w_i,z_i,y_i,x_i) \cdot R \cdot (x_n,s_n) \cdot R'
    \]
    is a Hamiltonian path of $W_n$ going from $s_1$ to $s_n$. The lemma follows.
\end{proof}

We are now ready to derive \cref{thm:always_path_maxdegree_3}.

\begin{proof}[Proof of \cref{thm:always_path_maxdegree_3}]
    We assume that $n'\geq 2$.
    Let $H=(H_i)_{i\in [t]}$, $t=\frac{1}{2}n'(n'+1)-1$, be a temporal digraph on $n'+1$ vertices with vertex set $\{v_0,\dots,v_{n'}\}$ as in Lemma~\ref{lemma:lower_bound_path_apex}. Let $W_{n'}$ and $\{s_1,\dots,s_{n'}\}$ be as in Lemma~\ref{lemma:hamiltonian_subcubic_graph}. For every $j\in [2,n']$, let $P_j$ be a Hamiltonian path of $W_{n'}$ with extremities $\{s_1,s_j\}$. We let $\vec{P_j}$ be the orientation of $P_j$ corresponding to a directed path going from $s_1$ to $s_j$, and $\rev{P_j}$ be the opposite orientation.
    
    We let $D$ be the temporal digraph with vertex set 
    \[
        V(D) = (V(H) \setminus v_0) \cup V(W_{{n'}}), 
    \]
    where each snapshot $D_i$, $1\leq i\leq \frac{1}{2}n'(n'+1)$ is obtained from $H_i - v_0$ by adding the arcs of the directed path:
    \[
    \begin{cases}
        \rev{P_n} \cdot (s_1v_1) &\text{if $\InN[H_i]{v_0} = \emptyset$ and $\OutN[H_i]{v_0}= \{v_1\}$;}\\[0em]
        (v_1,s_1) \cdot \vec{P_n} \cdot &\text{if $\InN[H_i]{v_0} = \{v_1\}$ and $\OutN[H_i]{v_0}= \emptyset$;}\\[0em]
        (v_j,s_j) \cdot \rev{P_j} \cdot (s_1v_1) &\text{if $\InN[H_i]{v_0} = \{v_j\}$ and $\OutN[H_i]{v_0}= \{v_1\}$; or}\\[0em]
        (v_1,s_1) \cdot \vec{P_j} \cdot (s_jv_j) &\text{if $\InN[H_i]{v_0} = \{v_1\}$ and $\OutN[H_i]{v_0}= \{v_j\}$.}
    \end{cases}
    \]
    Note that one of the four cases above holds by choice of $H$. 
    
    By construction, the underlying graph of $D$ is included in the graph obtained from $W_{n} \cup (\UnderlyingGraph{H}-v_0)$ by adding the matching $\{s_1v_1,s_2v_2,\ldots,s_{n'}v_{n'}\}$. Hence, $\UnderlyingGraph{D}$ has maximum degree $3$, and every snapshot of $D$ is a directed path. Moreover, $D$ does not admit any temporal exploration, as any temporal exploration of $D$ would yield a temporal exploration of $H$ (by replacing any occurrence of a vertex $u\in V(W_n)$ with $v_0$).
    Note that $D$ has $n=8n'-6$ vertices, and that 
    \[
        \ell(D') = \tfrac{1}{2}n'(n'+1) \geq \tfrac{1}{128} n^2,
    \]
    as desired. The result follows.
\end{proof}

We note that the condition $\Delta(\UnderlyingGraph{D}) \leq 3$ of Theorem~\ref{thm:always_path_maxdegree_3} is best possible.

\begin{proposition}
    \label{prop:underlying_cycle}
    Let $D$ be an always-unilateral temporal digraph of order $n\geq 2$ with $\Delta(\UnderlyingGraph{D})\leq 2$. If $\ell(D)\geq 2n-3$ then $D$ admits a temporal exploration.
\end{proposition}
\begin{proof}
    The result is trivial for $n=2$, so let us assume $n\geq 3$. Since every snapshot of $D$ is unilateral, $\UnderlyingGraph{D}$ must be connected, so there exists a labeling $v_0,\dots,v_{n-1}$ of $V(D)$ such that every snapshot is an orientation of a subgraph of the cycle $(v_0,\dots,v_{n-1},v_0)$.

    Note that each snapshot of $D$ contains a Hamiltonian directed path. Up to removing an edge, we thus assume that each snapshot of $D$ consists exactly in a Hamiltonian directed path. 
    By directional duality, we assume that $n-1$ snapshots of $D$ consist of a forward directed path, that is, consist of the directed path 
    \[
    Q_i = (v_i,v_{i+1},\ldots,v_{n-1},v_0,\ldots,v_{i-1})
    \]
    for some $i\in [0,n-1]$. We show that the temporal digraph $H=(H_i)_{i\in [n-1]}$ consisting exactly of $n-1$ such snapshots has a temporal exploration, hence implying that $D$ admits one as well.

    For $i\in [0,n-1]$, let $(u_1^i,u_2^i,\ldots,u_n^i)$ be the labeling of $V(D)$ along $Q_i$. 
    
    Note that, for every $j\in [n]$, and any pair of distinct integers $0\leq i< i'\leq n-1$, we have $u_j^i \neq u_j^{i'}$.
    In particular, for every $j\in [2,n]$, since each snapshot of $H$ consists exactly of the directed cycle $(v_0,v_1,\ldots,v_{n-1},v_0)$ minus one arc, there exists exactly one integer $i\in [0,n-1]$ such that
    \[
        u_{j-1}^iu_{j}^i \notin A(H_j).
    \]
    Hence, by the Pigeonhole Principle, there exists $i\in [0,n-1]$ such that, for every $j\in [2,n]$,
    \[
        u_{j-1}^iu_{j}^i \in A(H_j).
    \]
    It follows that $Q_i$ is a temporal exploration of $H$. The result follows.
\end{proof}
However, we do not know whether Theorem~\ref{thm:always_path_maxdegree_3} can be strengthened to always-strong digraphs.

\section{Semicomplete temporal digraphs}
\label{sec:always-semicomplete exploration}

We now turn our focus to {\it semicomplete temporal digraphs}, that is, temporal digraphs in which every snapshot is a semicomplete digraph. 
We also consider the particular case of {\it temporal tournaments}, that is, temporal digraphs where every snapshot is a tournament.
Note that semicomplete temporal digraphs are always-unilateral, but not necessarily always-strong.

In Section~\ref{subsec:general_upperbound_semicomplete}, we prove that every semicomplete temporal digraph $D$ of order $n\geq 2$ with lifetime exceeding $\frac{4}{3}n-2$ admits a temporal exploration. We further show that this is optimal by exhibiting, for every $n\geq 2$, a temporal tournament $D$ with $\Lifetime{D}=\lfloor\frac{4}{3}n-2\rfloor$ that does not admit any temporal exploration.

In Section~\ref{subsec:np_hardness_semicomplete}, we consider the computational complexity of the problem consisting of deciding whether a temporal digraph admits a temporal exploration. 
This problem is already known to be \NP/-hard in general~\cite{michailTCS634}. We prove that it remains \NP/-hard when restrained to temporal tournaments.
We finally consider the approximability of the corresponding optimization problem in Section~\ref{subsec:inapprox_semicomplete}.

\subsection{Minimum lifetime guaranteeing a temporal exploration}
\label{subsec:general_upperbound_semicomplete}

The purpose of this section is to prove Theorem~\ref{theorem:tournament-exploration-upper-bound} below, that is, to show that every semicomplete temporal digraph $D$ of order $n\geq 2$ with lifetime $\lfloor \frac{4}{3}n-1\rfloor$ admits a temporal exploration. We then prove Theorem~\ref{statement:tournament-exploration-lower-bound}, implying that this lower bound is optimal.

Similar to our proof of \cref{statement:always-unilateral walk covering S}, we construct a temporal exploration by successively connecting special non-spanned vertices which are able to reach any other non-spanned vertex in the next ``few'' snapshots.
The main difference to the general always-unilateral setting is that ``few'' can be taken to be $2$ here, as shown in the next lemma,
which generalizes to the temporal setting the celebrated property that all tournaments admit a {\it king} (see for instance~\cite[Proposition~1]{havetJGT35}).

Along this section, given two vertices $u$ and $v$ of a temporal digraph $D=(D_i)_{i\in [t]}$ and an integer $i\in [t]$, we denote by $u \to_i v$ the property $(u,v)\in A(D_i)$, and by $u \not\to_i v$ the property that $(u,v)\notin A(D_i)$.

\begin{lemma}
	\label{lemma:directed_king}
	Let \(D=(D_1,D_2)\) be a semicomplete temporal digraph.
	There exists \(v \in \V{D}\)
	such that \(\V{D} \subseteq \TemporalReach{2}{v}\). Moreover, if $v$ is the only such vertex, then $v$ is a source in both $D_1$ and $D_2$.
\end{lemma}
\begin{proof}
    For every vertex $v\in V(D)$, let us denote by $L(v) = V(D) \setminus \TemporalReach[D]{}{v}$ the set of vertices that $v$ cannot reach in $D$. Since $D$ is a semicomplete temporal digraph, note that $L(v)$ contains all vertices $u\in V(D)$ such that:
    \begin{itemize}
        \item $u \to_1 v$,
        \item $u \to_2 v$, and
        \item for every vertex $w$, if $v \to_1 w$, then $u \to_2 w$.
    \end{itemize}
    By definition, observe that for any pair of distinct vertices $\{u,v\}$, if $u\in L(v)$, then $v \notin L(u)$.
    Let $H$ be the auxiliary digraph on vertex set $V(D)$ with arc set 
    \[
        A(H) = \{(u,v) : u\in L(v)\}.
    \]
    By the remark above, $H$ does not contain any pair of opposite arcs. 
    We further have that $H$ is acyclic.
    \begin{claim}
        The oriented graph $H$ is acyclic.
    \end{claim}
    \begin{proofclaim}
        Assume for a contradiction that $H$ contains a directed cycle 
        \[
        C=(v_0,v_1,\dots,v_{r},v_0)
        \]
        for some $r\geq 2$. Let $0\leq i,j \leq r$ be two distinct integers such that $v_i \to_1 v_j$ and, with respect to this property, $(i-j) \bmod (r+1)$ is minimum. 
        By relabeling the vertices along $C$, we may assume without loss of generality that $i=0$.
    
        Note that, by construction of $C$, for every $0\leq k\leq r$ we have that $v_{k} \in L(v_{k+1})$, which by definition implies that $v_{k}\to_{1} v_{k+1}$ and $v_{k}\to_{2} v_{k+1}$ (with indices taken modulo $r+1$). In particular, we have $j\neq r$, and we now distinguish two possible cases.
    
        Assume first that $j=r-1$. Then, by definition of $i,j$ we have $v_0 \to_1 v_{r-1}$. Similarly, since $v_{r-1} \in L(v_{r})$, we have $v_{r-1} \to_2 v_r$. Therefore, $v_0\to_1 v_{r-1} \to_2 v_{r}$, hence $v_r$ is reachable from $v_0$ in $D$, a contradiction.
    
        Assume then that $1\leq j\leq r-2$. By choice of the arc $v_0v_j$, observe that $v_{j+1} \to_1 v_{r}$. Since $v_{j} \in  L(v_{j+1})$, it follows that $v_j \to_2 v_{r}$. Therefore, $v_0 \to_1 v_j \to_2 v_{r}$, which contradicts $v_{r} \in L(v_0)$. The claim follows.
    \end{proofclaim}

    It follows from the claim above that $H$ admits at least one source $v$. By construction, $L(v) = \InN[H]{v} = \emptyset$, which means that $v$ can temporally reach every other vertex.
    This shows the first part of the statement.

    As for the second part, assume for a contradiction that $v$ is the unique such vertex (that is, $v$ is the unique source of $H$) and that $v$ is not a source in at least one of $(D_1,D_2)$. Let $u\in V(D)$ be such that $u\to_1 v$ or $u \to_2 v$. Since $H$ is acyclic and $v$ is its unique source, there exists in $H$ a directed path 
		\(
        P=(v_0, v_1,\dots, v_r)
		\)
    going from $v=v_0$ to $u = v_r$. Let $0\leq i<j\leq r$ be two indices such that $v_j \to_1 v_i$ or $v_j \to_2 v_i$ and, with respect to this property, $j-i$ is minimum. The existence of $i,j$ is guaranteed as $u\to_1 v$ or $u \to_2 v$ by assumption. Moreover, by construction of $P$, $j\geq i+2$, as for every $i\in [0,r-1]$ we have $v_{i} \in L(v_{i+1})$, which in particular implies $v_i \to_1 v_{i+1}$ and $v_i \to_2 v_{i+1}$.

    Now, if $v_j \to_1 v_i$ then by choice of $i,j$ we have $v_j\to_1 v_i \to_2 v_{j-1}$, a contradiction, as $v_{j-1} \in L(v_j)$.
    Otherwise, $v_j \to_2 v_i$, and by choice of $i,j$ we have $v_{i+1} \to_1 v_j \to_2 v_i$, a contradiction, as $v_{i} \in L(v_{i+1})$. The statement follows.
\end{proof}

We further make use of the following straightforward observation. 

\begin{lemma}
	\label{lemma:exploration_K3}
	Let $D=(D_1,D_2)$ be a semicomplete temporal digraph of order $3$. If $D$ does not admit a temporal exploration, then both $D_1$ and $D_2$ are isomorphic to $\vec{C_3}$ and $A(D_1) \cap A(D_2) = \emptyset$.
\end{lemma}
\begin{proof}
    Let us denote $V(D) = \{x,y,z\}$.
    Assume first that $D_1$ is not isomorphic to $\vec{C_3}$. Since $D_1$ is semicomplete, in particular it contains a vertex, say, $x$, with out-degree $2$.
    Since $D_2$ is semicomplete, there is at least one arc between $z$ and $y$ in $D_2$, say, $yz\in A(D_2)$ without loss of generality. Then $(x,y,z)$ is a temporal exploration of $D$.

    Symmetrically, if $D_2$ is not isomorphic to $\vec{C_3}$, then it contains a vertex, say, $x$, with in-degree $2$, and there is an arc between $y$ and $z$ in $D_1$, say, $zy\in A(D_1)$. Then $(z,y,x)$ is a temporal exploration of $D$.

    Assume now that both $D_1$ and $D_2$ are isomorphic to $\vec{C_3}$. Without loss of generality, assume that $D_1$ has arc set $\{xy,yz,zx\}$. Then either $A(D_2) = A(D_1)= \{xy,yz,zx\}$ or $A(D_2) = \{xz,zy,yx\}$. In the former case, $(x,y,z)$ is a temporal exploration of $D$, and in the latter case $A(D_1) \cap A(D_2) = \emptyset$. The statement follows.
\end{proof}

We are now ready to prove the following key lemma, from which we easily derive Theorem~\ref{theorem:tournament-exploration-upper-bound}.

\begin{lemma}
    \label{lemma:tournament-exploration-upper-bound}
    Let $D=(D_i)_{i\in [4]}$ be a semicomplete temporal digraph of order $n\geq 3$. There exists a temporal path $P$ of length $2$ with terminal vertex $v^\star$ such that one of the following holds:
    \begin{alphaenumerate}
        \item $P$ is a temporal path of $\TemporalSlice{D}{1}{2}$ and $v^\star$ can reach every vertex of $V(D) \setminus V(P)$ in $\TemporalSlice{D}{3}{4}$; or
        \label{enum:tournament-exploration-upper-bound:1}
        \item $P$ is a temporal path of $\TemporalSlice{D}{1}{3}$ and $v^\star$ dominates $V(D) \setminus V(P)$ in $D_4$.
        \label{enum:tournament-exploration-upper-bound:2}
    \end{alphaenumerate}
\end{lemma}
\begin{proof}
    By \cref{lemma:directed_king}, there exists a vertex $v$ which can temporally reach every vertex of $V$ in $\TemporalSlice{D}{3}{4}$. Among all such vertices we pick $v$ with largest out-degree in $D_{3}$. 
    By applying \cref{lemma:directed_king} twice, we let $u$ and $w$ be two distinct vertices of $D$ such that both $u$ and $w$ can temporally reach every vertex of $V\setminus \{u,v,w\}$ in $\TemporalSlice{D}{3}{4}$. We let $S = V\setminus \{v,u,w\}$.

    Let $D' = D-S = (D_1',D_2',D_3',D_4')$. 
		If $\TemporalSlice{D'}{1}{2}$ admits a temporal exploration $P$, then $P$ satisfies~\cref{enum:tournament-exploration-upper-bound:1} and we are done. Therefore, we can assume that $\TemporalSlice{D'}{1}{2}$ does not admit such a temporal exploration. By \cref{lemma:exploration_K3} and by symmetry of the roles of $u$ and $w$, we may assume without loss of generality that $A(D_1') = \{vu,uw,wv\}$ and $A(D_2') = \{vw, wu, uv\}$.

    Assume first that, for some $s\in S$, $s\to_1 u$. Then $s\to_1u\to_2v$ yields a temporal path satisfying~\cref{enum:tournament-exploration-upper-bound:1}. Henceforth, we thus assume the $u \to_1 s$ for every $s\in S$.
    Similarly, if for some $s\in S$, $s\to_2 v$, then $u\to_1s\to_{2}v$ yields a temporal path satisfying~\cref{enum:tournament-exploration-upper-bound:1}. Therefore, we further assume that $v\to_2 s$ for every $s\in S$.

    It follows from the observation above that, for every $s\in S$, there exists a temporal path $w\to_1 v\to_2 s$. Therefore, if there exists $s\in S$  that can temporally reach every vertex in $V\setminus \{v,w\}$ in $\TemporalSlice{D}{3}{4}$, then $w\to_1 v\to_2 s$ is a temporal path satisfying~\cref{enum:tournament-exploration-upper-bound:1}. We may thus assume that $u$ is the unique vertex that can temporally reach every vertex of $V\setminus \{v,w\}$ in $\TemporalSlice{D}{3}{4}$. By the second part of \cref{lemma:directed_king}, it follows that $u\to_3 s$ and $u\to_4 s$ for every vertex $s\in S$.

    Assume now that $v \to_3 u$. Then $w\to_{1}v\to_{3}u$ yields a temporal path of $D$ satisfying~\cref{enum:tournament-exploration-upper-bound:2}, as $u$ dominates $S$ in $D_4$. Similarly, if $w\to_3 u$
    then $v\to_{2}w\to_{3}u$ yields a temporal path of $D$ satisfying~\cref{enum:tournament-exploration-upper-bound:2}.
    
    Therefore, we may assume that $u\to_3 v$, $u\to_3 w$, and $v \not\to_3 u$. Since $u$ dominates $S$ in $D_3$, we have that $u$ dominates $V\setminus \{u\}$ in $D_3$. In particular, $u$ can temporally reach every vertex in $V$ in $\TemporalSlice{D}{3}{4}$. Moreover, 
    \[
        |\OutN[D_3]{u}| = |V(D)|-1 > |\OutN[D_3]{v}|,
    \]
    where the last inequality follows from $v \not\to_3 u$.
    This contradicts the choice of $v$, and concludes the proof of the lemma.
\end{proof}
We are now ready to prove Theorem~\ref{theorem:tournament-exploration-upper-bound}.

\hypertarget{theorem:tournament-exploration-upper-bound:body}{}
\semicompleteExploration*

\begin{proof}
    We proceed by induction on $n$. The result is trivial when $n=2$, and it follows easily from \cref{lemma:exploration_K3,lemma:tournament-exploration-upper-bound} when $n\in \{3,4\}$. We thus assume $n\geq 5$. 
    
		Let $D=(D_i)_{i ∈ [t]}$.
		 Let $P$ be a temporal path of $(D_1,D_2,D_3,D_4)$ with terminal vertex $v^\star$ guaranteed by Lemma~\ref{lemma:tournament-exploration-upper-bound}, and let $D' = \TemporalSlice{D}{5}{t} - \V{P}$. Note that $t-4 \geq \lfloor \frac{4}{3}(n-3)-1 \rfloor = \lfloor \frac{4}{3}|V(D')|-1 \rfloor$ and that $n-3\geq 2$. Hence, by induction, $D'$ admits a temporal exploration $Q$. We let $u$ be the initial vertex of $Q$. If $P$ satisfies~\cref{enum:tournament-exploration-upper-bound:2}, then $v^\star$ dominates $u$ in $D_4$, and the concatenation of $P$ and $Q$ is a temporal exploration of $D$. Otherwise, $P$ satisfies~\cref{enum:tournament-exploration-upper-bound:1}, and there exists a temporal path $P'$ from $v^\star$ to $u$ in $\TemporalSlice{D}{3}{4}$, and the concatenation of $P$, $P'$, and $Q$ is an exploration of $D$. The result follows.
\end{proof}

We conclude with a proof that Theorem~\ref{theorem:tournament-exploration-upper-bound} is tight, even when all snapshots are tournaments, as shown by the following theorem.

\hypertarget{statement:tournament-exploration-lower-bound:body}{}
\semicompleteExplorationLowerBound*
	\begin{proof}
    Let $n\geq 2$ be an arbitrary integer and let $k,r\in \mathbb{N}$ be such that $n=3k+r$ and $r \leq 2$. Let $t=4k+r-2 = \lfloor \frac{4}{3}n-2\rfloor$.
    
    We define a temporal tournament $D=(D_i)_{i\in [t]}$ as follows, see Figure~\ref{fig:lower_bound_tournament} for an illustration.
		Let $X=\{x_1,\dots,x_r\}$ be a set of $r$ vertices ($X$ is empty if $r=0$) and, for every $j\in [k]$, let $U_j = \{u_j,v_j,w_j\}$. Let $U = \bigcup_{j=1}^k U_j$. We define the vertex set of $D$ as
    $V(D) = X \cup U$,
    so $D$ has order exactly $n$. Then, for every $i\in [t]$ we let $D_i$ contain the arcs of
    \begin{align*}
        \{x_jx_\ell : 1\leq j < \ell\leq r\}
        & \cup \{ x_j u : 1\leq j \leq r \text{~~and~~} u\in U\}\\[0em]
        & \cup \{yz : y\in U_j \text{~~and~~} z\in U_\ell \text{~~for some $j<\ell$}\}.
    \end{align*}
    Then, if $i$ is even, we let $D_i$ contain the arcs
    $\bigcup_{j=1}^k \{u_jv_j, v_jw_j, w_ju_j\}$,
    and, if $i$ is odd, we let $D_i$ contain the arcs 
    $\bigcup_{j=1}^k \{v_ju_j, w_jv_j, u_jw_j\}$. This completes the description of $D$.

    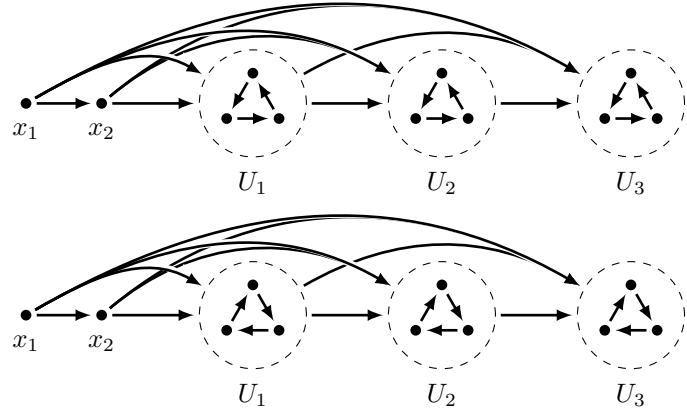
\begin{figure}
        \centering
        \begin{tikzpicture}
        \tikzset{bigvertex/.style = {shape=circle,dashed, draw, outer sep=2pt}}
            \begin{scope}
            \node[vertex, label=below:$x_1$] (x1) at (-0.5,0) {};
            \node[vertex, label=below:$x_2$] (x2) at (0.5,0) {};
            \begin{scope}[xshift=2.5cm]
                \node[vertex] (u1) at (90:0.4) {};
                \node[vertex] (v1) at (-150:0.4) {};
                \node[vertex] (w1) at (-30:0.4) {};
                \node[bigvertex, fit=(u1)(v1)(w1), inner sep=1pt, label=below:$U_1$] (U1) at (0,0) {};
            \end{scope}
            \begin{scope}[xshift=5cm]
                \node[vertex] (u2) at (90:0.4) {};
                \node[vertex] (v2) at (-150:0.4) {};
                \node[vertex] (w2) at (-30:0.4) {};
                \node[bigvertex, fit=(u1)(v1)(w1), inner sep=1pt, label=below:$U_2$] (U2) at (0,0) {};
            \end{scope}
            \begin{scope}[xshift=7.5cm]
                \node[vertex] (u3) at (90:0.4) {};
                \node[vertex] (v3) at (-150:0.4) {};
                \node[vertex] (w3) at (-30:0.4) {};
                \node[bigvertex, fit=(u1)(v1)(w1), inner sep=1pt, label=below:$U_3$] (U3) at (0,0) {};
            \end{scope}

            \draw[directededge] (U1) to (U2);
            \draw[directededge] (U2) to (U3);
            \draw[directededge] (U1) to[out=30, in=150] (U3);
            
            \draw[directededge] (x1) to (x2);
            \draw[directededge] (x2) to (U1);
            \draw[directededgeblank] (x2) to[out=40, in=150] (U2);
            \draw[directededgeblank] (x2) to[out=40, in=150] (U3);
            \draw[directededge] (x2) to[out=40, in=150] (U2);
            \draw[directededge] (x2) to[out=40, in=150] (U3);
            
            \draw[directededgeblank] (x1) to[out=30, in=150] (U1);
            \draw[directededgeblank] (x1) to[out=30, in=150] (U2.150);
            \draw[directededgeblank] (x1) to[out=30, in=150] (U3.150);
            \draw[directededge] (x1) to[out=30, in=150] (U1);
            \draw[directededge] (x1) to[out=30, in=150] (U2.150);
            \draw[directededge] (x1) to[out=30, in=150] (U3.150);

            \draw[directededge] (u1) -- (v1);
            \draw[directededge] (v1) -- (w1);
            \draw[directededge] (w1) -- (u1);
            
            \draw[directededge] (u2) -- (v2);
            \draw[directededge] (v2) -- (w2);
            \draw[directededge] (w2) -- (u2);
            
            \draw[directededge] (u3) -- (v3);
            \draw[directededge] (v3) -- (w3);
            \draw[directededge] (w3) -- (u3);
            \end{scope}

            \begin{scope}[yshift=-2.8cm]
            \node[vertex, label=below:$x_1$] (x1) at (-0.5,0) {};
            \node[vertex, label=below:$x_2$] (x2) at (0.5,0) {};
            \begin{scope}[xshift=2.5cm]
                \node[vertex] (u1) at (90:0.4) {};
                \node[vertex] (w1) at (-150:0.4) {};
                \node[vertex] (v1) at (-30:0.4) {};
                \node[bigvertex, fit=(u1)(v1)(w1), inner sep=1pt, label=below:$U_1$] (U1) at (0,0) {};
            \end{scope}
            \begin{scope}[xshift=5cm]
                \node[vertex] (u2) at (90:0.4) {};
                \node[vertex] (w2) at (-150:0.4) {};
                \node[vertex] (v2) at (-30:0.4) {};
                \node[bigvertex, fit=(u1)(v1)(w1), inner sep=1pt, label=below:$U_2$] (U2) at (0,0) {};
            \end{scope}
            \begin{scope}[xshift=7.5cm]
                \node[vertex] (u3) at (90:0.4) {};
                \node[vertex] (w3) at (-150:0.4) {};
                \node[vertex] (v3) at (-30:0.4) {};
                \node[bigvertex, fit=(u1)(v1)(w1), inner sep=1pt, label=below:$U_3$] (U3) at (0,0) {};
            \end{scope}

            \draw[directededge] (U1) to (U2);
            \draw[directededge] (U2) to (U3);
            \draw[directededge] (U1) to[out=30, in=150] (U3);
            
            \draw[directededge] (x1) to (x2);
            \draw[directededge] (x2) to (U1);
            \draw[directededgeblank] (x2) to[out=40, in=150] (U2);
            \draw[directededgeblank] (x2) to[out=40, in=150] (U3);
            \draw[directededge] (x2) to[out=40, in=150] (U2);
            \draw[directededge] (x2) to[out=40, in=150] (U3);
            
            \draw[directededgeblank] (x1) to[out=30, in=150] (U1);
            \draw[directededgeblank] (x1) to[out=30, in=150] (U2.150);
            \draw[directededgeblank] (x1) to[out=30, in=150] (U3.150);
            \draw[directededge] (x1) to[out=30, in=150] (U1);
            \draw[directededge] (x1) to[out=30, in=150] (U2.150);
            \draw[directededge] (x1) to[out=30, in=150] (U3.150);

            \draw[directededge] (u1) -- (v1);
            \draw[directededge] (v1) -- (w1);
            \draw[directededge] (w1) -- (u1);
            
            \draw[directededge] (u2) -- (v2);
            \draw[directededge] (v2) -- (w2);
            \draw[directededge] (w2) -- (u2);
            
            \draw[directededge] (u3) -- (v3);
            \draw[directededge] (v3) -- (w3);
            \draw[directededge] (w3) -- (u3);
            \end{scope}
        \end{tikzpicture}
        \caption{The construction of $D_i$ for $n=11$ with $i$ being even (top) and odd (bottom). Arcs to or from a set \(U_i\) in the picture represent arcs incident to all vertices inside said \(U_i\).}
        \label{fig:lower_bound_tournament}
    \end{figure}
    
    By construction, in each $D_i$, $X$ dominates $V(D) \setminus X$, and for every $j,\ell \in [k]$, $U_j$ dominates $U_\ell$ for $\ell > j$. Therefore, any temporal exploration  explores $X$ first and then the vertices of $U_1,\dots,U_k$ in this order. 
    Observe that at least $|X|-1$ snapshots are required to explore $X$ and that, for every $j$, at least four snapshots are required to explore $U_j$ (unless $j=1$ and $X=\emptyset$, in which case three snapshots are enough). Therefore, any temporal exploration of $D$ requires at least 
		\(
        r + 4k -1 = t+1
		\)
    snapshots. This shows that $D$ does not admit any temporal exploration.
\end{proof}

\subsection{\NP/-hardness}
\label{subsec:np_hardness_semicomplete}

The goal of this section is to show that deciding if a given temporal tournament $D$ admits a temporal exploration is \NP/-hard, even when $\ell(D)=|V(D)|-1$, see Theorem~\ref{thm:np_hardness} below. We formally consider the following problem.

\TemporalExplorationDef

To prove that \TemporalExploration/ remains \NP/-hard in this aforementioned restricted case, we proceed in two steps. We first show that the following problem is \NP/-hard.

\PrescribedHamiltonianPathDef

Our reduction for \PrescribedHamiltonianPath/ is inspired by the famous reduction showing that it is \NP/-hard to decide, given an oriented graph $D$ and two vertices $a,b$ of~$D$, whether $D$ admits an induced path from $a$ to $b$~\cite{Bangjensen2012}. The same general idea was used for numerous \NP/-hardness results on digraphs, see for instance~\cite{Bangjensen2012arcdisjoint,Bangjensen2016,Bangjensen2024}.

\begin{lemma}
    \label{lemma:tournament_Hamiltonian_prescribed_NP-hard}
    \PrescribedHamiltonianPath/ is \NP/-hard even if $\ell_1 = 1$ and, for every $i\in [r-1]$, $\ell_{i+1} \geq \ell_i + 2$. 
\end{lemma}
\begin{proof}
    We reduce from the {\sc $(3,B2)$-SAT} problem, which is the restriction of $3$-SAT in which every literal appears exactly twice (that is, for every variable $x$, both $x$ and $\overline{x}$ appear exactly twice). This problem has been shown to be \NP/-hard by Berman, Karpinski, and Scott~\cite{Berman2004}.

    Let $(X,\mathcal{C})$ be an instance of $(3,B2)$-SAT, where $X=\{x_1,\dots,x_n\}$ is the set of variables and $\mathcal{C}=\{C_1,\dots,C_m\}$ is the set of clauses. 
		The arbitrary ordering of the clauses $C_1,C_2,\ldots{},C_m$ induces an ordering of the occurrences of each literal $x_i$ and its negation $\overline{x_i}$ in these. 
    We now build a tournament $T$ from $(X,\mathcal{C})$ as follows, see Figure~\ref{fig:reduction_ham_prescribed} for an illustration.

    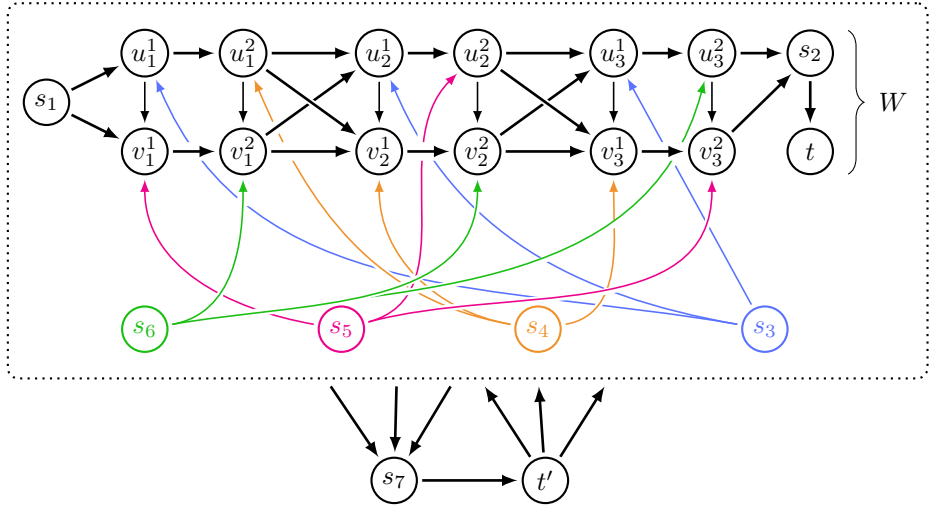
\begin{figure}
    \centering
    \begin{tikzpicture}
    \def\dx{1.3}
    \def\dg{0.5}
    \def\dy{0.65}
    \def\dclause{3}

    \node[labelledvertex] (s0) at (2-\dx,0) {$s_1$};
    \node[labelledvertex] (s1) at ({2+6*\dx+2*\dg},\dy) {$s_2$};
    \node[labelledvertex] (t1) at ({2+6*\dx+2*\dg},-\dy) {$t$};
    \node[labelledvertex] (u11) at (2,\dy) {$u_1^1$};
    \node[labelledvertex] (u12) at ({2+\dx},\dy) {$u_1^2$};
    \node[labelledvertex] (u21) at ({2+2*\dx+\dg},\dy) {$u_2^1$};
    \node[labelledvertex] (u22) at ({2+3*\dx+\dg},\dy) {$u_2^2$};
    \node[labelledvertex] (u31) at ({2+4*\dx+2*\dg},\dy) {$u_3^1$};
    \node[labelledvertex] (u32) at ({2+5*\dx+2*\dg},\dy) {$u_3^2$};
    \node[labelledvertex] (v11) at (2,-\dy) {$v_1^1$};
    
    \node[labelledvertex] (v12) at ({2+\dx},-\dy) {$v_1^2$};
    \node[labelledvertex] (v21) at ({2+2*\dx+\dg},-\dy) {$v_2^1$};
    \node[labelledvertex] (v22) at ({2+3*\dx+\dg},-\dy) {$v_2^2$};
    \node[labelledvertex] (v31) at ({2+4*\dx+2*\dg},-\dy) {$v_3^1$};
    \node[labelledvertex] (v32) at ({2+5*\dx+2*\dg},-\dy) {$v_3^2$};
    \node[labelledvertex, green] (s5) at (2,-\dclause) {$s_6$};
    \node[labelledvertex, magenta] (s4) at ({2+2*\dx},-\dclause) {$s_5$};
    \node[labelledvertex, orange] (s3) at ({2+4*\dx},-\dclause) {$s_4$};
    \node[labelledvertex, blue] (s2) at ({2+6*\dx+0.4},-\dclause) {$s_3$};
    
    \draw[thindirectededge, blue] (s2) to[out=170,in=-67] (u11);
    \draw[thindirectededge, blue] (s2) to[out=170,in=-67] (u21);
    \draw[thindirectededge, blue] (s2) to (u31);
    
    \draw[directededgeblank] (s3) to[out=170, in=-67] (u12);
    \draw[directededgeblank] (s3) to[out=170, in=-90] (v21);
    \draw[directededgeblank] (s3) to[out=10,in=-90] (v31);
    \draw[thindirectededge, orange] (s3) to[out=170, in=-67] (u12);
    \draw[thindirectededge, orange] (s3) to[out=170, in=-90] (v21);
    \draw[thindirectededge, orange] (s3) to[out=10,in=-90] (v31);
    
    \draw[directededgeblank] (s4) to[out=170, in=-90] (v11);
    \draw[directededgeblank] (s4) to[out=15, in=-135] (u22);
    \draw[directededgeblank] (s4) to[out=15, in=-92] (v32);
    \draw[thindirectededge, magenta] (s4) to[out=170, in=-90] (v11);
    \draw[thindirectededge, magenta] (s4) to[out=15, in=-135] (u22);
    \draw[thindirectededge, magenta] (s4) to[out=15, in=-90] (v32);
    
    \draw[directededgeblank] (s5) to[out=10, in=-90] (v12);
    \draw[directededgeblank] (s5) to[out=10, in=-105] (u32);
    
    \draw[thindirectededge, green] (s5) to[out=10, in=-105] (u32);
    \draw[edgeblank] (s5) to[out=10, in=-87] (v22);
    \draw[thindirectededge, green] (s5) to[out=10, in=-90] (v22);
    \draw[thindirectededge, green] (s5) to[out=10, in=-90] (v12);

    \draw[edgeblank, line width=7pt] (v11) to (v12);
    \draw[edgeblank, line width=7pt] (v21) to (v22);
    \draw[edgeblank, line width=7pt] (v31) to (v32);
    \draw[edgeblank, line width=5pt] (v12) to (u21);
    \draw[edgeblank, line width=7pt] (v12) to (v21);
    \draw[directededge] (u11) to (u12);
    \draw[directededge] (v11) to (v12);
    \draw[directededge] (u21) to (u22);
    \draw[directededge] (v21) to (v22);
    \draw[directededge] (u31) to (u32);
    \draw[directededge] (v31) to (v32);
    \draw[directededge] (u12) to (u21);
    \draw[directededge] (v12) to (u21);
    \draw[directededgeblank] (u12) to (v21);
    \draw[directededge] (u12) to (v21);
    \draw[directededge] (v12) to (v21);
    \draw[directededge] (u22) to (u31);
    \draw[directededge] (v22) to (u31);
    \draw[directededge] (v22) to (v31);
    \draw[directededgeblank] (u22) to (v31);
    \draw[directededge] (u22) to (v31);
    \draw[thindirectededge] (u11) to (v11);
    \draw[thindirectededge] (u12) to (v12);
    \draw[thindirectededge] (u21) to (v21);
    \draw[thindirectededge] (u22) to (v22);
    \draw[thindirectededge] (u31) to (v31);
    \draw[thindirectededge] (u32) to (v32);
    \draw[directededge] (s0) to (u11);
    \draw[directededge] (s0) to (v11);
    \draw[directededge] (u32) to (s1);
    \draw[directededge] (v32) to (s1);
    \draw[directededge] (s1) to (t1);

    \draw[decorate, decoration={brace, amplitude=6pt}]
    (11.3, 0.95) -- (11.3,-0.95)
    node[midway, right=8pt] (W) {$W$};

    \node[draw, thick, dotted, rounded corners, inner xsep=4pt, inner ysep=8pt, rectangle, fit=(s0)(u11)(s2) (W), outer sep=3pt] (box) {};

    \node[labelledvertex] (s6) at ($ (s0)!0.5!(W) + (-1,-5) $) {$s_7$};
    \node[labelledvertex] (t2) at ($ (s0)!0.5!(W) + (1,-5) $) {$t'$};

    \draw[directededge] (s6) to (t2);
    \draw[directededge] (box.-95) to (s6);
    \draw[directededge] (box.-110) to (s6);
    \draw[directededge] (box.-125) to (s6);
    \draw[directededge] (t2) to (box.-85);
    \draw[directededge] (t2) to (box.-70);
    \draw[directededge] (t2) to (box.-55);
    
    \end{tikzpicture}
    \caption{The tournament $T$ when $X= \{x_1,x_2,x_3\}$ and $\mathcal{C}= (C_1,C_2,C_3,C_4)$, with {\color{blue}$C_1= (x_1,x_2,x_3)$}, {\color{orange}$C_2=(x_1,\overline{x_2},\overline{x_3})$}, {\color{magenta}$C_3= (\overline{x_1},x_2,\overline{x_3})$}, and {\color{green}$C_4=(\overline{x_1},\overline{x_2},x_3)$}. 
		Every missing edge goes from right to left if it has both extremities in $W$ or in $\{s_2,s_3,s_4,s_5,s_6\}$, and from top to bottom if it has one extremity in $W$ and one in $\{s_2,s_3,s_4,s_5,s_6\}$.}
    \label{fig:reduction_ham_prescribed}
    \end{figure}

    We first describe the part of $T$ corresponding to the variables of $X$. For this, we let $T$ contain three vertices $\{s_1,s_2,t\}$ and, for every $i\in [n]$, the four vertices 
    \[
        X_i=\{u_i^1, u_i^2, v_i^1,v_i^2\},
    \]
    as well as the arcs of
    \[
        \{u_i^1u_i^2, v_i^1v_i^2, u_i^1v_i^1, u_i^2v_i^2, u_i^2v_i^1, v_i^2u_i^1\}.
    \]
    Intuitively speaking, the vertices $\{u_i^1,u_i^2\}$ correspond to the two occurrences of $x_i$, and the vertices $\{v_i^1,v_i^2\}$ to the two occurrences of $\overline{x_i}$.
    Then, for every $i\in [2,n]$, we add the four possible arcs from $\{u_{i-1}^2,v_{i-1}^2\}$ to $\{u_{i}^1,v_{i}^1\}$. All the other arcs go from larger to smaller indices, that is, $T$ contains the arcs of
    \begin{align*}
         \bigcup_{1\leq i \leq j-2\leq n-2} \Big\{yz: y\in X_j \text{~~and~~} z\in X_i \Big\}
        & \cup \bigcup_{2\leq i \leq n} \Big\{yz: y\in \{u_i^2,v_i^2\} \text{~~and~~} z\in X_{i-1}\Big\}\\[0em]
        & \cup \bigcup_{2\leq i \leq n} \Big\{yz: y\in X_i \text{~~and~~} z\in \{u_{i-1}^1,v_{i-1}^1\}\Big\}.
    \end{align*}

    Let $W=\{s_1,s_2,t\} \cup \bigcup_{i\in [n]}X_i$.
    Among $W$, we let the out-neighborhood of $s_1$ be precisely $\{u_1^1,v_1^2\}$, in-neighborhood of $s_2$ be precisely $\{u_n^2,v_n^2\}$, and the in-neighborhood of $t$ be precisely $\{s_2\}$. This concludes the description of the part of $T$ corresponding to the variables of $X$. By construction, observe that we have the following.

    \begin{claim}
        \label{claim:directed_path}
        Let $P$ be a directed path of $T[W]$ going from $s_1$ to $s_2$ of length $2n+1$. Then, for every $i\in[n]$, we have $ V(P) \cap X_i \in \{\{u_i^1,u_i^2\}, \{v_i^1,v_i^2\}\}$.
    \end{claim}

    We now describe the part of $T$ corresponding to the clauses $C_1,\dots,C_m$. For every $j\in [m]$, we add a new vertex $s_{j+2}$, which intuitively corresponds to the clause $C_j$. 
    Then, for every $j\in [m]$, we proceed as follows. Assume $C_j$ contains variables $x_h$, $x_i$, and $x_k$ (negated or not), then we let the out-neighborhood of $s_{j+2}$ in $W$ be precisely $\{a_h, a_i, a_k\}$, where
    \[
        a_h = 
        \begin{cases}
            u_h^1 &\text{if $C_j$ contains the first occurrence of $x_h$,}\\[0em]
            u_h^2 &\text{if $C_j$ contains the second occurrence of $x_h$,}\\[0em]
            v_h^1 &\text{if $C_j$ contains the first occurrence of $\overline{x_h}$,}\\[0em]
            v_h^2 &\text{otherwise, that is $C_j$ contains the second occurrence of $\overline{x_h}$.}\\[0em]
        \end{cases}
    \]
    and $a_i,a_k$ are defined analogously with respect to $x_i,x_k$.
    For any pair $1\leq i<j\leq m$, we add an arc from $s_{i+2}$ to $s_{j+2}$.
    We complete the description of $T$ by adding two extra vertices $\{s_{m+3},t'\}$ and the arcs incident to these so that $\InN[T]{s_{m+3}} = \OutN[T]{t'} = V(T) \setminus \{s_{m+3},t'\}$.

		Finally, we let $\ell_1 = 1$ and $\ell_{j} = 2n+2j-2$ for every $j \in [2,m+3]$. We claim that $(X,\mathcal{C})$ is a positive instance of $(3,B2)$-SAT if and only if $(T,(s_1,\dots,s_{m+3}), (\ell_1,\dots\ell_{m+3}))$ is a positive instance of \PrescribedHamiltonianPath/. Note that the order of $T$ is $4n + m + 5 ∈ \BigO{n+m}$ and that, given $(X,\mathcal{C})$,  $T$ can be constructed in polynomial time. 
		Note also that $\ell_1=1$ and $\ell_{i+1} \geq \ell_i + 2$ for every $1\leq i\leq m+2$.
    Therefore, showing the equivalence between the instances implies the result.

    \medskip

    Assume first that $(X,\mathcal{C})$ is a positive instance, and let $\alpha\colon X\to \{\true,\false\}$ be a truth assignment of $(X,\mathcal{C})$. We build a Hamiltonian path of $T$ as follows. For every $i\in [n]$, let $P_i$ be $(v_i^1,v_i^2)$ if $\alpha(x_i) = \true$ and $(u_i^1,u_i^2)$ otherwise.
    Then, for every $j\in [m]$, let $x_h$ be an arbitrary variable such that $C_j$ is satisfied by $x_h$ in $\alpha$; that is either
    \begin{itemize}
        \item $x_h \in C_j$ and $\alpha(x_h) = \true$, or
        \item $\overline{x_h}\in C_j$ and $\alpha(x_h) = \false$.
    \end{itemize}
    We further let $Q_j = (s_{j+2},a_h)$ where $a_h$ is defined as above. Note that, by definition, $a_h\in X_h$ is not spanned by $P_h$. Therefore, the concatenation 
    \[
        P' = (s_1) \cdot P_1\cdot P_2\cdot \ldots \cdot P_n\cdot (s_2,t) \cdot Q_1\cdot Q_2\cdot  \ldots \cdot Q_m \cdot (s_{m+3},t')
    \]
    is a directed path of $T$ such that $T$ covers all vertices of $T$ except some vertices of $W$, and for every $1\leq i\leq m+3$, $s_i$ appears at index $\ell_i$ along $P'$. Therefore, to show that $(T,(s_1,\dots,s_{m+3}), (\ell_1,\dots,\ell_{m+3}))$ is a positive instance of \PrescribedHamiltonianPath/, it remains to prove that $P'$ can be extended into a Hamiltonian path of $T$.
    To see that this is possible, recall that $P'$ terminates at $t'$, and that $t'$ dominates $W$. Moreover, it is a celebrated result of R\'{e}dei~\cite{Rede34}  that every tournament admits a Hamiltonian path. Hence, since $V(T) \setminus V(P')\subseteq W$, $P'$ can indeed be extended into a Hamiltonian of $T$ by taking an arbitrary Hamiltonian path of $T[V(T) \setminus V(P')]$.

    \medskip

    Conversely, assume now that $T$ contains a Hamiltonian path $P$ where, for each $i\in[1,m+3]$, $s_i$ appears at the index $\ell_i$ along $P$.
    Let $P_1$ be the restriction of $P$ to the first $2n+2$ vertices. Observe that $P$ goes from $s_1$ to $s_2$. Moreover, since $\ell_j\geq 2n+4$ for every $j \geq 3$, $P_1$ does not contain any vertex of $\{s_3,\dots,s_{m+3}\}$.
    Furthermore, since $t$ and $t'$ dominate $W\setminus \{s_2\}$, $P_1$ does not contain any of $\{t,t'\}$.
    Hence, $P_1$ is a directed path of $T[W]$ of length $2n+1$ going from $s_1$ to $s_2$. By Claim~\ref{claim:directed_path}, for every $i\in [n]$, we have $V(P_1) \cap X_i \in \{\{u_i^1,u_i^2\}, \{v_i^1,v_i^2\}\}$.
    We can thus define the following assignment
    \[
        \alpha\colon x_i \mapsto 
        \begin{cases}
            \true &\text{if } V(P_1) \cap X_i = \{v_i^1,v_i^2\},\\[0em]
            \false &\text{otherwise, that is } V(P_1) \cap X_i = \{u_i^1,u_i^2\}.\\[0em]
        \end{cases}
    \]
    We claim that $\alpha$ is a satisfying assignment. To see this, let $C_j$ be an arbitrary clause of $\mathcal{C}$ ($1\leq j\leq m$), and let us show that $C_j$ is satisfied by $\alpha$. 
    
    Let $w$ be the successor of $s_{j+2}$ in $P$. By definition of $\ell_1,\dots,\ell_{m+3}$, $w$ does not belong to $\{s_1,\dots,s_{m+3}\}$. Hence, $w \in W\setminus \{s_1,s_2\}$. By construction, the out-neighborhood of $s_{j+2}$ in $W$ is precisely $\{a_h,a_i,a_k\}$ (defined as above) where $x_h,x_i,x_k$ are the variables contained in $C_j$. Assume without loss of generality that $w=a_h$. 

    If $x_h$ is positive in $C_j$, then $a_h \in \{u_h^1,u_h^2\}$. Since $a_h \notin V(P_1)$ by definition of a path, it follows from Claim~\ref{claim:directed_path} that $V(P_1) \cap X_h = \{v_h^1,v_h^2\}$, which by definition implies that $\alpha(x_h) = \true$. 
    Similarly, if $x_h$ is negated in $C_j$, then $a_h \in \{v_h^1,v_h^2\}$. This implies that $V(P_1) \cap X_h = \{u_h^1,u_h^2\}$, and that $\alpha(x_h) = \false$.
    In both cases, $C_j$ is satisfied by $\alpha$. The result follows.
\end{proof}

We are now ready to derive the \NP/-hardness of \TemporalExploration/ from \PrescribedHamiltonianPath/.

\hypertarget{thm:np_hardness:body}{}
\explorationSemicompleteNPHard*
\begin{proof}
    We prove the hardness result by reducing from the restricted version of \PrescribedHamiltonianPath/, where the input $(T,(s_i)_{i\in [r]},(\ell_i)_{i \in [r]})$ satisfies $\ell_1=1$ and $\ell_{i+1}\geq \ell_i+2$ for every $i\in [r-1]$.
    This problem is \NP/-hard by Lemma~\ref{lemma:tournament_Hamiltonian_prescribed_NP-hard}.

    Let thus $(T,(s_i)_{i\in [r]},(\ell_i)_{i \in [r]})$ be such an instance of \PrescribedHamiltonianPath/. We build a temporal tournament $D=(D_i)_{i\in [t]}$ with $t = |V(D)|-1$ snapshots as follows.
    For every $j\in [r]$, we let $S_j = \{x_j,y_j,z_j\}$ be three distinct vertices such that $S_j \cap V(T) =\emptyset $ and such that the $S_j$'s are pairwise disjoint. We let the vertex set of $D$ be
    \[
        V(D) = (V(T) \setminus \{s_1,\ldots, s_r\}) \cup \bigcup_{j=1}^rS_j, 
    \]
    and we let $n = |V(D)| = t+1 = |V(T)|+2r$. Intuitively, each set $S_j$ in $D$ represents the vertex $s_j$ in $T$.
    We now define the arc set of the $D_i$'s. 
    Let $W = V(T) \setminus \{s_1,\dots,s_r\}$.
    First, for every $i\in [n-1]$ and every pair of vertices $u,v\in W$ we orient the edge between $u$ and $v$ in $D_i$ as in $T$, so
    \[
        A(D_i[W]) = A(T[W]).
    \]
    It remains to consider the edges incident to at least one vertex in $\bigcup_{j=1}^rS_j$. 
    For this, let us define $\rho(j) = \ell_j + 2j-2$.
    For every $j\in [r]$, we let the edges induced by $S_j$ alternate between the two possible $\vec{C_3}$, except on the specific time step $\rho(j)$ in which case they induce a transitive tournament. Formally, for every $i\in [t]$ and $j\in [r]$ we let 
    \[
        A(D_i[S_j]) = 
        \begin{cases}
            \{x_jy_j,x_jz_j,y_jz_j\} &\text{if $i=\rho(j)$,}\\[0em]
            \{x_jy_j,y_jz_j,z_jx_j\} &\text{if $i\neq \rho(j)$ and $i\equiv \rho(j)+1 \bmod 2$,}\\[0em]
            \{y_jx_j,z_jy_j,x_jz_j\} &\text{if $i\neq \rho(j)$ and $i \equiv \rho(j) \bmod 2$.}
        \end{cases}
    \]
    It remains to specify, for every $j\in [r]$, the orientation of the edges having exactly one extremity in $S_j$.
    For this, we let $D_i$ be the unique tournament on $V(D)$ satisfying the orientations above and such that, for every $j\in [r]$ and every $s\in S_j$, we have:
    \[
        \OutN[D_i]{s} \setminus (S_1\cup S_2\cup \ldots \cup S_j) =
        \begin{cases}
            W &\text{if $i \leq  \rho(j)-2$,}\\[0em]
            \OutN[T]{s_j} \cap W &\text{if $i=\rho(j)-1$,}\\[0em]
            \emptyset &\text{if $i=\rho(j)$,}\\[0em]
            \emptyset &\text{if $i=\rho(j)+1$,}\\[0em]
            \OutN[T]{s_j} \cap W &\text{if $i=\rho(j)+2$,}\\[0em]
            W \cup S_{j+1} \cup \ldots \cup S_r &\text{if $i\geq \rho(j)+3$.}
        \end{cases}
    \]

    This completes the description of $D$. We now prove that $(T,(s_i)_{i\in [r]}, (\ell_i)_{i\in [r]})$ is a positive instance of \PrescribedHamiltonianPath/ if and only if $D$ admits a temporal exploration, implying the result.
    Note that, since $D$ contains $|V(D)|-1$ snapshots, $D$ admits a temporal exploration if, and only if, there exists an ordering $v_1,\dots,v_n$ of $V(D)$ such that, for every $i\in [n-1]$, $v_iv_{i+1}$ is an arc of $D_i$.

    \medskip

    Assume first that $T$ admits a Hamiltonian path $P=(u_1,\dots,u_{|V(T)|})$ such that, for every $j\in [r]$, $s_j$ appears at index $\ell_i$ along $P$. Let $Q$ be the ordering of $V(D)$ obtained from $P$ by replacing $s_j$ with $(x_j,y_j,z_j)$ for every $j\in [r]$. Note that every vertex of $D$ appears exactly once along $Q$ because $P$ is a Hamiltonian path of $T$. 
    We claim that $Q$ is a temporal exploration of~$D$. 

    To see this, let us label $v_1,\dots,v_n$ the vertices along $Q$ and show that, for every $i\in [n-1]$, $v_iv_{i+1}$ is an arc of $D_i$. 
    Let us thus fix $i\in [n-1]$. This is clear if both $v_i$ and $v_{i+1}$ belong to $W$, as by construction $D_i[W] = T[W]$ and $v_iv_{i+1}$ is an arc of $T$. 
    Let thus $j\in [r]$ be such that $\{v_i,v_{i+1}\}\cap S_j \neq \emptyset$. Note that $j$ is unique, since the vertices $s_1,\dots,s_r$ do not appear consecutively along $P$ (as $\ell_{k+1} \geq \ell_k+2$ for every $k\in [r-1]$ by assumption). Hence, $v_iv_{i+1}$ is either $x_jy_j$, $y_jz_j$, $w_1x_j$ with $w_1$ being the predecessor of $x_j$ in $Q$, or $z_jw_2$ with $w_2$ being the successor of $z_j$ in $Q$. Note that both $w_1$ and $w_2$ (if they exist) belong to $W$.
    
    Since $s_j = u_{\ell_j}$, observe that by construction of $Q$ we have that $x_j$ is at position $\ell_j + 2j-2 =\rho(j)$ along $Q$. It is straightforward to check that, by construction of $D$, we have $w_1x_j \in A(D_{\rho(j)-1})$ (if $x_j$ has a predecessor in $Q$), $x_jy_j \in A(D_{\rho(j)})$, $y_jz_j \in A(D_{\rho(j)+1})$, and $z_jw_2 \in A(D_{\rho(j)+2})$ (if $z_j$ has a successor in $Q$). This shows that $D$ admits a temporal exploration.

    \medskip
    
    We now prove the converse. Hence, let us assume that $D$ admits a temporal exploration $Q=(v_1,\dots,v_n)$, so $v_iv_{i+1}\in A(D_i)$ for every $i\in [n-1]$. We first prove the following.
    \begin{claim}
        \label{claim:Sj_consecutive}
        For every $j\in [r]$, we have $S_j = \{v_{\rho(j)}, v_{\rho(j)+1}, v_{\rho(j)+2}\}$.
    \end{claim}
    \begin{proofclaim}
        We proceed by induction on $j\in [r]$. 
        We start with the case $j=1$. We claim that the vertices of $S_1$ must be consecutive along $Q$. Indeed, if this is not the case, then there exist indices $i <k$ such that $D_{i}$ contains an arc from $S_1$ to $V(D)\setminus S_1$ and $D_{k}$ contains an arc from $V(D)\setminus S_1$. This does not happen, as by definition $S_1$ is dominated by $V(D)\setminus S_1$ in $D_1,D_2$, and $S_1$ dominates $V(D)\setminus S_1$ in $D_4,\dots,D_t$.
        Therefore, the vertices of $S_1$ appear consecutively along $Q$. Since the edges induced by $S_1$ alternate between two opposite directed triangles except on $D_{\rho(1)}=D_1$, this is possible only if $v_1\in S_1$. Therefore, $S_1 = \{v_1,v_2,v_3\}$, as desired.

        Let now $j\in [2,r]$ be such that the statement holds for every $j'<j$.
        For every $j'\in [r]$, we denote by $S_{\geq j'}$ the set of vertices $\bigcup_{j''\geq j'}S_{j''}$.

        \begin{subclaim}
            \label{subclaim:predecessor_WSj}
            Let $u$ be a vertex of $S_{\geq j}$. Its predecessor along $Q$ belongs to $W \cup S_{\geq j}$. 
        \end{subclaim}
        \begin{proofsubclaim}
            Let $w$ be the predecessor of \(u\). 
            If the statement does not hold, then $w\in S_{j'}$ for some $j'<j$. By induction, this happens only if $w=v_{\rho(j')+2}$. But then, by definition the out-neighborhood of $w\in S_j'$ in $D_{\rho(j')+2}$ is included in $W$, so it does not contain $u$, a contradiction.
        \end{proofsubclaim}

        \begin{subclaim}
            \label{subclaim:minimum_index}
            Let $u$ be a vertex of $S_{\geq j}$. The index of $u$ along $Q$ is at least $\rho(j)$.
        \end{subclaim}
        \begin{proofsubclaim}
            Assume for a contradiction that some vertex $u\in S_{\geq j}$ has index at most $\rho(j)-1$, and among all such vertices we let $u$ be the one of smallest index along $Q$.
            Let $w$ be the predecessor of $u$ along $Q$, which exists as $u\notin S_1$ and $S_1= \{v_1,v_2,v_3\}$ by induction.
            By Subclaim~\ref{subclaim:predecessor_WSj}, $w\in W \cup S_{\geq j}$.
            Clearly, by choice of $u$, we must have $w\in W$. 
            Let $j'\geq j$ be the index such that $u\in S_{j'}$. Then the index of $u$ along $Q$ is at most $\rho(j)-1 \leq \rho(j')-1$, which implies that the index of $w$ is at most $\rho(j')-2$. This is a contradiction, as for every $i\leq \rho(j')-2$, $S_{j'}$ dominates $W$ in $D_i$.
        \end{proofsubclaim}

        \begin{subclaim}
            \label{subclaim:successor_WSj}
            Let $u$ be a vertex of $S_j$.
            Every vertex after \(u\) along \(Q\) is in \(W ∪ Q_{\geq j}\).
        \end{subclaim}
        \begin{proofsubclaim}
            If this is not the case, then $w$ belongs to $S_{j'}$ for some $j'<j$. By induction, this happens only if $w = v_{\rho(j')}$, and thus $u \leq v_{\rho(j')-1}$.
            We have $\rho(j')-1 < \rho(j)$, which is a contradiction to Subclaim~\ref{subclaim:minimum_index}.
        \end{proofsubclaim}

        The three subclaims above together imply that the vertices of $S_j$ are consecutive along $Q$.
        Indeed, if this is not the case, since $v_1 \in S_1$ by induction, there exist three indices 
        $i_1<i_2<i_3$ such that:
        \begin{itemize}
            \item $v_{i_1} \notin S_j$ and $v_{i_1+1} \in S_j$, 
            \item $v_{i_2}\in S_j$ and $v_{i_2+1} \notin S_j$, and
            \item $v_{i_3}\notin S_j$ and $v_{i_3+1} \in S_j$.
        \end{itemize}

        By Subclaim~\ref{subclaim:minimum_index}, we have $i_1\geq \rho(j)-1$.
        By Subclaim~\ref{subclaim:successor_WSj}, $v_{i_2+1}\in W \cup S_{\geq j+1}$.
        By construction, since $v_{i_2}\in S_j$, $v_{i_2}$ is dominated by $W \cup S_{\geq j+1}$ in $D_{\rho(j)}$ and $D_{\rho(j)+1}$. Hence, since $i_2\geq i_1+1\geq \rho(j)$, we have $i_2\geq \rho(j)+2$.
        By Subclaim~\ref{subclaim:successor_WSj}, $v_{i_3}\in W \cup S_{\geq j+1}$. Moreover, we have $i_3\geq i_2+1\geq \rho(j)+3$. This is a contradiction, as by construction $S_j$ dominates $W \cup S_{\geq j+1}$ in $D_i$ for every $i\geq \rho(j)+3$.

        We just proved that the vertices of $S_j$ are consecutive along $Q$. Moreover, since the edges induced by $S_j$ alternate between two directed triangles except on $\rho(j)$, at least one of the vertices in $S_j$ has index $\rho(j)$. It follows from Subclaim~\ref{subclaim:minimum_index} that $S_j = \{v_{\rho(j)},v_{\rho(j)+1}, v_{\rho(j)+2}\}$.
        The claim follows.
    \end{proofclaim}

    Let $P$ be the ordering obtained from $Q$ by replacing the occurrence of $S_j$ by $s_j$ (the three vertices of $S_j$ appearing consecutively along $Q$ by Claim~\ref{claim:Sj_consecutive}). Let $u_1,\dots,u_{|V(T)|}$ be the labelling of $V(T)$ along $P$.
    Since $S_j = \{v_{\rho(j)}, v_{\rho(j)+1}, v_{\rho(j)+2}\}$ for every $j\in [r]$, observe that $s_j = u_{\rho(j) - 2j+2} = u_{\ell_j}$. 
    It thus remains to prove that $u_ku_{k+1}$ is an arc of $T$ for every $k\in [|V(T)|-1]$. This is clear if both $u_k$ and $u_{k+1}$ belong to $W$ as then $u_k$ and $u_{k+1}$ are also consecutive in $Q$, and $T[W] = D_i[W]$ for every $i\in [t]$.
    We can thus assume that exactly one vertex of $\{u_k,u_{k+1}\}$ belongs to $W$, as $\ell_j\geq \ell_{j-1}+2$ for every $j\in [r]$.
    
    Assume that $u_{k+1} = s_j$ for some $j\in [r]$, hence implying that $u_k \in W$. Then, by Claim~\ref{claim:Sj_consecutive}, $u_k = v_{\rho(j)-1}$ and the successor of $u_k$ in $Q$ is precisely $v_{\rho(j)} \in S_j$. Since $u_k\in W$, and $u_kv_{\rho(j)}$ is an arc of $D_{\rho(j)-1}$, we have that $u_k \in \InN[T]{s_j}$, as desired.

    Assume finally that $u_k = s_j$ for some $j\in [r]$, hence implying that $u_{k+1} \in W$. Then $u_{k+1} = v_{\rho(j)+3}$, and the predecessor of $u_{k+1}$ in $Q$ is precisely $v_{\rho(j)+2} \in S_j$. Since $u_{k+1}\in W$, and $v_{\rho(j)+2}u_k$ is an arc of $D_{\rho(j)+2}$, we have that $u_{k+1} \in \OutN[T]{s_j}$, as desired.
    This shows the equivalence between the instances, and concludes the proof of the Theorem.
\end{proof}

\subsection{Inapproximability}
\label{subsec:inapprox_semicomplete}

We finally consider the approximability of the following optimization problem.

\TemporalExplorationFastestDef

In general, in polynomial time, it is known that \TemporalExplorationFastest/ cannot be approximated within a factor of $c\cdot |V(D)|$ 
for some absolute constant $c>0$, unless \P/ = \NP/~\cite{michailTCS634}. It is further known that, for every $\epsilon>0$, it cannot be approximated within a factor $2-\epsilon$, even when restrained to always-strong temporal digraphs~\cite{michailTCS634}.

The situation is different for semicomplete  temporal digraphs. Indeed, every temporal exploration clearly requires at least $|V(D)|-1$ snapshots, and Theorem~\ref{theorem:tournament-exploration-upper-bound} guarantees that $\lfloor \frac{4}{3}|V(D)|-1\rfloor$ snapshots are enough. Therefore, the trivial algorithm consisting of returning the minimum between $\Lifetime{D}$ and $\lfloor \frac{4}{3}|V(D)|-1\rfloor$ outputs a solution $t\in [\ell(D)]$ such that
\[
    t\leq \left\lfloor \tfrac{4}{3}t_{OPT} +\tfrac{1}{3}\right\rfloor.
\]

We conclude this section with the following result, showing that this trivial algorithm is essentially optimal, even for temporal tournaments.

\hypertarget{statement:semicomplete approximation hardness:body}{}
\explorationSemicompleteApproximationHard*
\begin{proof}
    Assume for a contradiction that such an approximation $\mathcal{A}$ exists, and let us deduce the existence of a polynomial-time algorithm which decides whether a given temporal tournament $D=(D_i)_{i\in [|V(D)|-1]}$ tournament admits a temporal exploration. The result then follows from Theorem~\ref{thm:np_hardness}.

    Let $D=(D_i)_{i\in [|V(D)|-1]}$ be such a temporal digraph of order $n$.
    Let $N= \lceil n/\epsilon\rceil$, and note that, in particular we have
    \[
        (\tfrac{4}{3}-\epsilon) (n+3N-1) < n+4N.
    \]

    The first step of the algorithm consists of building, from $D$, a temporal tournament $D'=(D_i')_{i\in [t]}$ of order $n+3N$ with $t=\lfloor \frac{4}{3}(n+3N)-1\rfloor$ as follows. Note that $D'$ is guaranteed to admit a temporal exploration by Theorem~\ref{theorem:tournament-exploration-upper-bound}, and that the size of $D'$ is at most polynomial in the size of $D$ (for fixed $\epsilon$).
    For every $j\in [N]$ we let $U_j = \{x_j,y_j,z_j\}$ be a set of three vertices, and we let the vertex set of $D'$ be 
    \[
        V(D') = V(D) \cup \bigcup_{j\in [N]}U_j.
    \]
    Then, for every $i\in [n-1]$ we let 
    \[
        A(D_i'[V(D)]) = A(D_i[V(D)]),
    \]
		and for every $n\leq i \leq t$ we orient the edges in $D_i'[V(D)]$ arbitrarily, say 
    \[
        A(D_i'[V(D)]) = A(D_1[V(D)]).
    \]
    Then, for every $i\in [t]$, we let $D_i'$ contain all the arcs from $V(D)$ to $\bigcup_{j\in [N]} U_j$ and all the arcs from $U_{j}$ to $U_{j'}$ whenever $j<j'$, hence guaranteeing that in every temporal exploration of $D'$, the vertices of $D'$ will be explored in the following order : $V(D),U_1,U_2,\ldots,U_N$.
    It remains to define the arcs inside each $U_j$. For every $i\in [t]$, we let
    \[
        A(D_i'[Y_j]) = 
        \begin{cases}
            \{x_jy_j,x_jz_j,y_jz_j\} &\text{if $i=n+3j-2$,}\\[0em]
						\{x_jy_j,y_jz_j,z_jx_j\} &\text{if $i\neq n+3j-2$ and $i\equiv n+3j-1 ~ (\bmod \, 2)$,}\\[0em]
						\{y_jx_j,z_jy_j,x_jz_j\} &\text{if $i\neq n+3j-2$ and $i \equiv n+3j ~ (\bmod \, 2)$.}
        \end{cases}
    \]
    Now, we can simply apply $\mathcal{A}$ on $D'$ and output true if the result is smaller than $n+4N$, and false otherwise.
    To see that the algorithm is correct, we only have to prove the following.
    \begin{claim}
        If $D$ admits a temporal exploration, then $D'$ admits a temporal exploration using at most $n+3N-1$ snapshots, and if $D$ does not admit a temporal exploration then every temporal exploration of $D'$ uses at least $n+4N$ snapshots.
    \end{claim}

    \begin{proofclaim}
        If $D$ admits a temporal exploration $(v_1,\dots,v_n)$, then it is straightforward to check that
        \[
            (v_1,\dots,v_n) \cdot (x_1,y_1,z_1) \cdot \ldots \cdot (x_N, y_N, z_N)
        \]
        is a temporal exploration of $D'$ using $n+3N-1$ snapshots.
        On the other hand, assume that $D$ does not admit any temporal exploration. Let $(v_1,\dots,v_r)$ be a temporal exploration of $D'$. Then, as mentioned above, the vertices of $D'$ are explored in the following order: $V(D),U_1,\ldots,U_N$. Covering the vertices of $D$ requires at least $n$ steps by assumption, which implies that, for every $j\in [N]$, when reaching the vertices of $U_j$, the arcs induced by $U_j$ are alternating between two opposite $\vec{C_3}$. It follows that covering the vertices of each $U_j$ requires at least \(4N\) additional steps.
    \end{proofclaim}
    
    The result now follows from the claim above and the fact that, by choice of $N$, 
    \[
    (\frac{4}{3}-\epsilon)(n+3N-1) < n+4N.\qedhere
    \]
\end{proof}

\section{Large minimum degree}
\label{sec:large minimum degree}

Towards generalizing \cref{theorem:tournament-exploration-upper-bound} to larger classes of temporal digraphs,
we consider a parameter \(c\) and
prove that all always-strong temporal digraphs on \(n\) vertices with lifetime \(\Lifetime{D} \geq f(c)\cdot n\)
admit a temporal exploration.
A natural way to generalize the concept of semicomplete digraphs is to consider
digraphs of large minimum degree.

To simplify notation, given an undirected graph \(G\),
we define \(\MindegComplement{G} \coloneqq \Abs{\V{G}} - \Mindeg{G} - 1\).
For a digraph \(D\), define \(\MindegComplement{D} \coloneqq \MindegComplement{\UnderlyingGraph{D}}\),
and for a temporal digraph \(D' = (D_i)_{i ∈ [t]}\), we define
\(\MindegComplement{D'} \coloneqq \max_{i ∈ [t]} \MindegComplement{D_i}\).
Note that for every \(v ∈ \V{G}\) there are at most \(\MindegComplement{G}\)
distinct vertices of \(\V{G} ∖ \{v\}\)
which do not have an edge to \(v\).
In particular, \(\MindegComplement{D'} = 0\)
if and only if every \(D_i\) is semicomplete.

The main result of this section is \cref{statement:degree-n-c-upper-bound-strong},
showing that \(f(\MindegComplement{D}) ∈ \BigO{\MindegComplement{D}}\)
suffices in order to guarantee the existence of a temporal exploration for \(D\).
In \cref{statement:degree-n-c-lower-bound},
we show that this bound is asymptotically tight.

We first prove a weaker result in \cref{sec:cnlog n upper-bound},
showing that \(\Lifetime{D} \geq d \MindegComplement{D} n \log n\),
for some constant \(d\),
suffices (\cref{statement:degree-n-c-upper-bound-strong-log}).
This proof is simpler and introduces several lemmas required to prove \cref{statement:degree-n-c-upper-bound-strong} later.
Then, we introduce \emph{travel plans}, the main object used to obtain \cref{statement:degree-n-c-upper-bound-strong},
in \cref{sec:travel plan}.
In \cref{sec:c2n upper bound} we show how to use the travel plans in order to derive \cref{statement:degree-n-c-upper-bound-strong}, and then we obtain a matching lower bound with \cref{statement:degree-n-c-lower-bound}.

\subsection{Simpler upper-bound}
\label{sec:cnlog n upper-bound}

We start by obtaining an analogue of \cref{lemma:directed_king} when \(\MindegComplement{D} \leq c\).
Towards this end, we first show that we can find few vertices which, together,
can reach all vertices in a given set \(X\).
We require the following well-known statement,
see for instance~\cite[Proposition 1.5]{Cranston2024}.
Recall that a graph \(G\) is \emph{\(k\)-colorable}
if we can partition \(\V{G}\) into at most \(k\) sets
\(X_{1}, X_{2}, \ldots, X_{t}\),
called \emph{color classes},
such that
for every edge \(uv\) in \(G\),
if \(u ∈ X_i\), then \(v \notin X_i\).
\begin{observation}
	\label{statement:degree-coloring}
	Let \(G\) be a graph of maximum degree \(c\).
	Then \(G\) is \((c + 1)\)-colorable.
\end{observation}

\begin{lemma}
	\label{statement:large-minimum-degree-few-2-dominators}
	Let \(c \geq 0\) and \(D = (D_1,D_2)\) be a temporal digraph of order $n$ 
	with \(\MindegComplement{D} \leq c\).
    Then, for every \(X ⊆ \V{D}\), there exists \(U ⊆ X\), \(\Abs{U} \leq 2c + 1\),
	such that
	\(X \subseteq \TemporalReach{2}{U}\).
\end{lemma}
\begin{proof}
	Let 
	\begin{align*}
		D_1' & = (\V{D}, \{(u,v) ∈ \A{D_1} \sth (u,v) ∈ \A{D_2} \text{ or } (v,u) ∈ \A{D_2}\}) \text{ and }
		\\[0em] 
		D_2' & = (\V{D}, \{(u,v) ∈ \A{D_2} \sth (u,v) ∈ \A{D_1} \text{ or } (v,u) ∈ \A{D_1}\}).
	\end{align*}
	Note that \(\UnderlyingGraph{D_1'} = \UnderlyingGraph{D_2'}\).
	Further, since \(\MindegComplement{D_1} \leq c\),
	we have \(\MindegComplement{D_1'} \leq 2c\).

	Let \(\bar{G}\) be the complement of \(\UnderlyingGraph{D_1'}\), that is,
	we have an edge \(uv\) in \(\bar{G}\) if, and only if, \(uv\) is not an edge in \(\UnderlyingGraph{D_1'}\).
	Clearly, the maximum degree of \(\bar{G}\) is at most  \(2c\).
	Hence, by \cref{statement:degree-coloring},
	there is a proper coloring of	\(\bar{G}\) with at most \(2c + 1\) colors.
	Let \(C_{1}, C_{2}, \ldots, C_{k}\) be the color classes of \(\bar{G}\).

    By definition, each \(C_i\) induces
	a semicomplete digraph in \(D_1'\) and
	another	semicomplete digraph in \(D_2'\).
	For each \(C_i\),
	let \(T^i = (T^i_1, T^i_2)\) be the temporal digraph where
	\(T^i_j\) is the digraph induced by \(C_i ∩ X\) in \(D_j'\).
	Note that \(T^i\) is a semicomplete temporal digraph.
	By \cref{lemma:directed_king},
	there is some \(v_i^* ∈ C_i ∩ X\)
	such that
	\(C_i ∩ X \subseteq \TemporalReach[T^i]{2}{v_i^*}\).
	This implies that \(C_i \cap X ⊆ \TemporalReach[D]{2}{v_i^*}\).
	Let \(U = \Set{u^*_{1}, u^*_{2}, \ldots, u^*_{k}}\).
	Clearly, \(X \subseteq \TemporalReach[D]{2}{U}\), completing the proof.
\end{proof}

Next, we strengthen the reachability properties of the vertices in \(U\),
ensuring that every \(u ∈ U\) can reach many vertices.
To achieve this, we prove in the next lemma
that, when going from one snapshot to the next,
we can either find a vertex \(u ∈ U\)
which is redundant, that is \(U ∖ \{u\}\) can still cover all vertices of \(X\),
or all vertices in \(U\) can reach many vertices of \(X\).

\begin{lemma}
	\label{statement:large-minimum-degree-reachability-expansion}
	Let \(c \geq 0\) and \(t \geq 2\) be integers.
	Let \(D = (D_i)_{i ∈ [t]}\) be a temporal digraph
	where 
	\(n = \Abs{\V{D}}\) and
	\(\MindegComplement{D} \leq c\). Let \(X ⊆ \V{D}\) and \(U ⊆ \V{D}\)
	such that
	\(X \subseteq \TemporalReach{t - 1}{U}\).
	Then one of the following holds:
\begin{enumerate}
		\item \label{statement:large-minimum-degree-reachability-expansion:drop}
			there exists $u\in U$
			such that
			\(X \subseteq \TemporalReach{t}{U ∖ \{u\}}\), or
		\item \label{statement:large-minimum-degree-reachability-expansion:large-reach}
            \(\Abs{U} \leq c+1\) and, for every $u\in U$, 
			\(\Abs{X ∖ \TemporalReach{t}{u}} \leq c\).
	\end{enumerate}
\end{lemma}
\begin{proof}
    Let $U=\{v_1,\dots,v_k\}$ and
	assume that \cref{statement:large-minimum-degree-reachability-expansion:drop} does not hold.
	Then, for every \(v_i ∈ U\) there is some \(w_i ∈ \TemporalReach{t - 1}{v_i} ∩ X\)
	such that
	\(w_i \notin \TemporalReach{t}{v_j}\)
	for every \(j\neq i\).

	Clearly \(w_i \neq w_j\) if \(i \neq j\).
	Moreover, if an arc \((w_i, w_j)\) exists in \(D_{t}\),
	then \(w_j ∈ \TemporalReach{t}{v_i}\), a contradiction to the definition of \(w_j\).
	Hence, the set \(W=\{w_{1}, w_{2}, \ldots, w_{k}\}\) is an independent set in \(D_{t}\).
	Because \(\MindegComplement{D_i} \leq c\),
	there are at most \(c+1\) distinct vertices in \(W\), implying that \(\Abs{U} \leq c+1\).

	Let \(v_i ∈ U\),
	let \(X' = X ∖ \TemporalReach{t - 1}{v_i}\), and
	let \(w ∈ \TemporalReach{t - 1}{v_i}\)
	such that
	\(w \not\in \OutN[D_{t}]{X'}\). Note that $w$ exists, as \(\InN[D_{t}]{w_i} \cap X' \neq ∅\) would imply
	the existence of some \(v_j \neq v_i\)
	with \(w_i ∈ \TemporalReach{t}{v_j}\).
	Because \(\MindegComplement{D_t} \leq c\),
	we have that
	\(\Abs{\OutN[D_{t}]{w_i} \cap X'} \geq \Abs{X'} - c\), and thus $\Abs{X \setminus \TemporalReach{t}{v_i}} \leq c$.
	Therefore, we obtain \cref{statement:large-minimum-degree-reachability-expansion:large-reach}.
\end{proof}

Now we can put the previous lemmas together.

\begin{lemma}
	\label{statement:degree-n-c-few-princesses}
	Let \(c \geq 0\) be an integer. Let \(D\) be a temporal digraph
	where
	\(n = \Abs{\V{D}}\), $\Lifetime{D}\geq 2c+2$, and
	\(\MindegComplement{D} \leq c\). Then, for every \(X ⊆ \V{D}\), there exists \(U ⊆ X\), \(\Abs{U} \leq c+1\),
	such that
	\(X \subseteq \TemporalReach{}{U}\) and,
	for every \(u ∈ U\),
	\(\Abs{X \setminus \TemporalReach{}{u}} ≤ c\).
\end{lemma}

\begin{proof}
	By \cref{statement:large-minimum-degree-few-2-dominators},
	there is some \(U' ⊆ X\), \(\Abs{U'} \leq 2c + 1\)
	such that
	\(X\subseteq \TemporalReach{2}{U'}\).

	We apply \cref{statement:large-minimum-degree-reachability-expansion} iteratively as follows.
	Start with \(U_1 = U'\).
	On step \(i\), apply \cref{statement:large-minimum-degree-reachability-expansion}
	to \(U_i\) and \(\TemporalSlice{D}{1}{i + 2}\).
	
	If \cref{statement:large-minimum-degree-reachability-expansion:large-reach} holds,
	then \(\Abs{X ∖ \TemporalReach{i+2}{u}} \leq c\) for every \(u ∈ U_i\) and
	\(\Abs{U_i} \leq c+1\).
	Hence, \(U \coloneqq U_i\) satisfies the conditions in the statement.
	Otherwise, \cref{statement:large-minimum-degree-reachability-expansion:drop} holds,
	so there is some \(u_i ∈ U_i\)
	such that
	\(\TemporalReach{i+2}{U_{i}\setminus \{u_i\}} ⊇ X\),
	and we let \(U_{i + 1} = U_i ∖ \{u_i\}\).

	Observe that, if \(\Abs{U_i} = 1\) at any step,
	then \(U_i\) satisfies the conditions in the statement.
	As \(\Abs{U_{i + 1}} < \Abs{U_i}\)
	and \(\Abs{U_1} \leq 2c + 1\),
	the iteration above proceeds for at most \(2c\) steps,
	yielding the desired set \(U\) when it stops.
\end{proof}

Later, we will need the following slight variation of the statement above,
allowing us to make the set \(U\) larger.

\begin{corollary}
	\label{statement:degree-n-c-exact-princesses}
	Let \(c \geq 0\) and \(d \geq 1\) be integers.
	Let \(D\) be a temporal digraph
	where
	\(n = \Abs{\V{D}}\) and
	\(\MindegComplement{D} \leq c\). If \(\Lifetime{D} \geq \bound{statement:degree-n-c-exact-princesses}{t}{c} \coloneqq 2c + 2\) then, for every \(X ⊆ \V{D}\), there exists \(U ⊆ X\), \(\Abs{U} = \min(c + d, \Abs{X})\),
	such that
	\(X\subseteq \TemporalReach{}{U}\) and,
	for every \(u ∈ U\),
	\(\Abs{X \setminus (\TemporalReach{}{u} ∪ U)} ≤ c\).
\end{corollary}
\begin{proof}
	If \(\Abs{X} < c + d\),
	then \(U = X\) satisfies the conditions in the statement.

	Otherwise, we construct sets \(U_{1}, U_{2}, \ldots, U_{k}\) and
	\(X_{1}, X_{2}, \ldots, X_{k}\) iteratively as follows.
	Let \(X_1 = X\).
	By \cref{statement:degree-n-c-few-princesses},
	there is some \(U_1 ⊆ X\), \(\Abs{U_1} \leq c+1\leq c+d\)
	such that
	\( \TemporalReach{t}{U_1} ⊇ X\) and
	\(\Abs{X \setminus \TemporalReach{t}{u}} ≤ c\)
	for every \(u ∈ U_1\).

	On step \(i > 1\),
	if \(\Abs{U_{i - 1}} \geq c + d\),
	then stop the construction.
	Otherwise, let \(X_i = X_{i - 1} ∖ U_{i-1}\).
	By \cref{statement:degree-n-c-few-princesses},
	there is some \(U_i' ⊆ X_i\), \(\Abs{U_i'} \leq c\)
	such that
	\(\TemporalReach{t}{U_i'} ⊇ X_i\) and
	\(\Abs{X_i ∖ \TemporalReach{t}{u}} ≤ c\)
	for every \(u ∈ U_i'\).

	If \(\Abs{U_{i - 1} ∪ U_i'} \leq c + d\),
	set \(U_i = U_{i - 1} ∪ U_i'\).
	Otherwise, take some \(U_i'' ⊆ U_i'\)
	such that
	\(\Abs{U_{i - 1} ∪ U_i''} = c + d\),
	set \(U_i = U_{i - 1} ∪ U_i''\) and
	stop the construction.

	After completing the construction,
	we obtain a set \(U \coloneqq U_k\)
	such that
	\(\Abs{U_k} = c + d\).
	Because 
	\(\TemporalReach{t}{U_1} ⊇ X\),
	it is immediate that \(U\) satisfies the conditions in the statement.
\end{proof}

Using the set \(U\) obtained above, we can construct a walk covering many vertices
using few snapshots.
As each vertex of \(U\) can reach every other vertex except for at most \(c\),
in each step we decrease the size of the set of reachable vertices by \(c\). 
While this prevents us from covering all vertices with one walk,
we can still cover a fraction of them.

\begin{lemma}
	\label{statement:degree-n-c-long-quick-walk}
	Let \(c \geq 1\) be an integer.
	Let \(D\) be a temporal digraph
	where
	\(\MindegComplement{D} \leq c\), and
	let \(X ⊆ \V{D}\).
	If \(\Lifetime{D} \geq 2\Abs{X}(1 + 1/c)\), then
    there is a temporal walk \(W\) in \(D\)
	such that
	\(\Abs{\V{W} \cap X} \geq \Abs{X} / (c + 1)\).
\end{lemma}
\begin{proof}
	Let \(\ell = 2c + 2\),
	let \(x = \Ceiling{\Abs{X} / (c + 1)} \) and
	let \(v^* ∈ X\) be some vertex.
    The statement is trivial when $X=\emptyset$, so we assume that $x\geq 1$. Moreover, we assume that $x\geq 2$, for otherwise we may take $W$ as an arbitrary single vertex $x\in X$.

    For every $i\in [x-1]$ let $D^i = \TemporalSlice{D}{1+(i-1)\ell}{i\ell}$. Note that this is well-defined as $\Lifetime{D}\geq (x-1)\ell$.
	We define pairwise distinct vertices \(v_{1}, v_{2}, \ldots, v_{x} ∈ X\)
	and temporal walks
	\(W_{1}, W_{2}, \ldots, W_{x - 1}\)
	such that
	\(W_i\) is a temporal \(v_i\)-\(v_{i+1}\)-walk in \(D^i\).

		Let $X_1=X$. Then, for every $i ∈ [x-1]$, let $v_i$ be a vertex such that $\Abs{X_i\setminus \TemporalReach[D^i]{\ell}{v_i}} \leq c$, and we further let $X_{i+1} = (X_i\cap \TemporalReach[D^i]{\ell}{v_i}) \setminus \{v_i\}$. The existence of $v_i$ is guaranteed by \cref{statement:degree-n-c-few-princesses}, by taking any $v_i \in U$.
    We note that \(\Abs{X_{i + 1}} \geq \Abs{X_i} - c - 1\), so 
    \[
        \Abs{X_{i+1}} \geq \Abs{X} - i(c+1) \geq |X| - (x-1)(c+1) >0.
    \]

		This defines vertices $v_1,\dots,v_{x-1}$, and we let $v_x$ be an arbitrary vertex of $X_x$. It follows by construction that, for every $i ∈ [x-1]$, there is a temporal $v_i$-$v_{i+1}$-path in $D^i$.
	Let \(W \coloneqq W_{1} \cdot W_{2} \cdot \ldots \cdot W_{x - 1}\)
	be the concatenation of all the walks defined above.
	By construction, \(W\) is a temporal walk in \(D\) visiting \(v_{1}, v_{2}, \ldots, v_{x} ∈ X\). The lemma follows.
\end{proof}

So far our results do not require strongly-connected snapshots.
Unfortunately, we have little control of
which vertex the walk given by the previous lemma will end on.
If the snapshots are only unilateral, we might end on a vertex
which will never be able to reach the other vertices again.
For always-strong temporal digraphs, however, this is not an issue
as we can use \(n-1\) snapshots in order to connect two walks.

\begin{theorem}
	\label{statement:degree-n-c-upper-bound-strong-log}
	Every always-strong temporal digraph \(D\) on \(n\) vertices with \(\MindegComplement{D} \leq c\)
	and \(\Lifetime{D} > (3c+7)n\log n\)
	admits a temporal exploration.
\end{theorem}
\begin{proof}
	If \(n \leq 2\), then the statement is clearly true for \(\Lifetime{D} \geq 1\).
	Moreover, if \(c = 0\), then the statement follows from \cref{theorem:tournament-exploration-upper-bound}.
	So assume \(n \geq 3\) and \(c \geq 1\).

	Let 
	\(d = 1 - 1/(c+1)\),
	\(d' \coloneqq 1 / d = (c+1) / c\),
	\(n_1 = n\), \(n_i = d n_{i - 1}\),
	\(\ell_i = n - 1 + 2n_id'\),
	\(t_0 = 0\),
	\(t_i = \sum_{j = 1}^i \lceil\ell_j\rceil\), and 
	\(k = \lfloor \log_{d'} n\rfloor +1 = \lfloor \log_2 n / \log_2 d' \rfloor + 1\).

	We first show that \(\Lifetime{D} \geq t_k\).
	Towards this end, observe that \(n_i = d^{i - 1}n\) and that 
	\(\log_{d'} n = \log n / \log \frac{c+1}{c} \leq c \log n\).
	Thus,
	\begin{align*}
		\sum_{i=1}^k n_i & = n \sum_{i=1}^kd^{i - 1} = \frac{n}{d} \cdot \frac{d (1 - d^k)}{1 - d} \leq \frac{n}{d} \cdot \frac{d}{1 - d} = n / (1 - d) = (c+1)n,
		\\[0em] t_k \leq \sum_{i = 1}^k( \ell_i+1) &= kn + 2d'\sum_{i = 1}^kn_i 
        \\[0em] &\leq kn + 2d'(c+1)n 
		\\[0em] & = n\lfloor \log_{d'} n + 1 \rfloor + 2(c + 2 + 1/c)n
		\\[0em] & \leq n(c\log n + 1) + 2(c + 3)n 
		\\[0em] & \leq (3c+7)n\log n \leq \Lifetime{D}.
	\end{align*}

	For each \(1 \leq i \leq k\),
	define \(H_i = \TemporalSlice{D}{1 + t_{i-1}}{t_i}\). Note that $\ell(D_i) = \lceil \ell_i\rceil$. We split $H_i$ in two, and let
	\(F_i = \TemporalSlice{H_i}{1}{\Ceiling{\ell_i} - n + 1}\) and
	\(G_i = \TemporalSlice{H_i}{\Ceiling{\ell_i} - n + 2}{\Ceiling{\ell_i}}\).
	Observe that \(\Lifetime{F_i} = \Ceiling{2n_i(1 + 1/c)}\)
    and
	\(\Lifetime{G_i} = n - 1\).
	
In what follows, we define, for some $x\leq k$, a sequence of temporal walks $W_1,\dots,W_{x}$ of $H_1,\ldots,H_{x}$ respectively, that together span $\V{D}$, and such that $\End{W_i} = \Start{W_{i+1}}$ for every $i ∈ [x-1]$. The statement then follows by concatenating these temporal walks.

	Set \(X_1 = \V{D}\) and repeat the following until $X_i = \emptyset$. We maintain that $|X_i| \leq n_i$, which clearly holds for $i=1$. 
    By \cref{statement:degree-n-c-long-quick-walk},
	there is some temporal walk \(Q_i\) in \(F_i\)
	such that
	\(\Abs{\V{Q_i} ∩ X_i} ≥ \Abs{X_i} / (c + 1)\).
    Let \(X_{i + 1} = X_i ∖ \V{Q_i}\), and note that \(\Abs{X_{i+1}} = \Abs{X_i ∖ \V{Q_i}} \leq n_{i + 1}\).
    Observe that the process indeed stops in at most $k$ steps since
    \[
        n_{k+1} = n \cdot d^{k} <1,
    \]
    as by definition we have $\frac{1}{d^k} = {d'}^k > {d'}^{\log_{d'} n} = n$. By construction, observe that $\bigcup_{i=1}^{x}V(Q_i) = V(D)$.
	For each \(i\in [x]\),
	let \(s_i = \Start{Q_i}\) and
	\(r_i = \End{Q_i}\).
	By \cref{statement:strong-temporal-u-v-path}, for every $1\leq i \leq x-1$,
	there is a temporal \(r_i\)-\(s_{i+1}\)-path \(R_i\) in \(G_i\).
	Define \(W_i = Q_i \cdot R_i\) and
	\(W = W_{1} \cdot W_{2} \cdot \ldots \cdot W_{x-1}\cdot Q_x\).
    It follows that \(W\) is a temporal exploration for \(D\).
\end{proof}

\subsection{Travel plan}

\label{sec:travel plan}

With each application of \cref{statement:degree-n-c-long-quick-walk},
we cover some fraction of the currently unexplored vertex set.
However, as \cref{statement:degree-n-c-few-princesses}
may give us multiple vertices, all of which reaching almost all other vertices,
there may actually be many choices for the walk given by \cref{statement:degree-n-c-long-quick-walk}.
If we choose poorly, the first walk will visit all ``good'' vertices of the temporal digraph,
and so subsequent applications of the lemma explore less and less vertices.
In order to avoid this, we define \emph{travel plans},
which tell us how we can construct
walks covering many vertices (see \cref{fig:travel-plan-example} for an illustration of the following definition).

\begin{definition}
	\label{definition:travel-plan}
	Let \(D\) be a temporal digraph.
	A \emph{travel plan of $D$ of gap \(g\), order \(k\) and width \(w\)} is a
	sequence \(\mathcal{T} = (T_{1}, T_{2}, \ldots, T_{k})\), $T_i\subseteq V(D)$,
	such that
	\begin{itemize}
		\item \(\Abs{T_i} \leq w\) for every \(i ∈ [k]\),
		\item \(\bigcup_{i=1}^k T_i = \V{D}\), 
		\item for every sequence \((i_1, v_1), (i_2, v_2), \ldots, (i_x, v_x)\)
			where
			\(v_j ∈ T_{i_j}\)
			for every \(j ∈ [x]\)
			and 
			\(i_j < i_{j + 1} - g\)
			for every \(j ∈ [x-1]\),
			there is a temporal walk \(W\) in \(D\)
			such that
			\(W\) spans \(v_{1}, v_{2}, \ldots, v_{x}\) in this order.
	\end{itemize}
\end{definition}

\begin{figure}
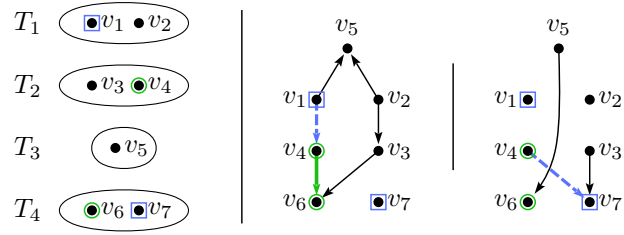

	\centering

	\caption{\label{fig:travel-plan-example}
	A travel plan of width \(2\) and gap \(1\) (left) for a temporal digraph (right).
	The sequences \((v_1, v_7)\) (blue rectangle) and \((v_4, v_6)\) (green circle)
	respect the gap of 1, and
	so there are temporal walks \(W_1 = (v_1, v_4, v_7)\), \(W_2 = (v_4, v_6)\) visiting those vertices in the given order.}
\end{figure}

To obtain a travel plan of small width and gap,
we repeatedly apply \cref{statement:degree-n-c-exact-princesses},
obtaining sets \(A_i\).
In order to ensure that every vertex in \(A_i\)
can reach every vertex in \(A_{i+2}\),
the set \(A_{i+2}\)
must be in the intersection of the reachability set of every vertex in \(A_i\).
While this intersection does not cover all vertices,
it does cover all vertices of large indegree.
The next observation allows us to put the uncovered vertices in one of the sets \(T_i\)
while constructing our travel plan.

\begin{observation}
	\label{statement:degree-n-c-small-hard-to-reach-set}
	Let \(c\) be an integer,
	let \(D\) be a digraph with \(\MindegComplement{D} \leq c\),
	let \(X ⊆ \V{D}\) and
	\(U = \{u ∈ X \mid \Indeg{v} \leq c\}\).
	Then \(\Abs{U} \leq 3c\) and
	\(\Abs{X ∖ \OutN{u}} \leq 2c\) for every \(u ∈ U\).
\end{observation}
\begin{proof}
	Because \(\MindegComplement{D} \leq c\),
	it is immediate that
	\(\Abs{X ∖ \OutN{u}} \leq 2c\) for every \(u ∈ U\).
	In particular, every \(u ∈ U\) has at least \(\Abs{U} - 2c\) outneighbors inside \(U\).
	
	Assume towards a contradiction that \(\Abs{U} > 3c\).
	Let \(G\) be the digraph induced by \(U\) in \(D\).

	Clearly, \(\Outdeg[G]{u} \geq 3c + 1 - 2c = c + 1\) for every \(u ∈ U\).
	Hence, \(\sum_{u ∈ U} \Indeg[G]{u} = \sum_{u ∈ U}\Outdeg[G]{u} \geq (c + 1) \Abs{U}\).
	Thus, there is at least one \(u ∈ U\) with
	\(\Indeg[G]{u} \geq c + 1\), a contradiction to the definition of \(U\).
	Thus, \(\Abs{U} \leq 3c\).
\end{proof}

\begin{figure}
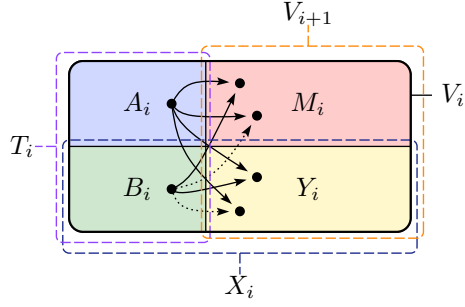

	\centering

	\caption{\label{fig:travel-plan-proof-construction} The sets constructed on step \(i\) in the proof of \cref{statement:degree-n-c-travel-plan}. The dotted arcs indicate non-reachable vertices in the temporal digraph \(H_i\).}
\end{figure}

The biggest challenge in constructing a travel plan of small width
is ensuring the required connectivity between the sets \(T_i\).
While \cref{statement:degree-n-c-exact-princesses} gives us
a small set \(U\) of vertices, each reaching almost all vertices of a set \(X\),
we ``lose'' vertices which are not reachable by \emph{every} vertex in \(U\).
\Cref{statement:degree-n-c-small-hard-to-reach-set} helps us gain connectivity from ``leftover'' vertices, letting us put them into one of the sets \(T_i\),
but they still do not reach every vertex in the remaining subset of vertices.
For this reason, we need to introduce a gap of one in the travel plan,
using the intermediate sets to complete the connectivity between non-adjacent sets.

\begin{figure}
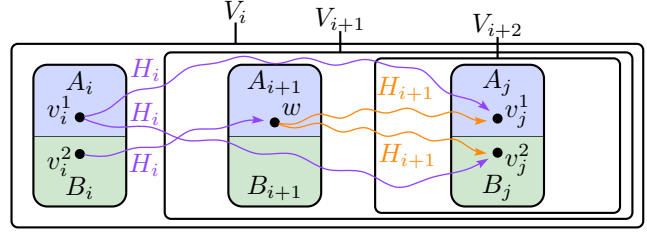

	\centering

	\caption{\label{fig:travel-plan-proof-paths} Illustration of the proof of \cref{statement:degree-n-c-travel-plan}. The temporal paths from \(v_i^1\) are all inside \(H_i\), while those from \(v_i^2\) need to first go to \(w\) in \(H_i\), and the to their targets in \(H_{i+1}\).}
\end{figure}

We first define
\begin{align*}
	\bound{statement:degree-n-c-travel-plan}{w}{c}  =
	5c + 1,
	\quad
	 \bound{statement:degree-n-c-travel-plan}{t}{n,c}  =
	(\bound{statement:degree-n-c-exact-princesses}{t}{c} + 1) ·
	n/(2c + 1)
	    \quad
    \text{and}
    \quad
	 \bound{statement:degree-n-c-travel-plan}{k}{n,c}  =
	n/\bound{statement:degree-n-c-travel-plan}{w}{c}.
\end{align*}

\begin{lemma}
	\label{statement:degree-n-c-travel-plan}
    Let $c$ be an integer and \(D\) be a temporal digraph
	on \(n\) vertices
	with
	\(\MindegComplement{D} \leq c\).
	If \(\Lifetime{D} \geq \bound{statement:degree-n-c-travel-plan}{t}{n,c} ∈ \Theta(n)\),
	then
	 \(D\) admits a travel plan \(\mathcal{T} = (T_i)_{i ∈ [k]}\) of gap 1,
	 order \(k \geq \bound{statement:degree-n-c-travel-plan}{k}{n,c} ∈ \Theta(n/c)\) and
	width at most \( \bound{statement:degree-n-c-travel-plan}{w}{c} ∈ \Theta(c) \).
	Additionally, 
	for all \(i,j ∈ [k]\),
	\(T_i ∩ T_j = \emptyset\) if \(i \neq j\).
\end{lemma}
\begin{proof}
	Let \(t_1 = \bound{statement:degree-n-c-exact-princesses}{t}{c}\),
	\(t_2 = t_1 + 1\),
	\(w = 5c + 1\).
	We iteratively	construct sets \(V_i, A_i, B_i, T_i\) and temporal digraphs \(H_i\) satisfying the following properties (see \cref{fig:travel-plan-proof-construction})
	\begin{enumerate}
		\item \label{item:degree-n-c-travel-plan:A_i}
			\(\Abs{A_i} \geq 2c + 1\) if \(\Abs{V_i} \geq 2c + 1\),
		\item \label{item:degree-n-c-travel-plan:M_i}
			\(V_{i+1} ⊆ \TemporalReach[H_i]{t_2}{v}\) for every \(v ∈ A_i\),
		\item \label{item:degree-n-c-travel-plan:N_i}
			\(\Abs{V_{i + 1} ∖ \TemporalReach[H_i]{t_2}{v}} \leq 2c\) for every \(v ∈ B_i\),
		\item \label{item:degree-n-c-travel-plan:w}
			\(T_i = A_i ∪ B_i\) and \(\Abs{T_i} \leq w\), and
		\item\label{item:degree-n-c-travel-plan:V_i+1}
			\(V_{i + 1} = V_i ∖ T_i\).
	\end{enumerate}
	Start by setting \(V_1 = \V{D}\).

	On step \(i\),
	if \(\Abs{V_i} \leq 2c + 1\), set \(T_i = A_i = B_i = V_i\), \(V_{i + 1} = \emptyset\) and stop the construction. 
    Otherwise, 
	let \(H_i = \TemporalSlice{D}{(i-1)t_2 + 1}{i t_2}\),
	\(G_i^1 = \TemporalSlice{H_i}{1}{t_1}\) and
	\(G_i^2 = \TemporalSlice{H_i}{t_1 + 1}{t_1 + 1}\).
	By \cref{statement:degree-n-c-exact-princesses},
	there is some set \(A_i ⊆ V_{i}\)
	such that
	\(\Abs{A_i} = 2c + 1\),
	\(V_i \subseteq \TemporalReach[G_i^1]{t_1}{A_i}\) and,
	for every \(u ∈ A_i\),
	\(\Abs{V_{i} \setminus (\TemporalReach[G_i^1]{t_1}{u} ∪ A_i)} ≤ c\).
	Let \(M_i = \bigcap_{u ∈ A_i} \TemporalReach[G_i^1]{t_1}{u} ∖ A_i\),
	\(X_i = V_{i} \setminus (M_i ∪ A_i)\),
	\(B_i = \{u ∈ X_i \mid \Indeg[G_i^2]{u} \leq c\}\) and
	\(Y_i = X_i ∖ B_i\).

	It is immediate that \cref{item:degree-n-c-travel-plan:A_i} holds. 
	We set \(T_i = A_i ∪ B_i\) and
	\(V_{i + 1} = V_i ∖ T_i = M_i ∪ Y_i\),
	satisfying~\cref{item:degree-n-c-travel-plan:V_i+1}.
	By \cref{statement:degree-n-c-small-hard-to-reach-set},
	\(\Abs{B_i} \leq 3c\) and
	\(\Abs{V_i ∖ \OutN[G_i^2]{u}} \leq 2c\) for every \(u ∈ B_i\).
	This satisfies \cref{item:degree-n-c-travel-plan:N_i} and \cref{item:degree-n-c-travel-plan:w}.

	Towards proving \cref{item:degree-n-c-travel-plan:M_i},
	let \(u ∈ Y_i\) and \(v ∈ A_i\).
	By definition, \(\Indeg[G_i^2]{u} \geq c + 1\).
	As \(\Abs{V_{i} \setminus (\TemporalReach[G_i^1]{t_1}{v} ∪ A_i)} \leq c\),
	the set \(\InN[G_i^2]{u} ∩ \TemporalReach[G_i^1]{t_1}{v}\) is not empty.
	Hence, \(u ∈ \TemporalReach[H_i]{t_2}{v}\), and so
	\(Y_i ⊆ \TemporalReach[H_i]{t_2}{v}\) for every \(v ∈ A_i\).
	As \(V_{i+1} = M_i ∪ Y_i\),
	\cref{item:degree-n-c-travel-plan:M_i} holds.
	This completes the construction.

	Let \(k\) be the number of iterations, that is,
	\(T_{1}, T_{2}, \ldots, T_{k}\) are constructed above.
	As \cref{item:degree-n-c-travel-plan:w} and \cref{item:degree-n-c-travel-plan:V_i+1} hold and
	the iteration above stops when \(\Abs{V_i} \leq 2c + 1\),
	we have	\(k \geq \Ceiling{n / w}\).
	Since \cref{item:degree-n-c-travel-plan:A_i} and \cref{item:degree-n-c-travel-plan:V_i+1} hold,
	\(k \leq n / (2c + 1)\).
	Hence, the temporal digraphs \(H_i\) are always defined
	since \(\Lifetime{D}\) is large enough.
    
	We now prove that \((T_{1}, T_{2}, \ldots, T_{k})\) is a travel plan of gap 1 for \(D\).
	Let \(v_i ∈ T_i\) and \(v_j ∈ T_j\)
	such that
	\(i < j-1\).
	We construct a temporal \(v_i\)-\(v_j\)-walk \(W_i\) in \(H_i · H_{i + 1}\) (see \cref{fig:travel-plan-proof-paths}).

	By \cref{item:degree-n-c-travel-plan:V_i+1}, \(v_j ∈ V_{i + 2} ⊆ V_{i + 1} ⊆ V_i\).
	Moreover, \(\Abs{A_{i+1}} \geq 2c + 1\), as \(A_{i + 2}\) exists and so \(\Abs{V_{i+1}} > 2c + 1\).
	If \(v_i ∈ A_i\),
	then by \cref{item:degree-n-c-travel-plan:M_i} there is a temporal \(v_i\)-\(v_j\)-walk \(W_i\)
	in \(H_i\).

	If \(v_i ∈ B_i\),
	then by \cref{item:degree-n-c-travel-plan:A_i} and \cref{item:degree-n-c-travel-plan:N_i} there is some \(w ∈ A_{i + 1}\)
	such that
	a temporal \(v_i\)-\(w\)-path \(W_i^1\) exists in \(H_i\).
	By \cref{item:degree-n-c-travel-plan:M_i},
	there is a temporal \(w\)-\(v_j\)-path \(W_i^2\) in \(H_{i+1}\).
	Hence, \(W_i = W_i^1 · W_i^2\) is a temporal \(v_i\)-\(v_j\)-walk in \(H_i · H_{i + 1}\).

	Given a sequence \(v_{i_1}, v_{i_2}, \ldots, v_{i_s}\) of vertices where
	\(v_{i_j} ∈ T_{i_j}\) and \(i_j + 1 < i_{j + 1}\) for all \(j ∈ [s]\),
	the walk \(W = W_{i_1} · W_{i_2} · \ldots · W_{i_{s-1}}\)
	is a temporal walk in \(D\)
	visiting all \(v_{i_j}\)'s in the given order.
	Hence, \((T_{1}, T_{2}, \ldots, T_{k})\) is a travel plan for \(D\), as desired.
\end{proof}

\subsection{Tight bounds}

\label{sec:c2n upper bound}

While a greedy strategy is suitable for constructing an exploration using many travel plans of our temporal digraph,
it would have the same issue as the proof of \cref{statement:degree-n-c-upper-bound-strong-log}:
if we make poor choices in the beginning, we can only benefit from smaller and smaller parts of the travel plan.

Intuitively, the travel plan allows us to randomly picking vertices from non-consecutive sets $T_i$
and always obtain a walk visiting all our choices.
However, the situation where a walk \(W\) visiting \(n/c\) vertices
while completely covering many sets \(T_i\) of the next travel plan looks very far from random,
yet is precisely the ``bad'' situation mentioned before.

If we have many travel plans in our temporal digraph
and we randomly assign each vertex to exactly one travel plan,
being in a ``bad'' situation means assigning two vertices to the same set \(T_i\) or two consecutive sets $T_i$.
Hence, as the number of travel plans increases, the probability we have a ``bad'' situation decreases, and the travel plans provides us with temporal walks covering the vertices assigned to them if no ``bad'' situation happens.
Thus, for our purposes, we only need to argue that the probability our random choices are ``bad'' is less than 1.

In order to prove the intuition above, we make probabilistic arguments;
in particular, we use the following celebrated symmetric version of the Lov\'asz Local Lemma, due to Erd\H{o}s and Lov\'asz~\cite{erdosIFS10}.

\begin{lemma}[{\sc Lov\'asz Local Lemma}~\cite{erdosIFS10}]
\label{lemma:LLL}
    Let $A_1,A_2,\dots,A_n$ be events in an arbitrary probability space. Suppose that each event $A_i$ is mutually independent of a set of all the other events but at most $d$, and that $\PP(A_i) \leq p$ for all $1\leq i \leq n$. If 
    $4pd \leq 1$ then 
		the probability that none of the events \(A_i\) occur is greater than 0.
		\end{lemma}

We model the restrictions on the choices of \(v_i ∈ T_i\) respecting the gap of the travel plan
as an undirected graph, with an edge between two vertices \(u,v\) if they cannot be chosen simultaneously,
either because they lie on the same set or because the gap between them is not large enough.
As our travel plans have small width and gap 1, this auxiliary graph has bounded degree.
So the question of finding a good choice of vertices is essentially
a variation of finding independent sets on graphs of bounded degree,
a classical application of the probabilistic method.

\begin{lemma}
	\label{statement:bounded degree independent set cover}
	Let \(G_{1}, G_{2}, \ldots, G_{\ell}\) be undirected graphs on the same vertex set \(V\)
	such that
	\(\MaxDegree{G_i} \leq \Delta\) for every $i\in [\ell]$.
	If \(\ell \geq \bound{statement:bounded degree independent set cover}{\ell}{\Delta} \coloneqq 16 \Delta + 1\), then there are \(X_{1}, X_{2}, \ldots, X_{\ell} ⊆ V\)
	such that
	\(V = \bigcup_{i = 1}^\ell X_i\) and
	\(X_i\) is an independent set in \(G_i\)
	for every \(i ∈ [\ell]\).
\end{lemma}
\begin{proof}
	Randomly partition \(\V{D}\) into \(X_{1}, X_{2}, \ldots, X_{\ell}\)
	by assigning a vertex \(v\) to \(X_i\) with probability \(1 / \ell\).

	For each \(i ∈ [\ell]\) and each edge in \(G_i\),
	define the random event \(A^i_e : \) both endpoints of \(e\) are in \(X_i\).
	Clearly, \(\PP(A^i_e) = p \coloneqq \ell^{-2}\).

	Each event \(A^i_e\) is mutually independent from all other events \(A^j_f\)
	where \(e ∩ f = \emptyset\).
	If \(i = j\), there are at most \(2(\Delta -1)\) edges \(f\neq e\) with \(e ∩ f \neq \emptyset\).
	If \(i \neq j\), then there are at most \(2(\Delta - 1) + 1\) such edges.
	Thus, \(A^i_e\) is mutually independent from all but at most
	\(d \coloneqq 2(\Delta - 1) + (2(\Delta - 1) + 1)(\ell - 1) \leq 4 \Delta \ell\) other events.
	
	All that remains in order to apply Lovász Local Lemma
	is to show the following inequality.
	\begin{align*}
		& 4dp \leq 1 
		\Leftarrow 4 \cdot 4 \Delta \ell \cdot (\ell^{-2}) \leq 1
		\Leftrightarrow 16 \Delta \leq \ell.
	\end{align*}
    By \cref{lemma:LLL}, the probability that none of the events \(A^i_e\) defined above occur is greater than 0.
	Hence, there is a choice of \(X_{1}, X_{2}, \ldots, X_{\ell} ⊆ \V{G}\) satisfying the conditions in the statement.
\end{proof}

Using \cref{statement:degree-n-c-travel-plan} we obtain many travel plans and
with \cref{statement:bounded degree independent set cover} we can find good walks covering all vertices.
Assuming strongly-connected snapshots, we can find paths connecting the walks given by the travel plans,
yielding our main result. Define
\begin{align*}
	\bound{statement:degree-n-c-upper-bound-strong}{t}{n,c} & =
	\bound{statement:bounded degree independent set cover}{\ell}{3 \bound{statement:degree-n-c-travel-plan}{w}{c} - 1} ·
	(\bound{statement:degree-n-c-travel-plan}{t}{n,c} + n - 1).
\end{align*}

\hypertarget{statement:degree-n-c-upper-bound-strong:body}{}
\largeMinimumDegree*
\begin{proof}
	Let \(\ell_1 = \bound{statement:degree-n-c-travel-plan}{t}{n,c}\),
    \(\ell_2 = n - 1\),
    \(\ell_3 = \ell_1 + \ell_2\),
	\(w = \bound{statement:degree-n-c-travel-plan}{w}{c}\),
	\(\Delta = 3w - 1\),
    \(p = \bound{statement:bounded degree independent set cover}{\ell}{\Delta}\) and
    \(t_i = 1 + (i - 1) \ell_3\).
    
    For each \(i \in [p]\),
    let \(H_i \coloneqq\TemporalSlice{D}{t_i}{t_i + \ell_1 -1}\) and
    let \(F_i = \TemporalSlice{D}{t_i + \ell_1}{t_i + \ell_1 + \ell_2 - 1}\).
	By \cref{statement:degree-n-c-travel-plan},
    each \(H_i\) contains
	a travel plan \(\mathcal{T}_i\) of gap 1, order \(k_i \geq n/w\)
	and width at most \(w\).
	
	Define the auxiliary undirected graph \(G_i\) as follows.
	Let \(\mathcal{T}_i = (T_{1}, T_{2}, \ldots, T_{k_i})\).
	Set \(\V{G_i} = \V{D}\) and add an edge \(\{u,v\}\) to \(G_i\)
	if \(u ∈ T_a, v ∈ T_b\) and \(\Abs{a - b} \leq 1\).
	Note that the degree of each vertex \(v\) is at most \(\Delta\).
	
	Applying \cref{statement:bounded degree independent set cover} to the graphs \(G_{1}, G_{2}, \ldots, G_{p}\),
	we obtain a partition of \(\V{D}\) into sets \(X_{1}, X_{2}, \ldots, X_{p}\)
	such that
	\(X_i\) is an independent set in \(G_i\).
	By definition of travel plans, there is a temporal walk \(R_i\) visiting \(X_i\) inside \(H_i\).
	Let \(a_i\) be the first vertex of \(R_i\) and let \(b_i\) be its last vertex.
	By \cref{statement:strong-temporal-u-v-path}, there is a temporal \(b_i\)-\(a_{i+1}\)-path \(P_i\) in \(F_i\)
	for each \(i ∈ [p - 1]\).
	Set \(W_i = R_i \cdot P_i\) for $i\in [p-1]$ and
	\(W = W_1 \cdot W_2 \cdot \ldots \cdot W_{p-1}\cdot R_p\).
	It is immediate that \(W\) is a temporal exploration of \(D\), as desired.
\end{proof}

\begin{figure}
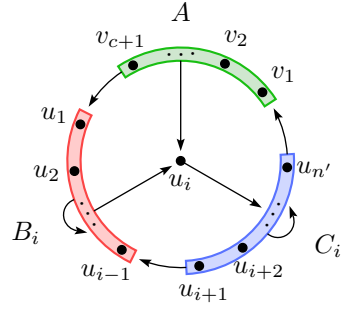

	\centering

	\caption{\label{fig:degree-n-c-lower-bound}
	The digraph \(D_i\) constructed in the proof of \cref{statement:degree-n-c-lower-bound}.
	An arc between two sets or between a vertex and a set indicates that we add all possible arcs (except loops)
	in the corresponding direction.
	In particular, \(A\) is an independent set, while \(B_i\) and \(C_i\) are bidirected cliques.
	}
\end{figure}

The following Theorem establishes the asymptotic tightness of the function in \cref{statement:degree-n-c-upper-bound-strong}.

\largeMinimumDegreeLowerBound*
\begin{proof}
	Let \(n' = n - c - 1\), \(A = \Set{v_{1}, v_{2}, \ldots, v_{c + 1}}\), and \(V=\{u_1,\dots,u_{n'}\} \cup A\). 
	For each \(i ∈ [n' - 1]\), define
	\(B_i = \Set{u_{1}, u_{2}, \ldots, u_{i - 1}}\),
	\(C_i = \Set{u_{i+1}, u_{i+2}, \ldots, u_{n'}}\) and
	\(D_i = (V, E_i)\), where (see \cref{fig:degree-n-c-lower-bound})
	\begin{align*}
		E_i  =~ & (A \times B_i) ∪ (A \times \{u_i\}) ∪ (C_i \times A) ∪ (C_i \times B_i) ∪ (B_i \times \Set{u_i}) ∪ (\Set{u_i} \times C_i)
		\\[0em]     & ∪ \Set{(u_a, u_b) \mid u_a, u_b ∈ B_i, a \neq b} ∪ \Set{(u_a, u_b) \mid u_a, u_b ∈ C_i, a \neq b}.
	\end{align*}
	We define \(D' = (D_i)_{i ∈ [n' - 1]}\) and
	\(D\) as the concatenation of \(c\) copies of \(D'\).
	Observe that every \(v ∈ B_i ∪ C_i\) has degree \(n - 1\) and
	every \(v ∈ A\) has degree \(n - c\) in every \(D_i\).
	Hence, \(\MindegComplement{D_i} \leq c\).
    As \(u_{n'} \in C_i\), we have \(A_i \cup B_i \subseteq \OutN[D_i]{u_{n'}}\).
    Every vertex in \(A_i \cup B_i\) can reach \(u_{n'}\)
    through \(u_i \neq u_{n'}\),
    and every vertex in \(C_i\) can reach \(u_{n'}\).
    Hence, every \(D_i\) is strongly connected.

	Let \(v_j,v_k ∈ A\) be distinct.
	We prove that there is no temporal \(v_j\)-\(v_k\)-walk in \(D'\).
	Note that \(\TemporalReach[D']{1}{v_j} = \Set{v_j, u_1}\).
	For \(1 < i \leq n'\), note that \(\TemporalReach[D']{i - 1}{v_j} = \{v_j\} ∪ B_{i-1}\). 
	Since \(u_{i}\) is the only outneighbor in $D_i$ of the vertices in \(\TemporalReach[D']{i-1}{v_j}\),
	we have that \(\TemporalReach[D']{i}{v_j} = \{v_j\} ∪ B_{i-1} ∪ \{u_{i}\}\).
	Hence, \(v_k \notin \TemporalReach[D']{n'}{v_j}\), as desired.

	Assume towards a contradiction that \(D\) admits a temporal exploration \(W\).
	For each \(1 \leq j \leq c + 1\),
	let \(t_j\) be the time step of
	the first occurrence of \(v_j\) along \(W\).

	As \(D\) has only \(cn'\) snapshots,
	by the pigeon-hole principle 
	there are \(t_j < t_k\)
	such that
	\(t_j,t_k ∈ \Set{(i - 1)n' + 1, \ldots, in'}\)
	for some \(1 \leq i \leq c\).
	However, this implies that
	the temporal \(v_j\)-\(v_k\)-path from \(t_j\) to \(t_k\)
	corresponds to a temporal \(v_j\)-\(v_k\)-path in \(D'\),
	a contradiction.
	Hence, \(D\) does not admit a temporal exploration.
\end{proof}

\section{Future work}
\label{sec:future work}

We collect here some interesting questions for future work.
We state them as conjectures to simplify notation,
although we do not necessarily believe them to be true.

\Cref{statement:degree-three-lower-bound} excludes the possibility of exploring always-strong temporal digraphs of maximum degree 3 and subquadratic lifetime,
yet the case of maximum degree 2 remains unsolved.
\begin{conjecture}
	\label{statement:degree-two-upper-bound}
	There is a function \(f\colon \Naturals \to \Naturals\), \(f(n) ∈ \BigO{n \log n} \)
	such that
	every temporal digraph \(D = (D_i)_{i \in [t]}\),
    where every \(D_i\) is a directed cycle,
	and \(\Lifetime{D} \geq f(\Abs{\V{D}})\)
	admits a temporal exploration.
\end{conjecture}

Note that $f(n) \in \Omega(n \log n)$. This can be shown by an easy adaptation of the proof of~\cite[Proposition~4.1]{EHK2021}, consisting in adding a path on $n/2$ vertices linking $t_{n/4-1}$ and $b_{n/4-1}$.
This is open even if the underlying undirected graph of \(D\) has maximum degree 3.

Most of the steps in the proof of \cref{statement:degree-n-c-upper-bound-strong}
also work on always-unilateral temporal digraphs.
The only issue is connecting the endpoints of the intermediate walks obtained
from each travel plan.
However, given an always-unilateral digraph \(D\) of \(\Lifetime{D} \geq n - 1\), consider the auxiliary digraph \(H\) with same vertex set as \(D\) and
with an arc \((u,v)\)	if \(u\) can temporally reach \(v\) in \(D\).
If \(H\) is a acyclic, then there is a path \(P\) spanning all vertices of \(D\)
which is contained as a subgraph in every snapshot of \(D\),
as every snapshot is unilateral.
Hence, taking one arc of \(P\) in each snapshots yields a temporal exploration of \(D\).
This observation implies that both extreme cases, namely,
when the reachability digraph is strongly connected and when it is acyclic,
are ``good''.
Intuitively, this should be exploitable in order to obtain a temporal exploration for always-unilateral temporal digraphs of subquadratic lifetime and
large minimum degree.

\begin{conjecture}
	\label{statement:degree-n-c-unilateral-upper-bound}
	There is a function \(f : \Naturals \to \Naturals\) and
	a constant \(0 < \epsilon \leq 1\)
	such that every always-unilateral temporal digraph \(D\)
	with 
	\(\MindegComplement{D} \leq c\) and
	\(\Lifetime{D} \geq f(c)n^{2 - \epsilon}\)
	admits a temporal exploration.
\end{conjecture}

It would be interesting to know if \largeminimumdegreeupperboundref{}
can be generalized to parameters smaller than \(\MindegComplement{D}\).
One such parameter is the maximum size of an independent set in
any snapshot of \(D\).

\begin{conjecture}
	\label{statement:small-independent-set-number}
	There is a function \(f : \Naturals \to \Naturals\)
	and some \(0 < \epsilon \leq 1\)
	such that
	every always-strong temporal digraph \(D = (D_i)_{i ∈ [t]}\)
	where no \(\UnderlyingGraph{D_i}\) contains an independent set of size more than \(k\)
	admits a temporal exploration 
	if \(t \geq f(k)n^{2 - \epsilon}\).
\end{conjecture}

The definition of travel plan can be transferred to the undirected setting in the natural way.
We are not aware of any similar auxiliary structure being used for proving the existence of temporal exploration in temporal (undirected) graphs.
Hence, in a broader direction, it would be interesting to determine when temporal graphs admit
travel plans of small width and gap.

\bibliography{references}

@article{BGMR2025ietg,
  author       = {Paul Bastide and
                  Carla Groenland and
                  Lukas Michel and
                  Cl{\'{e}}ment Rambaud},
  title        = {Improved exploration of temporal graphs},
  journal      = {CoRR},
  volume       = {abs/2511.22604},
  year         = {2025},
  eprinttype   = {arXiv},
  eprint       = {2511.22604},
  bibsource    = {dblp computer science bibliography, https://dblp.org}
}

@book{BG2009,
  author       = {J{\o}rgen Bang{-}Jensen and
                  Gregory Z. Gutin},
  title        = {Digraphs - Theory, Algorithms and Applications, Second Edition},
  series       = {Springer Monographs in Mathematics},
  publisher    = {Springer},
  year         = {2009},
  isbn         = {978-1-84800-997-4},
  bibsource    = {dblp computer science bibliography, https://dblp.org}
}

@article{Berman2004,
  title={Approximation Hardness of Short Symmetric Instances of {M}{A}{X}-3{S}{A}{T}},
  author={Berman, Piotr and Karpinski, Marek and Scott, Alexander},
  journal={Electronic Colloquium on Computational Complexity},
  year={2003},
  volume={TR03},
  url={https://api.semanticscholar.org/CorpusID:40181844}
}

@article{Bangjensen2012,
  title={Finding an induced subdivision of a digraph},
  author={Bang-Jensen, J{\o}rgen and Havet, Fr{\'e}d{\'e}ric and Trotignon, Nicolas},
  journal={Theoretical Computer Science},
  volume={443},
  pages={10--24},
  year={2012},
  publisher={Elsevier}
}

@article{Bangjensen2012arcdisjoint,
  title={Arc-disjoint spanning sub (di) graphs in digraphs},
  author={Bang-Jensen, J{\o}rgen and Yeo, Anders},
  journal={Theoretical Computer Science},
  volume={438},
  pages={48--54},
  year={2012},
  publisher={Elsevier}
}

@article{Bangjensen2024,
  title={Constrained flows in networks},
  author={Bang-Jensen, J{\o}rgen and Bessy, St{\'e}phane and Picasarri-Arrieta, Lucas},
  journal={Theoretical Computer Science},
  volume={1010},
  pages={114702},
  year={2024},
  publisher={Elsevier}
}

@article{Bangjensen2016,
  title={Finding good 2-partitions of digraphs II. Enumerable properties},
  author={Bang-Jensen, J{\o}rgen and Cohen, Nathann and Havet, Fr{\'e}d{\'e}ric},
  journal={Theoretical Computer Science},
  volume={640},
  pages={1--19},
  year={2016},
  publisher={Elsevier}
}

@article{Rede34,
  author = 	 {R\'{e}dei, L\'azl\'o},
  title = 	 {Ein kombinatorischer {S}atz},
  journal = 	 {Acta Litteraria Szeged},
  year = 	 1934,
  volume =	 7,
  pages =	 {39--43}
}

@book{Cranston2024,
  title     = {Graph Coloring Methods},
  author    = {Cranston, Daniel W.},
  year      = {2024},
  publisher = {Self-published},
  url       = {https://graphcoloringmethods.com/},
}

@article{ES2022,
  author       = {Thomas Erlebach and
                  Jakob T. Spooner},
  title        = {Exploration of k-edge-deficient temporal graphs},
  journal      = {Acta Informatica},
  volume       = {59},
  number       = {4},
  pages        = {387--407},
  year         = {2022},
  doi          = {10.1007/S00236-022-00421-5},
  bibsource    = {dblp computer science bibliography, https://dblp.org}
}

@article{EHK2021,
  author       = {Thomas Erlebach and
                  Michael Hoffmann and
                  Frank Kammer},
  title        = {On temporal graph exploration},
  journal      = {J. Comput. Syst. Sci.},
  volume       = {119},
  pages        = {1--18},
  year         = {2021},
  doi          = {10.1016/J.JCSS.2021.01.005},
  bibsource    = {dblp computer science bibliography, https://dblp.org}
}

@article{erdosIFS10,
  title={Problems and results on 3-chromatic hypergraphs and some related questions},
  author={Erd\H{o}s, Paul and Lov{\'a}sz, L{\'a}szl{\'o}},
  journal={Infinite and Finite Sets},
  volume={10},
  number={2},
  pages={609--627},
  year={1975},
  publisher={North-Holland, Amsterdam}
}

@article{michailTCS634,
  title={Traveling salesman problems in temporal graphs},
  author={Michail, Othon and Spirakis, Paul G},
  journal={Theoretical Computer Science},
  volume={634},
  pages={1--23},
  year={2016},
  publisher={Elsevier}
}

@article{havetJGT35,
  title={Median orders of tournaments: a tool for the second neighborhood problem and {S}umner's conjecture},
  author={Havet, Fr{\'e}d{\'e}ric and Thomass{\'e}, St{\'e}phan},
  journal={Journal of Graph Theory},
  volume={35},
  number={4},
  pages={244--256},
  year={2000},
  publisher={Wiley Online Library}
}

@article{changJCO25,
  title={The hamiltonian numbers in digraphs},
  author={Chang, Ting-Pang and Tong, Li-Da},
  journal={Journal of Combinatorial Optimization},
  volume={25},
  number={4},
  pages={694--701},
  year={2013},
  publisher={Springer}
}

@inproceedings{hkmm24cows,
  author       = {Meike Hatzel and
                  Stephan Kreutzer and
                  Marcelo Garlet Milani and
                  Irene Muzi},
  title        = {Cycles of Well-Linked Sets and an Elementary Bound for the Directed Grid Theorem},
  booktitle    = {65th {IEEE} Annual Symposium on Foundations of Computer Science, {FOCS} 2024, Chicago, IL, USA, October 27-30, 2024},
  pages        = {1--20},
  publisher    = {{IEEE}},
  year         = {2024},
  doi          = {10.1109/FOCS61266.2024.00011},
  bibsource    = {dblp computer science bibliography, https://dblp.org}
}

@article{hkmm26cowsi,
  author       = {Meike Hatzel and
                  Stephan Kreutzer and
                  Marcelo Garlet Milani and
                  Irene Muzi},
  title        = {Cycles of Well-Linked Sets I: an Elementary Bound for Directed Cycle Packing},
  journal      = {CoRR},
  volume       = {abs/2404.19222v3},
  year         = {2026},
  doi          = {10.48550/ARXIV.2404.19222},
  eprinttype    = {arXiv},
  eprint       = {2404.19222},
  bibsource    = {dblp computer science bibliography, https://dblp.org}
}

@inproceedings{ErlebachS18,
  author       = {Thomas Erlebach and
                  Jakob T. Spooner},
  editor       = {Igor Potapov and
                  Paul G. Spirakis and
                  James Worrell},
  title        = {Faster Exploration of Degree-Bounded Temporal Graphs},
  booktitle    = {43rd International Symposium on Mathematical Foundations of Computer
                  Science, {MFCS} 2018, Liverpool, UK, August 27-31, 2018},
  series       = {LIPIcs},
  volume       = {117},
  pages        = {36:1--36:13},
  publisher    = {Schloss Dagstuhl - Leibniz-Zentrum f{\"{u}}r Informatik},
  year         = {2018},
  doi          = {10.4230/LIPICS.MFCS.2018.36},
  bibsource    = {dblp computer science bibliography, https://dblp.org}
}

@inproceedings{ErlebachKLSS19,
  author       = {Thomas Erlebach and
                  Frank Kammer and
                  Kelin Luo and
                  Andrej Sajenko and
                  Jakob T. Spooner},
  editor       = {Christel Baier and
                  Ioannis Chatzigiannakis and
                  Paola Flocchini and
                  Stefano Leonardi},
  title        = {Two Moves per Time Step Make a Difference},
  booktitle    = {46th International Colloquium on Automata, Languages, and Programming,
                  {ICALP} 2019, Patras, Greece, July 9-12, 2019},
  series       = {LIPIcs},
  volume       = {132},
  pages        = {141:1--141:14},
  publisher    = {Schloss Dagstuhl - Leibniz-Zentrum f{\"{u}}r Informatik},
  year         = {2019},
  doi          = {10.4230/LIPICS.ICALP.2019.141},
  bibsource    = {dblp computer science bibliography, https://dblp.org}
}

@inproceedings{AdamsonGMZ22,
  author       = {Duncan Adamson and
                  Vladimir V. Gusev and
                  Dmitriy S. Malyshev and
                  Viktor Zamaraev},
  editor       = {James Aspnes and
                  Othon Michail},
  title        = {Faster Exploration of Some Temporal Graphs},
  booktitle    = {1st Symposium on Algorithmic Foundations of Dynamic Networks, {SAND}
                  2022, Virtual Conference, March 28-30, 2022},
  series       = {LIPIcs},
  volume       = {221},
  pages        = {5:1--5:10},
  publisher    = {Schloss Dagstuhl - Leibniz-Zentrum f{\"{u}}r Informatik},
  year         = {2022},
  doi          = {10.4230/LIPICS.SAND.2022.5},
  bibsource    = {dblp computer science bibliography, https://dblp.org}
}

@article{casteigts2012time,
  title={Time-varying graphs and dynamic networks},
  author={Casteigts, Arnaud and Flocchini, Paola and Quattrociocchi, Walter and Santoro, Nicola},
  journal={International Journal of Parallel, Emergent and Distributed Systems},
  volume={27},
  number={5},
  pages={387--408},
  year={2012},
  publisher={Taylor \& Francis}
}

@article{michail2016introduction,
  title={An introduction to temporal graphs: An algorithmic perspective},
  author={Michail, Othon},
  journal={Internet Mathematics},
  volume={12},
  number={4},
  pages={239--280},
  year={2016},
  publisher={Taylor \& Francis}
}

@article{bumpus2023edge,
  title={Edge exploration of temporal graphs},
  author={Bumpus, Benjamin Merlin and Meeks, Kitty},
  journal={Algorithmica},
  volume={85},
  number={3},
  pages={688--716},
  year={2023},
  publisher={Springer}
}

@article{fluschnik2020temporal,
  title={Temporal graph classes: A view through temporal separators},
  author={Fluschnik, Till and Molter, Hendrik and Niedermeier, Rolf and Renken, Malte and Zschoche, Philipp},
  journal={Theoretical Computer Science},
  volume={806},
  pages={197--218},
  year={2020},
  publisher={Elsevier}
}

@inproceedings{hamm2022complexity,
  title={The complexity of temporal vertex cover in small-degree graphs},
  author={Hamm, Thekla and Klobas, Nina and Mertzios, George B and Spirakis, Paul G},
  booktitle={Proceedings of the AAAI Conference on Artificial Intelligence},
  volume={36(9)},
  pages={10193--10201},
  year={2022}
}

@inproceedings{marino2021konigsberg,
  title={K{\"o}nigsberg sightseeing: Eulerian walks in temporal graphs},
  author={Marino, Andrea and Silva, Ana},
  booktitle={International Workshop on Combinatorial Algorithms},
  pages={485--500},
  year={2021},
  organization={Springer}
}

@article{marino2023eulerian,
  title={Eulerian walks in temporal graphs},
  author={Marino, Andrea and Silva, Ana},
  journal={Algorithmica},
  volume={85},
  number={3},
  pages={805--830},
  year={2023},
  publisher={Springer}
}

@inproceedings{aaron2014dmvp,
  title={DMVP: foremost waypoint coverage of time-varying graphs},
  author={Aaron, Eric and Krizanc, Danny and Meyerson, Elliot},
  booktitle={International Workshop on Graph-Theoretic Concepts in Computer Science},
  pages={29--41},
  year={2014},
  organization={Springer}
}

@article{akrida2021temporal,
  title={The temporal explorer who returns to the base},
  author={Akrida, Eleni C and Mertzios, George B and Spirakis, Paul G and Raptopoulos, Christoforos},
  journal={Journal of Computer and System Sciences},
  volume={120},
  pages={179--193},
  year={2021},
  publisher={Elsevier}
}

@inproceedings{erlebach2020non,
  title={Non-strict temporal exploration},
  author={Erlebach, Thomas and Spooner, Jakob T},
  booktitle={International Colloquium on Structural Information and Communication Complexity},
  pages={129--145},
  year={2020},
  organization={Springer}
}

@inproceedings{balev2025brief,
  title={Brief announcement: the shortest temporal exploration problem},
  author={Balev, Stefan and Sanlaville, {\'E}ric and Toullalan, Antoine},
  booktitle={4th Symposium on Algorithmic Foundations of Dynamic Networks (SAND 2025)},
  pages={18--1},
  year={2025},
  organization={Schloss Dagstuhl--Leibniz-Zentrum f{\"u}r Informatik}
}

@article{ilcinkas2018exploration,
  title={Exploration of the t-interval-connected dynamic graphs: the case of the ring},
  author={Ilcinkas, David and Wade, Ahmed M},
  journal={Theory of Computing Systems},
  volume={62},
  number={5},
  pages={1144--1160},
  year={2018},
  publisher={Springer}
}

@inproceedings{FMNRZ20,
  author       = {Till Fluschnik and
                  Hendrik Molter and
                  Rolf Niedermeier and
                  Malte Renken and
                  Philipp Zschoche},
  editor       = {Fedor V. Fomin and
                  Stefan Kratsch and
                  Erik Jan van Leeuwen},
  title        = {As Time Goes By: Reflections on Treewidth for Temporal Graphs},
  booktitle    = {Treewidth, Kernels, and Algorithms - Essays Dedicated to Hans L. Bodlaender
                  on the Occasion of His 60th Birthday},
  series       = {Lecture Notes in Computer Science},
  volume       = {12160},
  pages        = {49--77},
  publisher    = {Springer},
  year         = {2020},
  doi          = {10.1007/978-3-030-42071-0\_6},
  bibsource    = {dblp computer science bibliography, https://dblp.org}
}

@article{HCGD23,
  author       = {Mohammad Mehdi Hosseinzadeh and
                  Mario Cannataro and
                  Pietro Hiram Guzzi and
                  Riccardo Dondi},
  title        = {Temporal networks in biology and medicine: a survey on models, algorithms,
                  and tools},
  journal      = {Netw. Model. Anal. Health Informatics Bioinform.},
  volume       = {12},
  number       = {1},
  pages        = {10},
  year         = {2023},
  doi          = {10.1007/S13721-022-00406-X},
  bibsource    = {dblp computer science bibliography, https://dblp.org}
}

@book{moon1968topics,
  author    = {John W. Moon},
  title     = {Topics on Tournaments},
  publisher = {Holt, Rinehart and Winston},
  address   = {New York},
  year      = {1968}
}

@article{NSSJCTB173,
  title={Some results and problems on tournament structure},
  author={Nguyen, Tung and Scott, Alex and Seymour, Paul},
  journal={Journal of Combinatorial Theory, Series B},
  volume={173},
  pages={146--183},
  year={2025},
  publisher={Elsevier}
}

@article{APSComb21,
  title={Ramsey-type theorems with forbidden subgraphs},
  author={Alon, Noga and Pach, J{\'a}nos and Solymosi, J{\'o}zsef},
  journal={Combinatorica},
  volume={21},
  number={2},
  pages={155--170},
  year={2001},
  publisher={Springer}
}

\end{document}